\documentclass[table,11pt, letter]{article}
  
\newif\ifdraft
\draftfalse
\newif\ifEA
\EAfalse

\usepackage[margin=1in,letterpaper]{geometry} 
\usepackage{setspace}
\usepackage{xspace}
\usepackage{xfrac}

\usepackage{tikz,subfigure}
\usepackage{epsfig}
\usepackage{float}
\usepackage{placeins}

\usepackage{multirow}

\usepackage{enumerate}

\usepackage[verbose,colorlinks]{hyperref}
\hypersetup{%
    colorlinks=true, linktocpage=true, pdfstartpage=1, pdfstartview=FitV,%
    breaklinks=true, pdfpagemode=UseNone, pageanchor=true, pdfpagemode=UseOutlines,%
    plainpages=false, bookmarksnumbered, bookmarksopen=true, bookmarksopenlevel=1,%
    hypertexnames=true, pdfhighlight=/O,
    urlcolor=darkred, linkcolor=lightblue, citecolor=darkgreen, 
    pdftitle={},%
    pdfauthor={},%
    pdfsubject={},%
    pdfkeywords={},%
    pdfcreator={pdfLaTeX},%
    pdfproducer={LaTeX}%
}

\usepackage[style=alphabetic,backref=true,backend=biber,natbib=true,
maxbibnames=99,
maxcitenames=2,
mincitenames=1,
maxalphanames=4,
minalphanames=3,
firstinits=true,
date=year,
sorting=nyt,
bibencoding=utf8,
doi=false,isbn=false,
]{biblatex}

\bibliography{biblio}

\usepackage{amssymb,amsfonts,amsmath,amsthm}
\usepackage{mathtools}
\usepackage{nicefrac}

\usepackage[active]{srcltx}

\ifdraft
\usepackage{refcheck}
\fi

\newcommand{\remove}[1]{}

\renewcommand{\th}{\ensuremath ^{\mbox{\scriptsize{\it th}}}}
\newcommand{\NOTE}[1]{%
	\ifdraft $^{\textcolor{red} \clubsuit}$
	\marginpar{\setstretch{0.43}\textcolor{red}{\bf\tiny #1}}
	\fi
	}

\newcommand{\GNOTE}[1]{%
	\ifdraft {\bf \color{black}[T: #1 ]}
	\fi
	}
\usepackage{siunitx}

\newcommand{\fnote}[1]{%
	\ifdraft {\color{green}[F: #1 ]}
	\fi
	}
\newcommand{\tnote}[1]{%
	\ifdraft {\color{red}[T: #1 ]}
	\fi
	}

	\newcommand{\lfnote}[1]{%
	\ifdraft {\color{gray}[F: Consider this after arxiv.  #1 ]}
	\fi
	}

\newcommand{\irregular}{approximative regular\xspace}
\newcommand{\gap}{\ensuremath{\Gamma}\xspace}	
\newcommand{\irrgap}{\gap-\irregular}

\newcommand{\ie}{\textit{i.e.,}\xspace}
\newcommand{\eg}{\textit{e.g.,}\xspace}
\newcommand{\etal}{\textit{et al.}\xspace}

\newcommand{\F}{{\mathcal F}}
\def\naturals{\mathbb{N}}
\newcommand{\tvdist}[1]{\Vert #1 \Vert_{\mbox{\normalfont \tiny TV}}}
\newcommand{\twonorm}[1]{\Vert #1 \Vert_{\mbox{\normalfont\scriptsize 2}}}
\renewcommand{\Pr}[1]{\mathbb{P}\left[\,#1\,\right]}

\newcommand\E[1]{\mathbb{E}\left[\,#1\,\right]}

\newcommand{\tmeet}{t_{\operatorname{meet}}}
\newcommand{\tavgmeet}{t_{\operatorname{meet}}^{\pi}}
\newcommand{\tcoal}{t_{\operatorname{coal}}}
\newcommand{\coal}{T_{\operatorname{coal}}}
\newcommand{\tmix}{t_{\operatorname{mix}}}
\newcommand{\trel}{t_{\operatorname{rel}}}
\newcommand{\thit}{t_{\operatorname{hit}}}
\newcommand{\hit}{T_{\operatorname{hit}}}
\newcommand{\tcov}{t_{\operatorname{cov}}}
\newcommand{\tsep}{t_{\operatorname{sep}}}

\newcommand{\tavghit}{t_{\operatorname{avg-hit}}}

\newcommand{\cross}{\mathrm{intersect}}
\newcommand{\diam}{\operatorname{diam}}
\DeclarePairedDelimiter{\ceil}{\lceil}{\rceil}
\DeclarePairedDelimiter{\floor}{\lfloor}{\rfloor}

\newcommand{\dist}{\operatorname{dist}}
\newcommand{\vol}{\operatorname{vol}}

\newcommand{\newprocess}{\operatorname{immortal\ process}}

\newcommand{\pimortal}{P_{\operatorname{imm}}}
\newcommand{\pinter}{P_{\operatorname{int}}}
\newcommand{\Timortal}{T_{\operatorname{imm}}}
\newcommand{\norm}[1]{\left\lVert#1\right\rVert}

\newcommand{\id}{\mathsf{id}}
\newcommand{\ids}{\mathsf{IDs}}

\renewcommand{\leq}{\leqslant}
\renewcommand{\geq}{\geqslant}

\renewcommand{\tilde}{\widetilde}
\renewcommand{\epsilon}{\varepsilon}

\newcommand{\thmref}[1]{\autoref{thm:#1}}

\newcommand{\lemref}[1]{\autoref{lem:#1}}

\newcommand{\eq}[1]{\eqref{eq:#1}}

\renewcommand{\tilde}{\widetilde}

\newcommand\ifrac[2]{#1/#2}
\renewcommand\ifrac\nicefrac

\usepackage{xcolor} 
\usepackage{booktabs} 
\usepackage{float} 

 \usepackage{aliascnt}
\def\NewTheorem#1#2{%
  \newaliascnt{#1}{theorem}
  \newtheorem{#1}[#1]{#2}
  \aliascntresetthe{#1}
  \expandafter\def\csname #1autorefname\endcsname{#2}
}

 \newtheorem{theorem}{Theorem}[section]
\NewTheorem{lemma}{Lemma}
\NewTheorem{corollary}{Corollary}
\NewTheorem{conjecture}{Conjecture}
\NewTheorem{observation}{Observation}
\NewTheorem{assumption}{Assumption} 

\NewTheorem{definition}{Definition}
\NewTheorem{proposition}{Proposition}
\NewTheorem{remark}{Remark}

\NewTheorem{claim}{Claim}

\theoremstyle{remark}

\newcommand\blfootnote[1]{%
  \begingroup
  \renewcommand\thefootnote{}\footnote{#1}%
  \addtocounter{footnote}{-1}%
  \endgroup
}

\providecommand{\keywords}[1]{\textbf{\textit{keywords---}} #1}

\definecolor{darkred}{rgb}{0.5,0,0}
\definecolor{lightblue}{rgb}{0,0.4,0.8}
\definecolor{darkgreen}{rgb}{0,0.5,0}
\definecolor{grey}{rgb}{0.5, 0.5, 0.5}

\begin{document}

\title{On coalescence time in graphs\\{\large When is coalescing as fast as
meeting?}\ifEA~\medskip \\ \small{Extended Abstract}\footnote{The full paper appears at the end of the extended abstract in the submission file.}\fi}

\author{Varun Kanade\thanks{This work was supported in part by The Alan Turing Institute under the EPSRC grant EP/N510129/1.} \\ University of Oxford  \\  varunk@cs.ox.ac.uk \and
Frederik Mallmann-Trenn\thanks{This work was carried out while a student at the \'{E}cole normale sup\'{e}rieure and Simon Fraser University. This work was supported in part by NSF Award Numbers  CCF-1461559, CCF-0939370, and CCF-1810758.} \\ MIT \\  mallmann@mit.edu \and
Thomas Sauerwald\thanks{This work was supported by the ERC Starting Grant (DYNAMIC MARCH).} \\  University of Cambridge
\\  thomas.sauerwald@cl.cam.ac.uk}
\date{}
\thispagestyle{empty}
\setcounter{page}{0}
\pagenumbering{gobble}
\maketitle
 
\begin{abstract}
	Coalescing random walks is a fundamental stochastic process, where a set of
	particles perform independent discrete-time random walks on an undirected
	graph. Whenever two or more particles meet at a given node, they merge and
	continue as a single random walk. The {\em coalescence time} is defined as
	the expected time until only one particle remains, starting from one
	particle at every node. Despite recent progress such as by Cooper,
	Els\"asser, Ono, Radzik~\cite{CEOR13} and Cooper, Frieze and
	Radzik~\cite{CFR09}, the coalescence time for graphs such as binary trees,
	$d$-dimensional tori, hypercubes and more generally, vertex-transitive
	graphs, remains unresolved.
	
	We provide a powerful toolkit that results in tight bounds for various
	topologies including the aforementioned ones. The meeting time is defined as
	the worst-case expected time required for two random walks to arrive at the
	same node at the same time. As a general result, we establish that for
	graphs whose meeting time is only marginally larger than the mixing time (a
	factor of $\log^{2} n$), the coalescence time of $n$ random walks equals the
	meeting time up to constant factors. This upper bound is complemented by the
	construction of a graph family demonstrating that this result is the best
	possible up to constant factors. For almost-regular graphs, we bound the
	coalescence time by the hitting time, resolving the discrete-time variant of
	a conjecture by Aldous for this class of graphs. Finally, we prove that for
	any graph the coalescence time is bounded by $O(n^3)$ (which is tight for
	the Barbell graph); surprisingly even such a basic question about the
	coalescing time was not answered before this work. By duality, our results
	give bounds on the voter model and therefore give bounds on the consensus
	time in arbitrary undirected graphs.
	
	We also establish a new bound on the hitting time and cover time of regular
	graphs, improving and tightening previous results by Broder and
	Karlin~\cite{BK89}, as well as those by Aldous and Fill~\cite{AF14}.
	\blfootnote{An extended abstract based on this work appeared in SODA 2019.}
\end{abstract}
\vfill
\keywords{coalescing time, meeting time, random walks, voter model}\\

\ifEA
\else
\newpage
\tableofcontents 
\fi
\newpage
\pagenumbering{arabic}

\section{Introduction}

Coalescing random walks is a fundamental stochastic process on \emph{connected}
and \emph{undirected} graphs. The process begins with particles on some subset
of the nodes in the graph. At discrete time-steps, every particle performs one
step of an independent   random walk.%
\footnote{Throughout this paper, we use random walk and particle
interchangeably, assuming that every random walk has an identifier.}
Whenever two or more particles arrive at the same node at the same time-step,
they merge into a single particle and continue as a single random walk. The
\emph{coalescence time} is defined as the first time-step when only one
particle remains. The coalescence time depends on the number and starting
positions of the particles.

Studying the coalescence time is of substantial importance in distributed computing:
At the heart of many distributed computing applications lie consensus protocols and leader election \eg data consistency, consolidation of replicated states, synchronization of processes and devices \cite{Pel02, DGMSS11} and communication networks \cite{PVV09}). 
 Other applications  of the coalescence process appear in robotics \cite{GORN17}; here, robots perform random walks to gather samples from their environment and need to communicate these samples to all other robots.
Studying the coalescence time also implies results for other interaction types of random walks including predator and prey particles as well has annihilating particles \cite{interacting}.

\paragraph{Relationship to consensus protocols}
 Arguably the simplest consensus protocol achieving consensus on any undirected graph is the voter model.
%
%
Initially, every node has  a distinct opinion. At every round, each node
chooses synchronously one of its neighbors at random and adopts that node's
opinion.  The \emph{consensus time} is defined as the time it takes until only
one opinion remains. The voting process viewed backwards is exactly the same as
the coalescence process starting with a random walk on every node; thus, the
coalescence time and consensus time have the same distribution.  
Despite  recent progress by Cooper~\etal~\cite{CEOR13,CFR09} and
Berenbrink~\etal~\cite{BGKM16}, the coalescence time and consensus time are far from being well-understood---even for certain fundamental graphs as we describe below. 
Recently, there have been several studies on variants of the voter
model, most notably 2-Choices and 3-Majority which received ample attention~\cite{CER14, BCNPS15, CERRS15, BCNPT16, CRRS16, EFKMT16, BCEKMN17, GL17}. However, the behavior of these
processes is fundamentally different and despite their efficiency in reaching consensus on expanders and cliques, they are unsuitable on more general undirected graphs as the consensus time is exponential in some graphs.

In this paper, we follow the approach of \citet{CEOR13} and \citet{HP01} \fnote{more?} and study the consensus time through 
the more tangible analysis of the coalescence time.
When starting with two particles, the coalescence time is referred to as the
\emph{meeting time}. Let $\tmeet$ denote the worst-case expected meeting time over
all pairs of starting nodes and let $\tcoal$ denote the expected coalescence
time starting from one particle on every node. It is clear that $\tmeet \leq
\tcoal$; as for an upper bound, it can be shown that $\tcoal = O(\tmeet \log
n)$, where $n$ is the number of nodes in the graph. The main idea used to
obtain the bound is that the number of surviving random walks halves roughly
every $\tmeet$ steps. A proof of the result appears implicitly in the work of
Hassin and Peleg~\cite{HP01}.

Aldous~\cite{Ald:1991} showed in \emph{continuous-time} that the meeting time
is bounded by the maximum hitting time, $\thit := \max_{u, v} \thit(u, v)$,
where $\thit(u, v)$ denotes the expected time required to hit $v$ starting from
vertex $u$.
%
%
%
We observe that the result of Aldous also holds in discrete time. Thus, this
gives a bound of $O(\thit \log n)$ for the coalescing time; however, in general
$O(\thit)$ may be a loose upper bound on $\tmeet$.  In recent work, Cooper
\etal~\cite{CEOR13} provide results that are better than $O(\tmeet \log n)$ for
several interesting graph classes, notably expanders and power-law graphs.
They show that $\tcoal = O(({\log^4 n +
\twonorm{\pi}^{-2}})\cdot({1-\lambda_2})^{-1})$, where $\lambda_2$ is the
second largest eigenvalue of the transition matrix of the random walk and $\pi$
is the stationary distribution. Berenbrink \etal~\cite{BGKM16} show that
$\tcoal = O(m/(d_{\min} \cdot \Phi))$, where $m$ is the number of edges,
$d_{\min}$ is the minimum degree and $\Phi$ is the conductance. Their result
improves on that of Cooper \etal for certain graph classes, \eg cycles.

As mentioned before, despite the recent progress due to Cooper~\etal~\cite{CEOR13} and
Berenbrink~\etal~\cite{BGKM16},  for many fundamental graphs such as the binary
tree, hypercube and the ($d$-dimensional) torus, the coalescing time in the
discrete setting remains unsettled. We provide a rich toolkit allowing us to
derive tight bounds for many graphs including all of the aforementioned ones.
One of our main results establishes a relationship between the ratios
$\tcoal/\tmeet$ and $\tmeet/\tmix$, where $\tmix=\tmix(1/e)$ denotes the mixing
time.%
\footnote{The \emph{mixing time} is the first time-step at which the
distribution of a random walk starting from an arbitrary node is close to the
stationary distribution.} 
In particular, the result shows that if $\tmeet/\tmix = \Omega(\log^2 n)$, then
$\tcoal = O(\tmeet)$; however, we also provide a more fine-grained tradeoff.
For almost-regular graphs,\footnote{We call a graph \emph{almost-regular} if
$\deg(u)=\Theta(\deg(v))$ for all $u,v\in V$.} we bound the coalescence time
by the hitting time.  For vertex-transitive graphs we show that the coalescence
time, the meeting time, and the hitting time are equal up to constant factors.
Finally, we prove that for any graph the coalescence time is bounded by
$O(n^3)$; it can be easily verified that this is tight by considering the barbell graph. Surprisingly, the right bound on this fundamental quantity was not known prior to this work. Unlike in the analogous case of the cover time~\cite{AF14} where such a bound can be easily derived, the argument in the case of coalescence time appears significantly involved.\footnote{Cooper~\etal~\cite{CEOR13} mistakenly stated, as a side remark, that
this last result was a simple consequence of their main result.} Prior to this work, \citet{HP01} had shown a worst-case upper bound of $O(n^3 \log n)$. We
also give worst-case upper and lower bounds on the meeting time and coalescence
time that are tight for general graphs and regular (or nearly-regular) graphs.

In the process of establishing bounds on the coalescence time, we develop techniques to give tight bounds on the meeting time. We apply these to
various topologies such as the binary tree, torus and hypercube. We believe
that these techniques might be of more general interest.

The process of coalescing random walks was first  studied in
\emph{continuous time}; in this case, particles jump to a random neighboring
node when activated according to a Poisson clock with mean $1$.  As
\citet{CN16} recently pointed out \emph{``It is however, not clear whether the
continuous-time results apply to the discrete-time setting''}, and to the best
of our knowledge, there is no general way in which results in \emph{continuous
time} can be transferred to \emph{discrete time} or vice versa, even when the
random walks in discrete-time are \emph{lazy}.  In the continuous time setting,
\citet{C89} show that the coalescence time is bounded by $\Theta(\thit)$ for
tori. Oliveira~\cite{O12} \fnote{One of the reviewers asked us how the results in section 4 (the hitting time section) related to \cite{O12} }showed that the coalescence time is $O(\thit)$ in
general.  In a different work, \citet{Oli13} derived so-called mean field
conditions, which are sufficient conditions for the coalescing process on a
graph to behave similarly to that on the complete graph up to scaling by the expected meeting time. His main result (for
non vertex-transitive graphs) in \cite[Theorem 1.2]{Oli13}, implies that
$\tcoal=O(\tmeet)$ whenever $\tmix \cdot \pi_{\max} = O(1 / \log^4 n)$. One of
our main results, Theorem~\ref{thm:mixtradeoff}, implies $\tcoal=O(\tmeet)$ whenever
$\tmix/\tmeet = O(1 / \log^2 n)$. Notice that since $\tmeet \geq 1/(\| \pi
\|_2^2) \geq 1/\pi_{\max}$, our condition is considerably more
general---however, the results in \cite{Oli13} also establish mean-field
behavior (that is, when suitably scaled, the distribution of the coalescence time is similar to that on a complete graph), while ours are only concerned with the expected coalescence time,
$\tcoal$. On the other hand, our result also applies to graphs where $\tcoal \gg \tmeet$ such as the star graph,
and together with Theorem~\ref{thm:lowerboundgraph},  demonstrate that the trade-off between meeting and mixing time is the best possible.

\def\vsp{1ex} 
\newcolumntype{C}[1]{>{\centering\arraybackslash}p{#1}}
\newcolumntype{L}[1]{>{\arraybackslash}p{#1}}
\newcommand{\twolines}[2]{\begin{tabular}{@{}l@{}} #1  \hfil \\#2 \end{tabular}}
\newcommand{\ripref}[2]{ \it \hyperref[#2]{\it #1}~\ref{#2}    }
\definecolor{testcolor}{HTML}{AAD3F7}
%

\begin{table}
	\rowcolors{2}{testcolor!50}{}
	\vspace{0.5em} 
	\resizebox{\linewidth}{!}{ 
	\begin{tabular}{L{2.4cm}L{2.8cm}L{1.8cm}L{3.9cm}llL{1.5cm}L{1.0cm}L{1.5cm}} 
		\toprule
		Graph & \twolines{$\tmix$}{} & \twolines{$\tmeet$}{} & \twolines{}{} &
		\twolines{$\tcoal$}{} & \twolines{}{} & \twolines{$\thit$}{} \\
		\midrule
		Binary tree  & $\Theta(n)$ & $\Theta(n \log n)$ & 
		\ripref{Thm.}{thm:hittingtime}$\&$\ripref{Thm.}{thm:treelower}& 
		$\Theta(n \log n)$ & \ripref{Thm.}{thm:hittingtime}$\&$\ripref{Thm.}{thm:treelower}&
		$\Theta(n \log n)$  \\
		Clique & $\Theta(1)$ & $\Theta(n)$ &
		\cite{CEOR13,BGKM16} $\&$\ripref{Thm.}{thm:mixtradeoff} & $\Theta(n )$ &
		\cite{CEOR13,BGKM16} $\&$\ripref{Thm.}{thm:mixtradeoff} & $\Theta(n )$ \\
		Cycle  & $\Theta(n^2)$ &  $\Theta(n^2)$ &
		\cite{BGKM16} $\&$\ripref{Thm.}{thm:hittingtime} & $\Theta(n^2)$ &
		\cite{BGKM16} $\&$\ripref{Thm.}{thm:hittingtime} & $\Theta(n^2 )$ \\
		Rand. $r$-reg.  & $\Theta(\log n)$ & $\Theta(n)$ &
		\cite{CFR09,CEOR13,BGKM16} $\&$\ripref{Thm.}{thm:mixtradeoff} &
		$\Theta(n)$ & \cite{CFR09,CEOR13,BGKM16}
		$\&$\ripref{Thm.}{thm:mixtradeoff} & $\Theta(n )$ \\
		Hypercube  & $\Theta(\log n \log\log n)$ & $\Theta(n)$
		& \ripref{Thm.}{thm:hittingtime}  & $\Theta(n )$ &
		\ripref{Thm.}{thm:mixtradeoff}&  $\Theta(n) $              \\
		Path  & $\Theta(n^2)$ & $\Theta(n^2)$ & \cite{BGKM16} $\&$\ripref{Thm.}{thm:hittingtime} & $\Theta(n^2)$ &
		\cite{BGKM16} $\&$\ripref{Thm.}{thm:hittingtime} & $\Theta(n^2)$ \\
		Star & $\Theta(1)$ & $\Theta(1)$ & folklore & $\Theta(\log n)$
		&\cite{HP01},\ripref{Prop.}{lem:beer} $\&$ \ripref{Thm.}{thm:graphclasses} & $\Theta(n)$ \\ 
		Torus $(d=2)$& $\Theta(n)$ & $\Theta(n \log n)$ &
		\ripref{Thm.}{thm:hittingtime}  & $\Theta(n \log n)$  &
		\ripref{Thm.}{thm:hittingtime} & $\Theta(n\log n)$ \\
		Torus $(d>2)$ & $\Theta(n^{2/d})$ & $\Theta(n)$ &
		\ripref{Thm.}{thm:mixtradeoff} & $\Theta(n)$ &
		\ripref{Thm.}{thm:mixtradeoff} & $\Theta(n)$ \\ 
		\bottomrule
	\end{tabular}}%
	\begin{flushleft}
		\begin{small}
			\caption{\label{mastertable}A summary of bounds on the mixing,
			meeting, coalescence and hitting times for fundamental topologies for
			discrete-time random walks. All bounds on the mixing and hitting
			times appear directly or implicitly in \cite{AF14}.\ifEA~The cited results may refer to those appearing in the full paper.\fi}
		\end{small}
	\end{flushleft}
\end{table}
\subsection{Contributions}

In this work, we provide several results relating the coalescence and meeting
times to each other and to other fundamental quantities of random walks
on undirected graphs. In particular, our focus is on understanding for which
graphs the coalescence time is the same as the meeting time, as we know that
$\tcoal$ is always in the rather narrow interval of $[\tmeet, O(\tmeet
\cdot \log n)]$.  As a consequence of our results, we derive new and re-derive existing
bounds on the meeting and coalescence times for several graph families of
interest. These results are summarized in \autoref{mastertable} and discussed
in greater detail in \autoref{sec:special}\ifEA~of the full paper\fi.  Formal definitions of all
quantities used below appear in \autoref{sec:notation}\ifEA~of the full paper\fi.
Throughout this paper, we assume that random walks are \emph{lazy} meaning that w.p. $1/2$ the walk stays put.

Our first main result relates $\tcoal$ to $\tmeet$ and $\tmix$. As already
mentioned in the introduction, the  crude bound $\tcoal = O(\tmeet \log n)$ is
well-known. However, this bound is not in general tight, as demonstrated by our
result below.

\begin{theorem} 
	\label{thm:mixtradeoff} 
	For any graph $G$, we have
	\[ 
				\tcoal = O\left( \tmeet \left( 1 + \sqrt{\frac{\tmix}{\tmeet}} \cdot \log n \right) \right),
			\]
			\fnote{Page 3. In connection with Theorem 1.1, could you indicate the range of $\tmeet/\tmix$ - I don't see the point. Am I wrong?} 
			Consequently, when $\tmeet \geq \tmix \log^2 n$, $\tcoal =
			O(\tmeet)$.
\end{theorem}

The proof of \autoref{thm:mixtradeoff} appears in
\autoref{sec:upperboundtcoal}\ifEA~of the full paper\fi. One interesting aspect about this bound is that
it can be used to establish $\tcoal=\Theta(\tmeet)$ even without having to know
the quantities $\tmeet$ or $\tmix$. This flexibility turns out to be particularly
useful when dealing with random graph models for ``real world'' networks, where
we establish (nearly-)tight and sublinear bounds (w.r.t. to the number of vertices)  in~\autoref{sec:realworld}\ifEA~of the full paper\fi.

Another interesting feature of our theorem is that the main result of Cooper~\etal~\cite[Theorem~1]{CEOR13} can be reproven by combining~\cite[Theorem~2]{CEOR13} with \autoref{thm:mixtradeoff} (see~ \autoref{pro:frederik}\ifEA~in the full paper\fi).

Our next main result shows that the bound in \autoref{thm:mixtradeoff} is tight up to a constant factor, which we establish by constructing an explicit family
of graphs. Interestingly, for this family of almost-regular graphs we also have $\thit \gg
\tmeet$, thus showing that $\thit$ may be a rather loose upper bound for
$\tcoal$ in some cases.\footnote{Note that the star also exhibits  $\thit \gg \tmeet$. However, the star is not almost-regular.}

\begin{theorem}\label{thm:lowerboundgraph}
	For any sequence $(\alpha_n)_{n \geq 0}$, $\alpha_n \in [1,\log^2 n]$ there
	exists a family of almost-regular graphs $ (G_n)$, with $G_n$ having $\Theta(n)$ nodes and
	satisfying $\frac{\tmeet}{\tmix}= \Theta(\alpha_n)$ such that
	\[
	  \tcoal = \Omega \left(\tmeet \cdot \Bigl(1+\sqrt{ \frac{\tmix}{\tmeet}}
	  \cdot  \log n \Bigr)\right).
	\]
\end{theorem}

The above two results show that that $\tmeet/\tmix$ should be $\Omega(\log^2
n)$ to guarantee that $\tcoal = O(\tmeet)$.

A natural question is therefore whether in the case of structured
sub-classes such as regular graphs, or vertex-transitive graphs, or special
graphs such as grids, tori, binary trees, cycles, real-world (power-law)
graphs, \textit{etc.}, better bounds can be obtained through other methods. We provide results that
are tight or nearly tight in several of these cases; some of these results were
previously known using other methods, some are novel to the best of our
knowledge. 

\begin{theorem}
	\label{thm:hittingtime} 
The following hold for graphs of the stated kind
	\begin{enumerate}[(i)]
	
	\item  For any graph $G$,
        \[
          \tcoal = O\left( \thit \cdot \log \log n \right).
        \]
        \item For any graph $G$ with maximum degree $\Delta$ and average degree $d$,
        \[ \tcoal = O\left( \thit + \tmeet \cdot \log (\Delta/d) \right). \]
        Hence for any almost-regular graph $G$, 
			$ \tcoal = O\left( \thit\right). $		
			\item For any  vertex-transitive $G$, 
			\[ \tcoal = \Theta\left( \tmeet\right) = \Theta\left( \thit \right). \]

		\item In the case of 		
		binary trees, $d$-dimensional tori/grids, paths/cycles, expanders, hypercubes, random power law graphs,\footnote{The exact model is specified in \autoref{sec:realworld}\ifEA~of the full paper\fi.} we have $\tcoal = \Theta(\tmeet)$.
\end{enumerate}
\end{theorem}
The proof of the first three statements of~\autoref{thm:hittingtime} appear in \autoref{sec:hittingtime} and the last statement follows from the results in \autoref{sec:special}\ifEA~of the full paper\fi.
We point out that since $\tmeet=O(\thit)$ for any graph,\footnote{In \autoref{pro:relatingmeetandhit}\ifEA~of the full paper\fi, we prove this formally by following the proof for the continuous setting \cite[Proposition 14.5]{AF14}.}
\autoref{thm:hittingtime} implies the bound $\tcoal=O(\thit)$ not only for
almost-regular graphs, but also for dense graphs where $|E|=\Theta(n^2)$.
This settles the discrete-time analogue of a conjecture by Aldous~\cite[Open
Problem 14.13]{AF14} for these graph classes.
 In very recent work, Oliveira and Peres
improve on these results and establish that $\tcoal = O(\thit)$ holds for all
undirected graphs~\cite{OP:2018}.

Another natural question is to express $\tmeet$ or $\tcoal$ solely in terms of
$\tmix$, the spectral gap $1-\lambda_2$ or other connectivity properties of
$G$. We derive several such bounds on $\tmeet$, $\thit$ and $\tcoal$.

As a by-product of our techniques, we also derive new bounds on $\thit$ and
$\tcov$, the cover-time. The detailed results are given in
\autoref{sec:meeting}\ifEA~of the full paper\fi, but we highlight the results for regular graphs here:
\begin{theorem}\label{thm:spectral}
Let $G$ be any graph with $\gap=\Delta/\delta$, where $\Delta$ is the maximum degree and $\delta$ the minimum degree.
 It holds that \[\thit = O(\gap n/ \sqrt{1-\lambda_2})=O(\gap n/\Phi),\]
 where $\Phi$ is the conductance of the graph and  $\lambda_2$ is the  second largest eigenvalue of the transition matrix $P$ of a lazy random walk.
	Consequently, $\tmeet \leq \tcoal = O(\gap n \log(\gap)/ \sqrt{1-\lambda_2})=O(\gap n \log(\gap)/ \Phi)$  and $\tcov = O(\gap n \log n/\Phi)$.
\end{theorem}

We point out that so far the best possible bound on $\tcoal$ for regular graphs
has been $\tcoal = O( n / (1-\lambda_2))$ from \cite{CEOR13}.\footnote{Alternatively, the
same bound as the known bound can also be derived from the bound on the conductance in \cite{BGKM16} together with
Cheeger's inequality.} The best possible bound on $\thit$ (and $\tcov$) in terms
of $1-\lambda_2$, was $\thit = O(n/ (1-\lambda_2))$ and $\tcov = O( n \log n /
(1-\lambda_2))$ due to Broder and Karlin~\cite{BK89} from 1989. In all four
cases, $\tmeet$, $\tcoal$, $\thit$, and $\tcov$, \autoref{thm:spectral} improves the
dependency on $1/(1-\lambda_2)$ (or, equivalently $\tmix$), by almost a
square-root (we refer the reader to \autoref{thm:cooperimproved} and \autoref{thm:nonreg}\ifEA~of the full paper\fi~for further details). As a result of this improvement, we get a bound of $O(n/\Phi)$ on the hitting time which is the best known bound on the hitting time (and cover time) in terms of the conductance  and improves the bound of \cite[Corollary 6.2.1]{AF14} by a factor of $1/\Phi$.

We also derive a general lower bound on $\tmeet$ that combines the trivial
bound, $1/\| \pi \|_2^2$, with the minimum number of
collisions~(see~\autoref{thm:meetingtime}.$(iii)$\ifEA~of the full paper\fi). Although this bound does not directly yield the correct lower bound for binary trees, it forms the basis of a later analysis in \autoref{thm:treelower}\ifEA~of the full paper\fi.

Finally, we also provide asymptotically tight worst-case bounds on $\tmeet$ and $\tcoal$.  We show
that on any graph the coalescence time must be at least $\Omega(\log n)$ and is
no more than $O(n^3)$. For regular (and in particular vertex-transitive) graphs these bounds
become $\Omega(n)$ and $O(n^2)$ (See also~\autoref{lessertable} on
page~\pageref{lessertable}\ifEA~of the full paper\fi, which also contains an explanation why these bounds are asymptotically tight.) These two new upper bounds for general and regular
graphs complete the picture of worst-case bounds:

\begin{theorem}\label{thm:graphclasses}
The following hold for graphs of the stated kind.
\begin{enumerate}[(i)]
	\item For any graph $G$ we have $\tmeet \in [\Omega(1), O(n^3)]$ and $\tcoal \in [\Omega(\log n),O(n^3)]$. 
	\item For any regular graph $G$ we have $\tmeet, \tcoal \in [\Omega(n), O(n^2)]$.
\end{enumerate}
\end{theorem}
The proof of \autoref{thm:graphclasses} appears in \autoref{sec:worst-case-bounds}\ifEA~of the full paper\fi. \medskip

\subsubsection*{Summary of Technical Contributions} 

Our work also makes several technical contributions, which might be of interest for future research on coalescing walks and other stochastic processes; these are explained in greater detail in \autoref{sec:intro:proofideas}. Below we give a very brief summary.
\begin{itemize}
	\item {\bf Conditional Expectation Approach}. Most of our results make use of the conditional expectation approach given in \eqref{eq:central}, a very simple yet extremely powerful tool, which to the best of our knowledge has not been used in the context of meeting and coalescing times before.
	\item {\bf Division of Particles into two Groups}. One basic ingredient in our proof is a domination result that allows us to divide random walks into a group of ``destroyers'' ($\mathcal{G}_1$), which are particles that cannot be eliminated, and a group of remaining particles ($\mathcal{G}_2$), which can be eliminated by any other random walk. This domination result might be helpful to analyze other stochastic processes involving different types of particles, e.g.~\cite{CFR09}.
	\item {\bf New Concentration Inequalities}. We derive a new concentration inequality for random walks on graphs in \autoref{sec:conc}\ifEA~of the full paper\fi. Unlike previous approaches which are based on the mixing time (or the closely related spectral gap), our new inequality depends only on the hitting time and improves on the existing bounds when the mixing time is close to the hitting time. These tighter inequalities are required to derive worst-case upper bounds on the colaescence time.
\end{itemize}

\subsection{Proof Ideas and Technical Contributions}\label{sec:technical}
\label{sec:intro:proofideas}

When dealing with processes involving concurrent random walks, a significant
challenge is to understand the behavior of ``short'' random walks.
This challenge appears in several settings, \eg in the context of cover time
of multiple random walks~\cite{Alon11,ER09}, where \citet[Section 6]{ER09}
highlight the difficulty in analyzing the hitting time distribution before its
expectation. In the context of concentration inequalities for Markov chains,
\citet[p.~863]{Lez89} points out the requirement to spend at least mixing time
steps before taking any samples.  Related to that, in property testing, dealing
with graphs that are far from expanders has been mentioned as one of the major
challenges to test the expansion of the graph by~\citet{CS10}.

In our setting, we also face these generic problems and devise different
methods to get a handle on the meeting time distribution before its
expectation. Despite our focus being on coalescing and meeting times, several of
our approaches can be leveraged to derive new bounds on other random walk
quantities such as hitting times or cover times (see~\autoref{sec:meeting}\ifEA~of the full paper\fi).

\subsection*{Bounds on $\tcoal$ in terms of $\tmix$ and $\tmeet$}

The key ingredient in the proof of \autoref{thm:mixtradeoff}, where we express
$\tcoal$ as a tradeoff between $\tmeet$ and $\tmix$ is a better understanding of meeting events prior to the meeting time. More precisely, we derive a tight
bound on the probability $p_\ell$ that two random walks meet before $\ell$
time-steps, for $\ell$ in the range $[\tmix, \tmeet]$. Arguing about meeting
probabilities of walks that are much shorter than $\tmeet$ allows us to
understand the rate at which the number of \emph{alive} 
random walks is decreasing.

Optimistically, one may hope that starting with $k$ random walks, as there are
${k \choose 2}$ possible meeting events, roughly ${k \choose 2} \cdot p_\ell$
meetings may have occurred after $\ell$ time-steps. However, the
non-independence of these events turns out to be a serious issue and we require
a significantly more sophisticated approach to account for the dependencies.
We divide the $k$ random walks into disjoint groups $\mathcal{G}_1$ and $\mathcal{G}_2$ (with
$|\mathcal{G}_1|$ usually being much smaller than $|\mathcal{G}_2|$) and walks of $\mathcal{G}_1$ can't be eliminated.  The domination
of the real process by the group-restricted one is established by introducing a formal 
concept called $\newprocess$ at the beginning of 
\autoref{sec:process}\ifEA~of the full paper\fi.
In this  stochastic process, we can expose the random walks
of $\mathcal{G}_1$ first and consider meetings with random walks in $\mathcal{G}_2$
(for an illustration, see \autoref{fig:selector} on page~\pageref{fig:selector}\ifEA~of the full paper\fi). 
Conditioning
on a specific exposed walk in $\mathcal{G}_1$, the events of the different walks in $\mathcal{G}_2$
meeting this exposed walk are indeed independent. 
In fact, we will also use the symmetric case where the roles of $\mathcal{G}_1$ and $\mathcal{G}_2$ are switched. 
Thus, the problem then
reduces to calculating the probability of a random walk in $\mathcal{G}_2$ having a `good
trajectory', \ie one which many random walks in $\mathcal{G}_1$ would meet with large
enough probability.

Surprisingly, it suffices to divide trajectories into only two categories\ifEA\else~(\autoref{lem:classes2})\fi. Although, one may expect that a more fine-grained
classification of trajectories would result in better bounds, this turns out not
to be the case. In fact, the bound that we derive on the coalescing time in
\autoref{thm:mixtradeoff} is tight, and this is precisely due to the
tightness of~\autoref{lem:classes2}\ifEA~of the full paper\fi. The tightness is established by the
following construction (cf.~\autoref{fig:holygraph2}).  The graph is designed
such that the vast majority of meetings (between any two random walks) occur in
a relatively small part of the graph ($G_2$ in \autoref{fig:holygraph2}). On
average, it takes a considerable number of time-steps before random walks
actually get to this part of the graph. What this implies is that for
relatively short trajectories (of length significantly smaller than $\tmeet$),
it is quite likely that other random walks will not meet them\ifEA\else~(cf.~\autoref{lem:classes2})\fi.
There is a bit of a dichotomy here, once a walk reaches $G_2$ it is likely that
many random walks will meet it; however, a random walk not reaching $G_2$ is
unlikely to be met by any other random walk.

\begin{figure}

	\centering
\includegraphics[page=2]{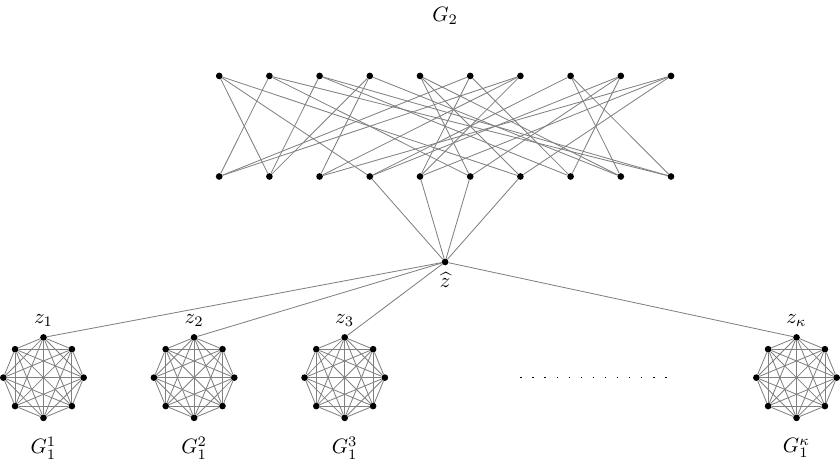}
\caption{The graph\ifEA~that exhibits\else~described in \autoref{sec:lowerbound}
with\fi~$\tcoal = \Omega(\tmeet + \sqrt{\tmeet/\tmix}\cdot \log n \cdot \tmix)$. 
}
\label{fig:holygraph2}
\end{figure}

 Equipped with \autoref{thm:mixtradeoff}, we can
  bound $\tcoal=\Theta(\tmeet)$ for all graphs satisfying $\tmeet /\tmix \geq \log^2 n$.
 Therefore, the problem of bounding $\tcoal$ reduces to bounding $\tmeet$.

For some of the other results including \autoref{thm:lowerboundgraph} and
\autoref{thm:hittingtime}, we will need a more fine-grained approach to derive
lower (or upper bounds) on the probability that two walks meet during a certain
number of steps, which may or may not be smaller than the mixing time or meeting
time. The starting point is the following simple observation. If we have two
random walks $(X_t)_{t \geq 0}$ and $(Y_t)_{t \geq 0}$, and count the number of
collisions $Z:= \sum_{t=0}^{\tau-1} \mathbf{1}_{X_t=Y_t}$ before time-step
$\tau \in \naturals$, %
then 
\begin{align}\label{eq:central}
	\Pr{ Z \geq 1} &= \frac{ \E{Z}    }{ \E{Z \, \mid \, Z \geq 1} }.
\end{align}
If we further assume that both walks start from the stationary distribution, then we have
\begin{align*}
  \Pr{ Z \geq 1} &= \frac{  \tau \cdot \| \pi \|_2^2  }{ \E{Z \, \mid \, Z \geq 1} }.
\end{align*}
To the best of our knowledge, this is the first application of this formula to meeting (and coalescence) times. However, we should mention that variants of this formula have been used by Cooper and
Frieze in several works (\eg~\cite{CF05}) to derive accurate bounds on the hitting (and cover time) on various classes of random graphs, and in \citet{BPS12} to bound the collisions of random walks on infinite graphs. Using~\eqref{eq:central}, we are able to obtain several improvements to
existing bounds on the meeting time, and as a consequence for coalescing time.
We believe that our work further highlights the power of this basic identity. 

The crux of \eqref{eq:central} is that in order to lower (or upper) bound the
probability that the two walks meet, we need to derive a corresponding bound on
$\E{Z \, \mid \, Z \geq 1}$, \ie the number of collisions conditioning on the
occurrence of at least one collision. Our results employ various tools to get a handle on this quantity, but here we
mention one that is quite intuitive: 
\begin{align}
  \E{Z \, \mid \, Z \geq 1} \leq \max_{u \in V} \sum_{t=0}^{\tau-1} \sum_{v \in V} \left( p_{u,v}^t \right)^2. \label{eq:simp}
\end{align}
The inner summand $\sum_{v \in V} (p^t_{u, v})^2$ is the probability that two
walks starting from the same vertex $u$ will meet after a further $t$ steps. Thus,
summing over  $t$ and conditioning on the first meeting happening
(\ie the condition $Z \geq 1$) at some vertex $u$ before time-step $\tau$ yields
the bound in~\eqref{eq:simp}. Despite the seemingly crude nature of this
bound, it can be used to derive new results for $\thit, \tmeet$ and $\tcoal$
that significantly improve over the state-of-the-art for regular graphs (see
\autoref{sec:meeting}\ifEA~of the full paper\fi, or the last paragraph in this section for a summary).

\subsubsection*{Bounds on $\tcoal$ in terms of $\thit$}

The derivation of our bounds on $\tcoal$ in terms of $\thit$ (\autoref{thm:hittingtime}) are based on two general reduction results, that might be useful in other applications:
\begin{theorem}[Reduction Results]\label{thm:lastintro}
The following results hold for any graph $G$:
\begin{enumerate}
\item The coalescence process reduces the number of walks from $n$ to $O(\log^3 n)$ in $O(\thit)$ steps with probability at least $1-n^{-1}$. (see \autoref{thm:mostgeneral}\ifEA~of the full paper\fi)
\item The coalescence process reduces the number walks from $\log^4 n$ to $(\Delta/d)^{O(1)}$ in $O(\thit)$ steps in expectation, where $\Delta$ is the maximum degree and $d$ is the average degree (see \autoref{thm:keylemma}\ifEA~of the full paper\fi)
\end{enumerate}
\end{theorem}

A basic ingredient are new concentration inequalities, which are derived in
\autoref{sec:conc}\ifEA~of the full paper\fi. Our concentration inequalities yield sufficiently strong
bounds for upper tails of returns (or other, possibly more complex random
variables) by a random walk of length $\thit$, while most of the existing
bounds (\eg~\cite{CLLM12,Lez89}) require that the expectation of the random
variable is at least as large as $\tmix$. While $\tmix \leq \thit$ in general, the challenging case in our analysis is when $\tmix \approx \thit$ and in this cases our concentration inequalities provide stronger upper tails than the existing ones.

Equipped with these concentration results, the proof of~\ifEA\autoref{thm:lastintro} (Part 1)\else\autoref{thm:mostgeneral}\fi~is surprisingly simple and rests again on
\eq{central}. First, by a straightforward bucketing argument on the degree
distribution, we show that with high probability, we can find for each random
walk $(X_t)_{t \geq 0}$ with label $i$ a set $S$ (depending on the trajectory
of $X_t$), so that with high probability, (i) each vertex in $S$ is visited
frequently during $O(\thit)$ steps, and (ii) each vertex in $S$ has the same
degree up to constant factors. Conditioning on this, it follows that a second random
walk $(Y_t)_{t \geq 0}$ will have sufficient collisions with $(X_t)_{t \geq 0}$
in expectation, \ie $\E{Z}$ is large enough.  To bound $\E{Z \, \mid \, Z \geq
1}$, we use the concentration inequalities to establish that with high
probability, the trajectory $(X_t)_{t \geq 0}$ will be good in the sense that
$\E{Z \, \mid \, (x_0,x_1,\ldots), Z \geq 1}$ is not too large.  Combining
these bounds yields $\Pr{Z \geq 1} = \Omega(1/\log^3 n)$, and a straightforward
division into groups $\mathcal{G}_1$ and $\mathcal{G}_2$ of sizes $\Theta(\log^3 n)$ and $n -
|\mathcal{G}_1|$ shows that all random walks in $\mathcal{G}_2$ can be eliminated in $O(\thit)$
steps.

The proof of the second reduction result (\ifEA\autoref{thm:lastintro} (Part 2)\else\autoref{thm:keylemma}\fi) is more
involved, although it again revolves around \eq{central}. The issue is that we
can no longer repeat the simple bucketing argument\ifEA~described above\else~from \autoref{thm:mostgeneral}\fi~about the degree distribution, since the number of
buckets may vastly exceed the number of walks. Furthermore, we may no longer
obtain ``w.h.p.''-bounds on the probability for certain good events. For all
these reasons, a refined approach is needed.

Our analysis allocates small phases of length $O(\thit/\kappa)$ in order to
halve the number of random walks, where $k=\kappa^{c}$ is the number
of walks at the beginning of the phase, for some suitably large constant $c$. The
first step is to show that starting from any vertex, there exists a large set
of vertices, so that each vertex is visited the ``right'' amount of time, but
also that it was not too unexpected to visit that vertex.  The latter condition
is quite subtle, but it allows us to arrange a proper scheduling of the walks
to show that, regardless of which vertices the random walk $i$ decides to visit
in that set, there are enough walks that are able to reach these vertices by
then. In other words, it rules out the possibility that, despite two random
walks visiting the same set of vertices, they never collide (for an
illustration, see \autoref{fig:ballsandbins} on page
\pageref{fig:ballsandbins}\ifEA~of the full paper\fi).  Using our concentration bounds with a careful
choice of the slackness parameters in terms of $\kappa$, the above approach can
eventually be shown to reduce the number of random walks  $k$ by a constant
fraction within $O(\thit/\kappa)$ steps. Repeating this iteratively yields the
bound $O(\thit)$. 

\subsubsection*{Bounds on $\thit$ and Worst-Case Bounds} 

With the two reduction results, \autoref{thm:hittingtime} follows immediately.
Furthermore, the aforementioned results can be also used to derive worst-case
upper and lower bounds on meeting and coalescing time on general and regular
graphs that are tight up to constant factors. Some of these were known, or
follow directly from existing results, the others are novel to the best of our
knowledge.
\tnote{What is the connection between the two paragraphs. Isn't the next sentence just making the previous sentence more precise??}

We proceed by establishing that $\tcoal = O(n^3)$ on all graphs. The proof of
$\tcoal=O(n^3)$ (\autoref{thm:graphclasses}) follows by first applying both
reductions (\ifEA\autoref{thm:lastintro} (Parts 1 \& 2)\else\autoref{thm:mostgeneral} and \autoref{thm:keylemma}\fi) to reduce the
number of walks from $n$ to $(\Delta/d)^{O(1)}\leq (n^2/|E|)^{O(1)}$  in
$O(\thit)$.  We have, by \autoref{pro:relatingmeetandhit}\ifEA~(see full paper)\fi, $\tmeet \leq 4 \thit
= O(n \cdot |E|)$, where this last bound follows from \cite{AKLLR79}.

Finally, combining the bound $\tmeet  = O(n \cdot |E|)$ together
with $\tcoal(S_0) =  O(\tmeet \cdot \log(|S_0|))$ (\autoref{lem:beer}\ifEA~of the full paper\fi) for any
set of start vertices $S_0$, yields that after additional
\[
   O(\tmeet \cdot \log(|S_0|)) = O(n \cdot |E| \cdot \log( n^2/|E|)) = O(n^3)
\]
steps the coalescing terminates. 
The fact that this is tight can be easily
verified by considering the Barbell graph.%
\footnote{This $n$-vertex graph is constructed by taking two cliques of size
$n/4$ each, and connecting them through a path of length $n/2$.} 

For regular graphs, the same argument as before shows that $\tcoal = O(n^2)$,
and this is matched by the cycle, for instance. The proofs of the other results
are straightforward, and we refer the reader to \autoref{sec:trivialproof}\ifEA~of the full paper\fi.

\subsubsection*{Bounds on $\tmeet$ and Other Results}

In \autoref{sec:meeting}\ifEA~of the full paper\fi, we derive several bounds on $\tmeet$.  These bounds
are derived more directly by \eqref{eq:central} and/or (\ref{eq:simp}), and
involve other quantities such as $\| \pi \|_2^2$ or the eigenvalue gap
$1-\lambda_2$.  One important technical contribution is to combine routine
spectral methods involving the spectral representation and fundamental matrices
that have been used in previous works, \eg Cooper et al.~\cite{CEOR13} with
some short-time bounds on the $t$-th step probabilities. This allows us to
improve several bounds, not only on $\tmeet$ and $\tcoal$ but also $\thit$ and
$\tcov$, by significantly reducing the dependency on the spectral gap or mixing
time---by almost a square root factor. As a corollary, we also derive a new
bound on the cover time for regular graphs that considerably improves over the best known
bound by Broder and Karlin~\cite{BK89} from 1989.

\subsubsection*{Concrete Topologies}

Finally, in \autoref{sec:special}\ifEA~of the full paper\fi, we apply the derived upper and lower bounds on
$\tmeet$ and $\tcoal$ on various fundamental topologies including grids, expanders and hypercubes. In most cases, these
results follow immediately from the general bounds by plugging in corresponding
values for $\|\pi\|_2^2$, $\thit$ or $\tmix$. One exception is the binary tree,
for which it seems surprisingly non-trivial to derive a lower bound of
$\tmeet=\Omega(n \log n)$. Here again we use a refinement of
(\ref{eq:central}) that restricts the vertices to leaf-nodes $u$, for which
$\sum_{t=1}^{\tmix} (\sum_{v \in V} p_{u,v}^t )^2 = \Omega(\log n)$. The matching upper bound $\tmeet=O(n \log n)$ follows from $\tcoal = O(\thit)$ for almost-regular graphs (\autoref{thm:hittingtime}).

Of particular interest might be the analysis of ``real-world'' graph models\ifEA\else~given in~\autoref{sec:realworld}\fi. There we show how to utilize our bounds from earlier sections to establish $\tcoal = \Theta(\tmeet)$ on two random graph models, leading to bounds on $\tcoal$ that are sublinear in the number of vertices.

\subsection{Discussion and Future Work}	

In this work we derived several novel bounds on $\tcoal$. Our first main result
implies that a gap of just $\Omega(\log^2 n)$ between $\tmix$ and $\tmeet$ is
sufficient to have $\tcoal=\Theta(\tmeet)$. We also proved that this result is
essentially tight. Further, we derived several new bounds on $\tcoal$ based on
$\thit$. For almost-regular-graphs, our new result implies the following hierarchy for the discrete-time setting,
\[
  \tmeet \leq \tcoal = O(\thit),
\] 
which refines the already known result $\tmeet = O(\thit)$. Finally, we also determined tight worst-case lower and upper bound for $\tcoal$.

For future work, an obvious problem is to extend the $\tcoal=O(\thit)$ result
to all graphs (so far, we only know $\tcoal=O(\thit \cdot \log \log n)$). Even
more ambitious would be to try to prove that the continuous-time variant and the
discrete-time process are (asymptotically) equivalent, as this would immediately resolve the $\tcoal=O(\thit)$
problem. A different direction may be to further explore lower bounds on
$\tmeet$\ifEA.\else; in this work we only derived one lower bound on $\tmeet$ in
\autoref{thm:meetingtime}.\fi

\ifEA
\else

\section{Notation and Preliminaries}\label{sec:notation}

Throughout the paper, let $G = (V, E)$ denote an undirected, connected graph
with $|V| = n$ and $|E| = m$. For a node $u \in V$, $\deg(u)$ denotes the
degree of $u$ and  $N(u)=\{ v \colon (u,v) \in E\}$ the \emph{neighborhood} of $u$. By $\Delta$, $\delta$ and $d=\frac{1}{n} \sum_{u \in V} \deg(u)$, we denote the maximum, minimum and average degree, respectively. We say $G$ is \irrgap if $\Delta/\delta = \Gamma$. 

Unless stated otherwise, all random walks are assumed to be discrete-time
(indexed by natural numbers) and lazy, \ie if $P$ denotes the $n \times n$
transition matrix of the random walk, $p_{u,u} = \ifrac{1}{2}$, $p_{u,v} = \ifrac{1}{(2\deg(u))}$ for any edge $(u,v) \in E$ and $p_{u,v}= 0$ otherwise.  We define $p_{u,v}^t$
to be the probability that a random walk starting at $u \in V $ is at node $v
\in V$ at time $t\in \naturals$.  Furthermore, let $p_{u,\cdot}^t$ be the
probability distribution of the random walk after $t$
time steps starting at $u$.
By $\pi$ we denote the \emph{stationary distribution}, which satisfies $\pi(u)=\ifrac{\deg(u)}{(2m)}$ for all $u\in V$.

Let $d(t)=\max_u \tvdist{p_{u,\cdot}^t - \pi} $ and $\bar d(t)=\max_{u,v}
\tvdist{p_{u,\cdot}^t - p_{v,\cdot}^t} $, where $\tvdist{\cdot}$ denotes the
total variation distance. Following Aldous and Fill~\cite{AF14}, we define the
\emph{mixing time} to be $\tmix(\varepsilon) = \min \{ t\geq 0 : 	\bar d(t)
\leq\varepsilon \}$ and for convenience we will write $\tmix=\tmix(1/e)$.  We define separation from stationarity to be $s(t) = \min\{ \varepsilon  \colon p_{u,v}^t \geq (1- \varepsilon) \pi(v) \mbox{ for all $u,v \in V$} \}$. Then $s(\cdot)$ is submultiplicative, so in particular, non-increasing~\cite{AF14}, and we can define the \emph{separation threshold time} $\tsep = \min \{ t\geq 0 : s(t) \leq e^{-1} \}$ and, by
\cite[Lemma 4.11]{AF14}, $\tsep \leq 4 \tmix$. 
We write $\hit(u,v)$ to denote the first time-step $t \geq 0$ at which a random walk starting at $u$ hits $v$. In particular, $\hit(u,u)=0$.
The \emph{hitting time} $\thit(u,v)=\E{\hit(u,v)}$ of any pair of nodes $u,v\in V$ is the expected time required for
a random walk starting at $u$ to hit $v$. Thus, $\thit(u,v)$ is the expectation of $\hit(u,v)$.
The hitting time of a graph $\thit=\max_{u,v}
\thit(u,v)$ is the maximum over all such pairs.

For $A \subseteq V$, we use
$\thit(u,A)$, to denote the expected time required for a random walk starting
to $u$ to hit some node in the set $A$.  Furthermore, we define
$\thit(\pi,u)=\sum_{v\in V}\thit(v,u)\cdot \pi(v)$.
Furthermore, we define $\tavghit= \sum_{u,v\in V} \pi(u)\cdot\pi(v)\cdot \thit(u,v)$.

Let $\tmeet(u,v)$ denote the expected time when two random walks starting at
$u$ and $v$ first arrive at the same node at the same time, and we write
$\tavgmeet$  for the expected meeting time of two random walks starting at two
independent samples from the stationary distribution. Finally, let $\tmeet =
\max_{u,v} \tmeet(u,v)$ denoted the worst-case expected meeting time.

We define the coalescence process as a stochastic process as follows: Let $S_0
\subseteq V  $ be the set of nodes for which there is initially one random walk on it, and for all $v\in S_{t}$ let 
\[ Y_v(t) =
\begin{cases}
	u \in N(v) & \mbox{w.p. $\frac{1}{2|N(v)|}$} \\
	v					& \mbox{w.p. $\frac{1}{2}$}	
\end{cases}
\]
The set of \emph{active} nodes in step $t+1$ is given by
$S_{t+1} = \{ Y_v(t) ~|~ v \in S_{t} \}$. 
The process satisfies the Markov property, \ie
\begin{equation}\label{eq:MarkovProp}
	\Pr{S_{t+1}~|~\mathcal{F}_t}=\Pr{S_{t+1}~|~S_{t}},
\end{equation}
where $\mathcal{F}_t$ is the filtration up to time $t$, which, informally speaking, is the history of all random decisions up to time $t$.
Finally, we define the \emph{time of coalescence} as $\coal(S_0)= \min \{t \geq 0~|
~ |S_t|=1\}$. 
%
Throughout this paper, the expression w.h.p. (\emph{with high probability})
means with probability at least $1 - n^{-\Omega(1)}$ and the expression w.c.p.
(\emph{with constant probability}) means with probability $c >0$ for some constant $c$.
We use $\log n$ for the natural logarithm.
\autoref{app:lazyrw} contains some known results about Markov Chains that
we frequently use in our proofs.

\section{Bounding $\tcoal$ for large $\tmeet/\tmix$} 
\label{sec:upperboundtcoal} 

In this section we prove~\autoref{thm:mixtradeoff}, one of our main results.
We refer the reader to \autoref{sec:technical} for a high-level description of
the proof ideas.

\subsection{Stochastic Process}\label{sec:process} 

In order to prove our first main result, it is helpful to consider a more
general stochastic process, $\pimortal$, called the $\newprocess$, involving
multiple independent random walks. In the $\newprocess$, whenever several
random walks arrive at the same node at the same time a subset of them (rather
than just one) may survive, while the remaining are merged with one of the
surviving walks. To identify the random walks, we assume that each walk has a
natural number (in $\naturals$) as an identifier. In order to define this
process formally, we introduce some additional notation and definitions; then
we state and prove some auxiliary lemmas.  A related  concept was introduced in
\cite[Section 3.4]{O12} under the name of ``allowed killings''.

\begin{figure}[!ht]

\begin{center}
\includegraphics[page=1]{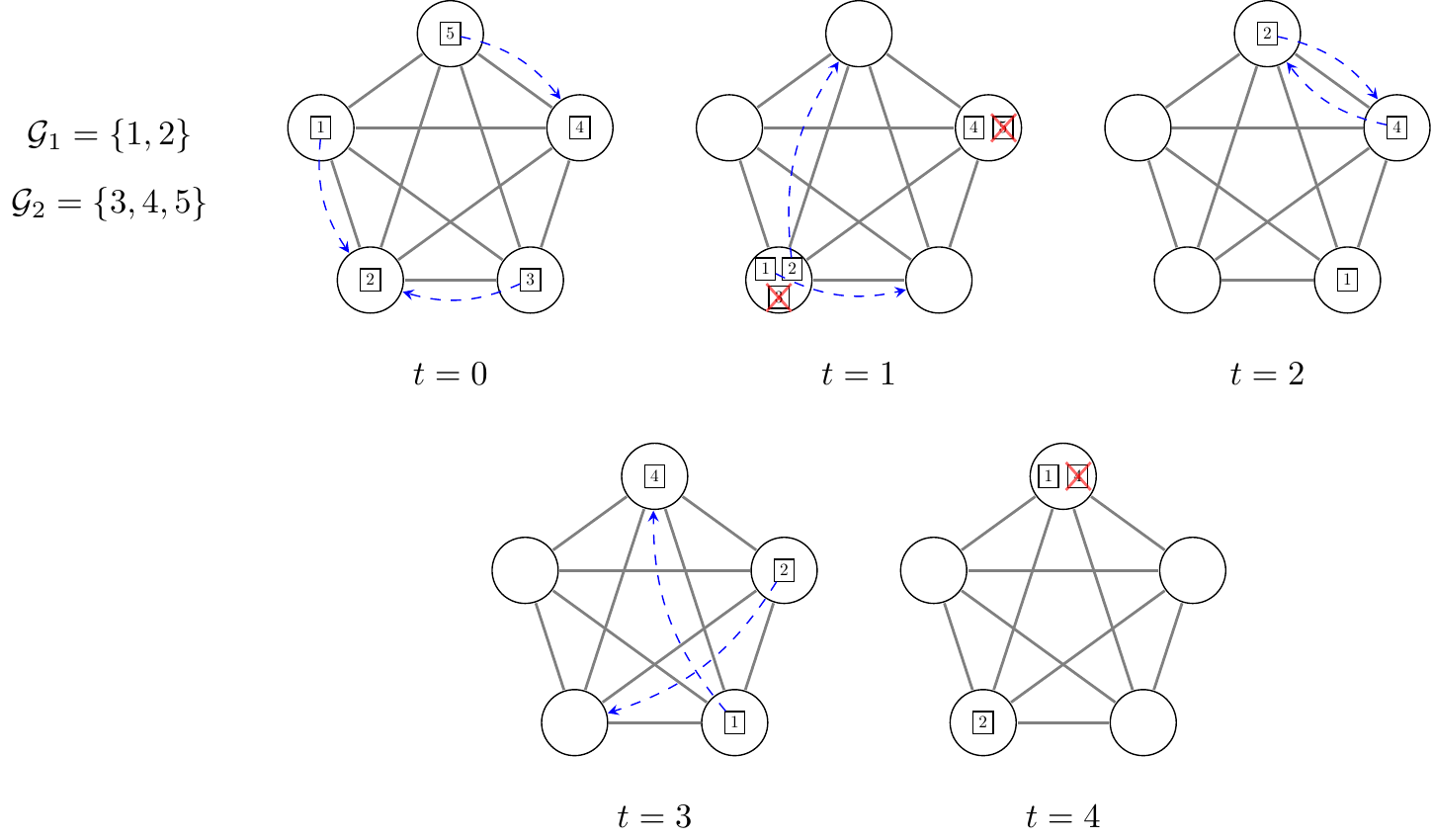}
\end{center}

\caption{Illustration of the process $\pimortal$.}\label{fig:selector}
\end{figure}

As mentioned before, we assume that every random walk $r$ has a unique
identifier $\id(r) \in \naturals$.%
\GNOTE{T: I understand that it makes sense here not too restrict the labels to
$1,2,\ldots,k$, but it looks a bit inconsistent to what we are doing later.} 
We divide the $\id$s into two groups $\mathcal{G}_1$, the group of immortal walks and
$\mathcal{G}_2$ the group of the remaining (mortal) walks. Whenever two or more walks
collide at a node and at least of of these walks is in $\mathcal{G}_1$, then all walks with $\id$s in $\mathcal{G}_1$ survive, while all walks
with $\id$s in $\mathcal{G}_2$ are killed (merged with some walk with $\id$ in $\mathcal{G}_1$).
Furthermore, if all walks have $\id$s in $\mathcal{G}_2$, \ie there are no walks with
$\id$ in $\mathcal{G}_1$, then the walk with the minimum $\id$ among these walks
survives. The $\id$s along with the assignment of $\id$s to groups determine
which of the random walks that arrive at a given node at the same time survive. 

Formally, let $\pimortal$ denote the following process: 
\begin{enumerate}
	\item At time $0$, $S_0= \{  (u_r, \id(r)) \}$, where $u_r$ is the starting
		node of random walk $r$ and $\id(r)$ is its identifier.
	\item At time $t$, several random walks may arrive at the same node. The
		process $\pimortal$ allows some subset of them to survive, while the rest
		`coalesce' with one of the surviving walks. Formally, $S_{t + 1}$ is
		defined using $S_t$ as follows.  Define the (random) next-step position
		of the random walk with $\id$ $i \in \naturals$ which is on node $v \in
		V$ to be		
		\[ Y_{v,i}(t)  :=
			\begin{cases}
				u~ ~\text{ where $u \in N(v)$ }& \mbox{w.p. $\frac{1}{2|N(v)|}$} \\
				v~ ~				& \mbox{w.p. $\frac{1}{2}$}	,
			\end{cases}
		\]
		Let $R_v(t) := \{ ( Y_{v,i}(t),i) ~|~ (v,i) \in S_t  \}, v\in V$ be the set
		of next-step positions (before merging happens) for random walks that
		were at node $v$ at time $t$. Let 
		\[\hat{R}_v(t) := \{ (v, i) ~|~ \exists
		u \in V, (v, i) \in R_u(t) \}\]
		 be the random walks that have arrived at
		node $v$ at time-step $t + 1$, just before merging happens. Then, merging
		happens w.r.t. the $\id$s as follows: 
		\begin{enumerate}
			\item If there exists $i \in \mathcal{G}_1$ such that $(v, i) \in \hat{R}_v(t)$ (at least one walk with $\id$ in $\mathcal{G}_1$ arrives at $v$), then 
				\[ S_{v}(t+1) := \{ (v, j) ~|~ (v, j) \in \hat{R}_v(t), j \in \mathcal{G}_1 \} \]
			\item If there is no $i \in \mathcal{G}_1$, such that $(v, i) \in \hat{R}_v(t)$ and $\hat{R}_v(t) \neq \emptyset$ (no walk with $\id$ in $\mathcal{G}_1$ arrives at $v$, but at least one walk arrives at $v$), then
				\[ S_{v}(t+1) := \{ (v, j) \}, \]
				where $j = \min \{ i~|~ (v, i) \in \hat{R}_v(t) \}$. 
			\item Otherwise, $S_{v}(t+1) := \emptyset$, \ie no walk arrived at $v$.
		\end{enumerate}
Finally, let
		\[ S_{t + 1} := \bigcup_{v \in V}  S_{v}(t+1) . \]
\end{enumerate}

We now relate this more general process, $\pimortal$, to the coalescing process
defined in \autoref{sec:notation}.  Let $P$ be regarded as a special instance
of $\pimortal$ with $\mathcal{G}_1=\{ 1\}$.  In process $P$, only one of several walks
arriving at the same node survives and by convention the one having the
smallest $\id$ is chosen. Let $(S_t)_{t = 0}^{\infty}$
denote the stochastic process $P$. If we define $\bar{S}_t := \{ v ~|~ (v, i)
\in S_t \}$, then $(\bar{S}_t)_{t = 0}^{\infty}$ is a coalescence process as
defined in \autoref{sec:notation}. Moreover, $P$ represented by $(S_t)_{t =
0}^{\infty}$  is the coalescence process which additionally keeps track of the
$\id$s. Throughout this paper, we assume that every random walk of $S_0$ is on
a distinct node. 

In the following we show that the time it takes to reduce to $k$ random walks in the original process $P$ is majorized by the time it takes in $\pimortal$ to reduce to $k$ random walks.
While this might be intuitive, one needs to be very careful about the dependencies between the meetings of different random walks: For instance a random walk which is 
immortal in $\pimortal$ might eliminate many other random walks whereas the corresponding coupled random walk  in $P$ might be eliminated early and therefore cannot eliminate said random walks. \fnote{feel free to improve this}

\begin{proposition}\label{lem:majormotal}
	Consider the following two processes:
	\begin{enumerate}
		\item Process $P$ is the standard process of coalescing random walks,
			viewed as a special case of $\pimortal$ with $\mathcal{G}_1 = \{1 \}$ as
			described above.
		\item Process $\pimortal$ is the process defined above using groups $\mathcal{G}_1$
			and $\mathcal{G}_2$, where $1 \in \mathcal{G}_1$. 
	\end{enumerate}
	Let $T^k$, $\Timortal^k$ be the stopping times given by the condition that
	fewer than $k$ random walks remain for the two processes respectively.
	Assume both processes start with the same initial configuration, \ie the
	vertices occupied by walks in both processes are identical and  there is
	only one walk per vertex in either process.  Then, there exists a coupling
	such that \GNOTE{T: Do we want to add that this holds with probability $1$? I think if you don't write that your coupling only holds w.p. p, then it's understood that it holds w.p. 1}
	\[ T^k \leq \Timortal^k . \]
\end{proposition}
\begin{proof}
\fnote{reviewers asks us to write why $T^k \leq \Timortal^k $ is not obvious}
	We will give a coupling between the moves of walks in $\pimortal$ and
	$\pinter$, a new process that is essentially intermediate between $P$ and
	$\pimortal$; furthermore, we will show that the original process $P$ is
	essentially a restricted view of the process $\pinter$.  The process
	$\pinter$ will label the walks \emph{dead}, \emph{alive}, and
	\emph{phantom}.  We emphasize that a phantom walk is not considered alive.  Note that the
	processes $P$ and $\pimortal$ can be viewed as processes which assign labels
	to each random walk of the type alive and dead. 
	
	Let $S^{\operatorname{Q}}_t$ denote the set of tuples of alive  walks in
	process $Q\in \{P, \pinter, \pimortal\}$ at time $t$.  Let
	$\bar{S}^{\operatorname{Q}}_t = \{ v ~|~ (v, i) \in S^{\operatorname{Q}}_t
	\}$ for $Q\in \{P, \pinter, \pimortal\}$ be the set of nodes which are
	occupied by at least one alive walk (there might be several in $\pimortal$
	at $t\geq 1$). In order to prove the proposition, we show that there exists
	a coupling, such that for any $t\in \naturals$
	\begin{equation}\label{ABN14r}
		\bar{S}^{\operatorname{P}}_t \subseteq  \bar{S}^{\pinter}_t
	\end{equation}
	\begin{equation}\label{BYD}
		\bar{S}^{\pinter}_t \subseteq  \bar{S}^{\pimortal}_t 
	\end{equation} 
	implying that $|\bar{S}^{\operatorname{P}}_t| \leq |\bar{S}^{\pimortal}_t|
	$ which yields the claim since
	\[ T^k = \min \{ t\geq 0 \colon |\bar{S}^{P}_t| \leq k  \} \leq \min \{
		t\geq 0 \colon |\bar{S}^{\pimortal}_t| \leq k  \}= \Timortal^k. \]
	
	We now define $\pinter$.  As mentioned above, the walks in $\pinter$ will be
	given three kinds of labels alive, dead, or phantom; the dead walks  do not
	continue ahead in time; alive and phantom walks do. 

	Formally, $\pinter$ using the groups $\mathcal{G}_1$ and $\mathcal{G}_2$ is defined as follows.
	We say that walk $r$ is of type $\mathcal{G}_i$, if $\id(r) \in \mathcal{G}_i$ for $i \in \{1,
	2\}$.  Whenever at least one walk arrives\footnote{Throughout, by arrive we
	take into account that walks may arrive at a node from the same node through
	laziness.} on a node, then the following happens.
	\begin{enumerate}
				\item  At least one of the walks is  of type $\mathcal{G}_1$
						\begin{enumerate}

		\item At least one walk of type $\mathcal{G}_1$ is alive
		\begin{enumerate}
		\item the walk of $\mathcal{G}_1$ with the smallest $\id$ is labeled as alive (even if it was labeled phantom before)
		\item all other walks of  type $\mathcal{G}_1$ (if there are any) are labeled as phantom
		\item alive  walks of type $\mathcal{G}_2$ are labeled dead (if present).  
		\end{enumerate}
		
		\item All walks of type $\mathcal{G}_1$ are phantom walks
		\begin{enumerate}
			\item 
		 There is no walk of type $\mathcal{G}_2$
		\begin{enumerate}
		\item No label is changed
		\end{enumerate}

\item There is at least one walk  of  type $\mathcal{G}_2$
		\begin{enumerate}
		\item the walk of type $\mathcal{G}_1$ with the smallest $\id$ is labeled as alive 
		\item all other walks of type $\mathcal{G}_1$ (if there are any) are labeled as phantom
		\item alive  walks of type $\mathcal{G}_2$ are labeled dead.  
		\end{enumerate}

\end{enumerate}
					\end{enumerate}

		\item All walks are of type $\mathcal{G}_2$
		\begin{enumerate}
		\item the walk of $\mathcal{G}_2$ with the smallest $\id$ is labeled as alive 
				\item all other walks  are labeled as dead.
		\end{enumerate}
						\end{enumerate}
	Note, that walks of $\mathcal{G}_1$ are either alive or phantom and walks of $\mathcal{G}_2$ are
	either alive or dead. Also, note that in the process $\pinter$, there is at
	most one \emph{alive} walk at any given node. Throughout the proof we regard
	the processes in two stages: First, each random walk selects a destination
	(possibly the same node it was on) and moves there.  In the second phase the
	walks are merged according to the process. See \autoref{fig:the3processes} for an illustration.

\begin{figure}
	\centering
	\includegraphics[width=0.5\textwidth]{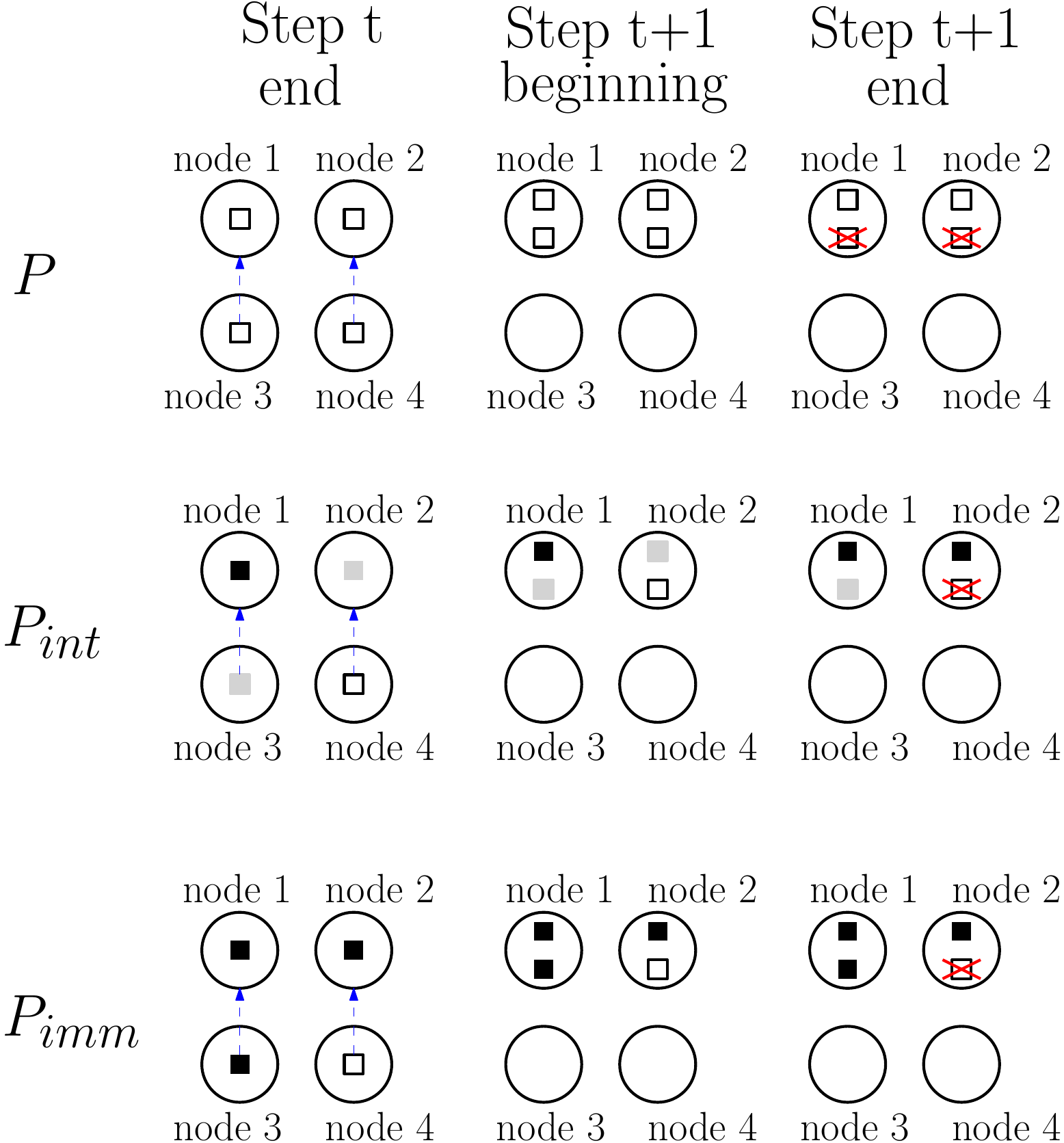}
	\caption{An illustration of couplings between the processes.  The squares depict the random walks. Walks of $\mathcal{G}_1$ are colored black and grey (phantom) and the nodes of $\mathcal{G}_2$ are white. The blue arrows denote the moving decisions. Observe that in $\pinter$ a phantom becomes alive (and a walk of $\mathcal{G}_2$ is labeled dead).    }\label{fig:the3processes}
\end{figure}

	We prove \eqref{ABN14r} by induction on $t$ starting from the same initial
	configuration: if $v\in \bar{S}^{\operatorname{P}}_t $, then $v\in
	\bar{S}^{\pinter}_t $. Consider the inductive step from $t$ to $t+1$ and
	assume that the claim holds at the end of round $t$ (after merging
	happened). For the (unique) random walk at $v\in
	\bar{S}^{\operatorname{P}}_t$ under process $P$, we couple its transition to
	node $Y_v({t+1})$ (where we possibly have $Y_v({t+1})=v$) with the corresponding
	alive walk of $\bar{S}^{\pinter}_t $ (there might be several walks of
	$\pinter$, however only one is alive and we couple with this alive walk).
	Let $S$ be the set of nodes to which a random walk in $P$ moved, \ie $S=\{
		Y_v({t+1}) \colon v \in \bar{S}^{\operatorname{P}}_t  \}$. Observe, that
	before the merging takes place in round $t+1$ (but moves have been made),
	there is, by induction hypothesis and the coupling, at least one alive walk
	of $\pinter$ on each node of $S$. 
	Furthermore, the definition of $\pinter$ ensures that whenever an alive
	random walk moves to a node, then after merging takes place, at
	least\footnote{By definition, there is actually exactly one alive walk.} one
	alive walk remains.  Thus, our coupling ensures that if $v\in
	\bar{S}^{\operatorname{P}}_{t+1} $, then $v\in \bar{S}^{\pinter}_{t+1} $.
	In words, if one looks at the subsets where there is an alive walk of
	$\pinter$, this is essentially the standard coalescence process.
	%
	This finishes the proof of \eqref{ABN14r} and we turn to proving \eqref{BYD}.
	
	When starting from the same initial configuration, we will provide a
	coupling that satisfies the following invariants.
	\begin{enumerate}
		\item There is a bijective map from the alive and phantom walks of $\pinter$
			to the alive walks of $\pimortal$, such that the following holds.
			All  walks of $\pinter$ of type $\mathcal{G}_i$ are mapped to walks of $\pimortal$ of type $\mathcal{G}_i$, for $i\in \{ 1, 2\}$.
			\item
			Whenever a walk of type $\mathcal{G}_2$ is labeled dead in $\pimortal$, then
			it is also labeled dead in $\pinter$ and vice versa.
	\end{enumerate}

	At the beginning there are no dead or phantom walks in $\pinter$, there are
	no dead walks in $\pimortal$, all walks are alive and as the starting
	positions in $\pimortal$ and $\pinter$ are the same, an arbitrary bijective
	mapping may be chosen, so long as it respects node positions and walk types. 
	
	Assume the invariant holds at time $t$. We take one random walk step for
	each alive or phantom random walk in $\pinter$. These are coupled with the
	corresponding walks in $\pimortal$, under the chosen map. Walks that are
	already dead are neither simulated  in $\pinter$ nor in $\pimortal$.  Hence,
	we can ensure the bijection between the walks of $\mathcal{G}_1$ in both processes
	holds at time $t+1$.
	
	We now prove the second invariant. Note that whenever a walk $r$ of type
	$\mathcal{G}_2$ in $\pimortal$($\pinter$) is labeled dead, this implies there must have been
	another walk $r'$ on the same node at the same time.  Since there is a
	bijective map, $r'$ must be on the same node in $\pinter$($\pimortal$).  We
	have that either $r'$ is of type $\mathcal{G}_1$ or $r'$ is of type  $\mathcal{G}_2$ and that
	$\id(r^\prime) < \id(r)$. In either case, $r$ is also killed (labeled dead) in
	$\pinter$($\pimortal$). Hence, we can ensure the bijection between the walks
	of $\mathcal{G}_2$ in both processes holds at time $t+1$.  Thus, the invariant holds
	at time $t+1$.  By induction, and since the alive walks of $\pinter$ are a
	subset of the alive walks of $\pimortal$ the invariant holds throughout the
	process and yielding \eqref{BYD}. 
	This  finishes the proof.
\end{proof}

\subsection{Meeting Time Distribution Prior to $\tmeet$ }

Let $(X_t)_{t\geq 0}$ and $(Y_t)_{t\geq 0}$ be independent random walks
starting at arbitrary positions. For $\tau$ a multiple of $\tmix$, the
following lemma gives a lower bound on the probability of intersection of the
two random walks in $\tau$ steps.

\begin{lemma} \label{lem:ape}
	Let $(X_t)_{t\geq 0}$ and $(Y_t)_{t\geq 0}$ be two independent random walks
	starting at arbitrary positions. Let $\cross(X_{t },Y_{t}, \tau)$ be the
	event that there exists $0 \leq s \leq \tau$, such that $X_s = Y_s$.  Then
	 \[ \Pr{\cross(X_{t},Y_{t}, 5 \tmix )} \geq \frac{1}{32\alpha},
	 \]
	where $\alpha = \tmeet/\tmix$. Furthermore, there exists a constant $c > 0$, such that for any $1 \leq b \leq \frac{e -
	1}{e} \cdot \alpha$, we have
	\[ \Pr{\cross(X_{t},Y_{t}, c \, b \, \tmix )} \geq \frac{b}{\alpha},
	\]
\end{lemma}
\begin{proof}
	First, let $(\tilde{X}_t)_{t \geq 0}$ and $(\tilde{Y}_t)_{t \geq 0}$ be two
	random walks that start from two independent samples drawn from the stationary
	distribution and are run for $\ell := 2 \lceil \alpha \rceil \lceil\tmix\rceil$ steps.
	Notice that $\ell \geq 2 \tmeet$, and hence, by Markov's inequality,

	\begin{align}
		\Pr{\cross(\tilde{X}_{t},\tilde{Y}_{t}, \ell)} \geq \frac{1}{2}.
		\label{eq:nike}
	\end{align}
	Furthermore, if we divide the interval $[1,\ell]$ into $2 \lceil \alpha
	\rceil$ consecutive sections of length $\lceil \tmix \rceil$ each
	, the probability for
	a collision in each of these section is identical and therefore the union
	bound implies
	\begin{align}
		\Pr{\cross(\tilde{X}_{t},\tilde{Y}_{t}, \ell)} \leq 2 \lceil \alpha
		\rceil \cdot \Pr{\cross(\tilde{X}_{t},\tilde{Y}_{t}, \tmix)},
		\label{eq:puma}
	\end{align}
	and hence combining equation (\ref{eq:nike}) and (\ref{eq:puma}) yields
	\[
		\Pr{\cross(\tilde{X}_{t},\tilde{Y}_{t}, \tmix)} \geq \frac{1}{4 \lceil
		\alpha \rceil}.
	\] 
	Consider now two independent random walks $(X_t)_{t\geq 0}$ and
	$(Y_t)_{t\geq 0}$ starting at arbitrary positions. By applying
	\autoref{lem:randomwalkcoupling} to both walks, with probability at least
	$(1 - e^{-1})^2$ both $X_{4 \tmix}$ and $Y_{4 \tmix}$ are drawn
	independently from the stationary distribution since $4\tmix \geq \tsep$. Therefore, 
	\[
		\Pr{\cross(X_{t},Y_{t}, 5 \tmix)} \geq (1 - e^{-1})^2 \cdot
		\Pr{\cross(\tilde{X}_{t},\tilde{Y}_{t}, \tmix)} \geq (1 - e^{-1})^2 \cdot
		\frac{1}{4 \lceil \alpha \rceil}.
	\] 
Observing that for any $\alpha \geq 1$, the RHS above
	expression is greater than $1/(32 \alpha)$ completes the proof of the first
	part. 	
	For the second part, we consider $k$ blocks of length $ 5 \tmix$.  
	Due to independence of different blocks, the probability of that the two walks meet in at least one of the $k$ blocks is at least $1-(1-\frac{1}{32\alpha})^k$.
		We set $k:= \ceil*{32b/(1-e^{-1})}$, $x := 1/(32\alpha)$.
	We distinguish between two cases.
	
	Case $k \cdot x < 1$:
		We use the fact that $(1 - x)^k \leq e^{-xk} \leq 1 -
	(1 - e^{-1})xk$ for $0 \leq x < 1$, $k \geq 0$ and  $xk \leq 1$.
	We
	derive that the probability of  intersecting after $k$ blocks is at least
	$1-(1-\frac{1}{32\alpha})^k \geq (1-e^{-1})k/(32\alpha)=b/\alpha$.
	
	Case $k \cdot x \geq 1$: We have
	$1-(1-\frac{1}{32\alpha})^k \geq 1-(1-\frac{1}{32\alpha})^{32 \alpha} \geq 1-1/e \geq b/\alpha $.
	In both cases the second part follows.
\end{proof}

At the heart of the proof of \autoref{thm:mixtradeoff} lies the following
lemma that analyses the marginal distribution of the meeting time
distribution.  That is, we only expose the first random walk
$(X_t)_{t=0}^{\tau}$, and look at how this affects the probability of
meeting. In essence, we show that at least one of the two ``orthogonal''
cases hold. In Case 1 (corresponding to set $C_1$), there is at least a
modest probability that after exposing $(X_t)$, $(Y_t)$ will intersect
with significant probability. Otherwise, in Case 2  (corresponding to set $C_2$), there is a
significant probability that after exposing $(X_t)$, $(Y_t)$ will
intersect with at least a modest probability.  
\fnote{reviewer: difficult to understand. In particular it is a bit strange that $Y_t$ sis a trajectory and not a vertex} 
\begin{lemma}\label{lem:classes2}
	Fix $\tau \in \mathbb{N}$ and a graph $G$. Let $(X_t)_{t = 0}^{\tau}$
	and $(Y_t)_{t = 0}^{\tau}$ be independent random walks, where the
	starting nodes $X_0$ and $Y_0$ are drawn independently from the
	stationary distribution $\pi$ (w.r.t. to $G$),
	\GNOTE{T: Why do make assumptions about the starting nodes, if it
	doesn't seem to be used in the proof?F: I tend to agree} 
	and the walks are run for $\tau$ steps. Let $p=\Pr{\cross(X_{t},Y_{t},
	\tau)}$ and let $ \mathcal{T}_\tau$ denote the set  all possible
	trajectories of a walk of length $\tau$ in $G$ (including possible
	self-loops).  We define the following two categories $C_1$ and $C_2$
	with  $C_1 \subseteq C_2$ 
	\GNOTE{T: I think notation would be more consistent with the remaining
	paper if $(z_0,\ldots,z_{\tau})$ would be $(x_0,\ldots,x_{\tau})$. F:
	I actually changed it to $z_i$ because I wanted to make sure that the
	reader doesn't confuse these trajectories with the walk $X_t$!}
	\begin{align*}
		C_1 &:= \{ (z_0, \ldots, z_{\tau}) \in \mathcal{T}_\tau \colon \Pr{
			\exists 0 \leq s \leq \tau, Y_s = z_s} \geq \sqrt{p} \} \\
		C_2 &:= \{ (z_0, \ldots, z_{\tau}) \in \mathcal{T}_\tau \colon \Pr{
			\exists 0 \leq s \leq \tau, Y_s = z_s} \geq {p}/{3} \}. 
	\end{align*}
	Then, $\Pr{ (X_t)_{t=0}^{\tau} \in C_1 } \geq \frac{p}{3}$ or $\Pr{
		(X_t)_{t=0}^{\tau} \in C_2} \geq \frac{\sqrt{p}}{3}$.
\end{lemma}

While the actual lower bounds on the probabilities appear rather crude, it turns out that the ``significant probability'' $\sqrt{p}/3$ is best possible, as we demonstrate in our lower bound construction later. Remarkably, the fact that the ``modest probability'' is only $p/3$ and much smaller than $\sqrt{p}/3$ does not affect the tightness of our bound, since in \autoref{claim:cat}, we can make up for this gap in both cases through a simple amplification argument over the unexposed random walks.

\begin{proof}
	Let us suppose  that $\Pr{(X_t)_{t=0}^\tau \in C_1 }
	< \frac{p}{3}$. We show that this implies $\Pr{(X_t)_{t=0}^\tau \in C_2} \geq
	\frac{\sqrt{p}}{3}$. Assume for the sake of contradiction $\Pr{(X_t)_{t=0}^\tau \in C_2} <
	\frac{\sqrt{p}}{3}$. We have 
	\begin{align*}
		p &= \Pr{\cross(X_t, Y_t, \tau)} \\
		&\leq \Pr{ (X_t)_{t=0}^\tau \in C_1} \cdot 1 + \Pr{ (X_t)_{t=0}^\tau \in
		(C_2 \setminus C_1) } \cdot \sqrt{p} + \Pr{ (X_t)_{t=0}^\tau \not\in C_2}
		\cdot \frac{p}{3} \\
		&< p/3 + \sqrt{p}/3 \cdot \sqrt{p} + p/3 \leq p,
	\end{align*}
	a contradiction. This completes the proof.
\end{proof}

It is well-known that starting with $k$ random walks, the coalescence
time is bounded by $O(\tmeet \log k)$, this can be deduced from the proof
presented in \cite{HP01}. For the sake of completeness, we give a
self-contained proof\footnote{One might be tempted to pair random walks in groups of two and run them for $2\tmeet$ time steps so that, by Markov inequality, they meet with probability at least $1/2$. Repeating this iteratively would yield the claim. To formalize such an argument one would need to disallow coalescence between different pairs of random walk which differs from the stochastic process we reduce to in \autoref{sec:process}. }.

\begin{proposition} \label{lem:beer} 
	We have 
	$\tcoal(S_0) = O(\tmeet \log|S_0|)$. 
\end{proposition}
\begin{proof}
	Let $P$ be the coalescing process (with $\id$s) defined in
	\autoref{sec:process}. Recall that $\mathcal{G}_1=\{ 1 \}$.  Let  $S_t$ be set
	of  coalescing random walks at an arbitrary time-step $t$.  In the
	following we show the slightly stronger claim  that the expected time
	to reduce the number of random walks by a constant factor is
	$O(\tmeet)$.

	Formally, we fix an arbitrary time-step $t_0$. With $T:= \min \{  t
	\geq t_0 \colon |S_t| \leq 99/100 \cdot |S_{t_0}|, |S_{t_0}|\geq 100
	\}$ denoting the first time-step the number of coalescing random walks
	reduces by a factor of $99/100$, we will prove that $\E{T}=O(\tmeet)$.
	Iterating the argument $O(\log|S_0|)$ times  implies that the expected
	time it takes to reduce to $100$ random walks is $ O(\tmeet
	\log|S_0|)$.  Note that the expected time to reduce from $100$ random
	walks to $1$ is bounded by $O(\tmeet)$. Hence, the claim $\tcoal(S_0)
	= O(\tmeet \log|S_0|)$ follows.

	It remains to show that the expected number of time steps it takes to
	reduce the number of random walks  by a factor of $99/100$ is indeed
	$O(\tmeet)$.

	We divide time into blocks of length  $\tau:=c\frac{e-1}{e}\tmeet
	+4\tmix$, where $c$ is the constant of \autoref{lem:ape}, \ie
	$\Pr{\cross(X_t, Y_t, c \frac{e - 1}{e} \tmeet)} \geq \frac{e -
	1}{e}$. We are primarily interested in what happens at the end of the
	blocks, \ie at time steps $t_0, t_0 + \tau, t_0 + 2 \tau, \ldots$. For
	simplicity, we will start counting time from $0$ at the beginning of
	each block.
	Let $(X_t)_{t \geq 0}$ be the random walk with $\id$ 1.  After
	$4\tmix$ steps, we can couple the state of the random walk
	$(X_t)_{t\geq  4\tmix}$ with a node drawn from $\pi$ with probability
	at least $(1 - e^{-1})$, since $4\tmix \geq \tsep$ (see
	\autoref{lem:randomwalkcoupling}). Further, note that conditioned on
	this coupling, the statement of \autoref{lem:classes2} implies that
	$(X_t)_{t\geq 4\tmix} \in C_2$ w.p. at least $p/3$, where we used $C_2
	\subseteq C_1$, and where $p := \Pr{\cross(\tilde{X}_t, \tilde{Y}_t, c
	\cdot \frac{e - 1}{e} \cdot \tmeet)} \geq \frac{e - 1}{e}$ for
	$\tilde{X}_0, \tilde{Y}_0 \sim \pi$. 

	We condition on the successful coupling of $X_{4 \tmix}$ with a node
	drawn from $\pi$ and that $(X_t)_{t \geq 4 \tmix} \in C_2$, which
	happens with probability at least $(1-e^{-1})p/3=\frac{(e - 1)^2}{3 e^2}$ (called
	event $\mathcal{E}$). Finally, consider any random walk $(Y_t)_{t\geq 0}b$ with $\id$ other
	than $1$. Again with probability at least $1 - e^{-1}$ we can couple
	$Y_{4 \tmix}$ with a node drawn from $\pi$ and conditioned on
	successful coupling, $(Y_t)_{t \geq 4\tmix}$ meets $(X_t)_{t \geq 4
	\tmix}$ between time-steps $[4 \tmix , \tau]$ with probability at
	least $p/3$, by definition of $C_2$. 
	Thus, conditioned on event $\mathcal{E}$, 
	each walk of $\mathcal{G}_2$ vanishes w.p. $(1-e^{-1})p/3=\frac{(e - 1)^2}{3 e^2}$ and thus
	the expected
	fraction of walks killed in the $\tau$ time-steps is at least
	$\frac{(e - 1)^2}{3 e}$. 

		Let $Z_\ell=|S_{t_0+\ell\cdot \tau}|$ denote the number of random
	walks alive at the beginning of block $\ell$.
	\begin{align*}
		\E{Z_{\ell} \, \mid \, \mathcal{F}_{t_0+(\ell-1)\cdot \tau}} &\leq
		Z_{\ell -1 } - (Z_{\ell-1}-1)\cdot \frac{(e-1)^4}{9e^4} \leq
		Z_{\ell - 1} -   \frac{Z_{\ell-1}}{100}.
	\end{align*}

	The above holds as long as $Z_{\ell - 1} \geq 100$. We can therefore
	apply \autoref{lem:drift} with parameters $g=99/100 \cdot S_0 $ and
	$\beta=99/100$ to obtain that $\E{T} = O\left(  \tau
	\right)=O(\tmeet)$, which completes the proof.
\end{proof}

\subsection{Upper Bound - Proof of \autoref{thm:mixtradeoff}}
\label{theProcess}

We commence by considering the process $\pimortal$ defined in
\autoref{sec:process}.  This allows us to establish  \autoref{claim:cat}
providing us with the following tradeoff.  For a given period $\tau$ of
length at least $\tmix$  we obtain a bound on the required number of
periods to reduce the number of random walks by an arbitrary factor.  The
proof relies heavily on \autoref{lem:classes2} which divides the walks of
$\mathcal{G}_1$ into two groups allowing us to expose the walks of $\mathcal{G}_1$ {first}
and {then} to calculate the probability of the walks of $\mathcal{G}_2$ to
intersect with them.  In fact, we will also use the symmetric case where
the roles of $\mathcal{G}_1$ and $\mathcal{G}_2$ are switched.  These probabilities are
derived from the time-probability tradeoff presented in
\autoref{lem:ape}.  We then use \autoref{claim:cat} to derive a bound on
the number of time steps it takes to reduce the number of walks to
$\ceil{2\alpha}$, where $\alpha=\tmeet/\tmix$ (\autoref{lem:firstpart}).
From there on we employ \autoref{claim:cat} to reduce the number of walks
to $1$ in $O(\tmeet)$ time steps.  Melding both phases together yields
the bound of \autoref{thm:mixtradeoff}.

We now define a process $\pimortal(S_0, k)$ with $k < |S_0|$, which is a parameterized version of the process $\pimortal$  defined in \autoref{sec:process}:
\begin{itemize}
	\item Let $|S_0| = k^\prime$; there are $k^\prime$ random walks with $\id$s
		$1, \ldots, k^\prime$ and starting nodes $v_1, \ldots, v_{k^\prime}$.
		Thus, $S_0 = \{(v_i, i) ~|~ 1 \leq i \leq k^\prime \}$.
	\item Let $\mathcal{G}_1 = \{1, \ldots, k \}$ and $\mathcal{G}_2 = \{k + 1, \ldots, k^\prime
		\}$. 
		Recall that, by definition of $\pimortal$, we have that if some random walks with $\id$s in $\mathcal{G}_1$ and some with $\id$s
		in $\mathcal{G}_2$ are present on the same node at the same time, only the ones
		with $\id$s in $\mathcal{G}_1$ survive. 
		If all the random walks have
		$\id$s in only in $\mathcal{G}_1$, then all of them survive. 
		If all random walks have $\id$s only in $\mathcal{G}_2$, then only the one with the smallest $\id$ survives.
\end{itemize}

We define \[\ids(S_t) := \{ \id(r) ~|~ (u_r, \id(r)) \in S_t \}, t\in
\naturals .\] The following lemma gives the expected time it takes to
reduce the number of random walks in $\mathcal{G}_2$ from $k' - k$ to some arbitrary
integer $g \geq k$: given a period of length $\tau$ and integer $g$,
assuming that $k=|\mathcal{G}_1|$ is large enough, we derive a bound on the number
of periods of length $\tau$ until the walks in $\mathcal{G}_2$ are reduced to $g$.
The required size of $k$ is a function of the probability for two random
walks drawn from $\pi$ intersecting after $\tau$ time steps.
\begin{claim} \label{claim:cat}
	Let $\tau \in \mathbb{N}$, let $(X_t)_{t = 0}^{\tau}$ and $(Y_t)_{t =
	0}^{\tau}$ be independent random walks run for $\tau$ steps, with $X_0$ and $Y_0$ drawn
	independently from $\pi$. Let $p_\tau \leq
	\Pr{\cross(X_{t},Y_{t}, \tau)}$ be a lower bound on the probability of the
	intersection of the two walks during the $\tau$ steps.
	Consider an instantiation of $\pimortal(S_0, k)$. Suppose that $k \geq \frac{3}{(1
	- e^{-1}) \cdot p_{\tau}}$.  For some $1 \leq g \leq |S_0| - k$, define the
	stopping condition $T_g = \min\{ t \geq 0 ~|~ |\ids(S_t) \cap \mathcal{G}_2| \leq g
	\}$. Then the expected stopping time satisfies
	\[ \E{T_g} = O \left( (4 \tmix + \tau) \cdot
	\sqrt{\frac{1}{p_{\tau}}} \cdot (\log |\mathcal{G}_2| - \log g)\right).
	\]
\end{claim}
	We first describe the high-level proof idea, before delving into the
	formal proof. We divide time into blocks of size  $ 4\tmix + \tau$.
	For any random walk $(Z_t)_{t=0}^{4\tmix +\tau}$ we can couple its
	position after $4\tmix \geq \tsep$  w.c.p. with a node drawn from
	$\pi$. Thus, conditioning on the success of this coupling we have, by
	\autoref{lem:classes2},   $\Pr{ (Z_t)_{t=4\tmix}^{4\tmix +\tau} \in
	C_1} \geq \frac{p_\tau}{3}$ or $\Pr{ (Z_t)_{t=4\tmix}^{4\tmix+\tau}
	\in C_2} \geq \frac{\sqrt{p_\tau}}{3}$.  In the former case we have
	that w.c.p. there is at least one random walk $r$ in $\mathcal{G}_1$ which is,
	due to independence of the walks, in class $C_1$. The hypothetical
	extension  of the trajectory of any random walk in $r' \in \mathcal{G}_2$
	intersects with $r$  w.p. $ c \sqrt{p_{\tau}}/3$, where the constant
	arises due to the fact that we also need to couple the state of $r'$
	at time $4 \tmix$ to a node drawn according to $\pi$. (We need to
	consider the hypothetical extension because the walk $r'$ may get
	eliminated sooner--this only helps us.) Thus, $r'$ gets eliminated
	w.p. at least $c\sqrt{p_{\tau}}$ for a suitable constant $c$.

	In the latter case we have that w.p. at least $c \sqrt{p_{\tau}}/3$ a
	random walks of $\mathcal{G}_2$ is in class $C_2$. Every random walk in that
	class intersects w.c.p. with at least one of the walks of $\mathcal{G}_1$.
	Thus, in both cases, we have that in each block a random walk of $\mathcal{G}_2$
	is eliminated w.p. a least $c\sqrt{p_\tau}$  for some constant $c$. 
Thus, the number of random walks in $\mathcal{G}_2$ decrease in expectation by a factor of $c\sqrt{p_\tau}$. 
\begin{proof}
	We will consider the process in \emph{blocks} each consisting of $ 4\tmix +
	\tau$ time-steps. For convenience in the proof, we'll restart counting
	time-steps from $0$ at the beginning of each block; we keep track of the
	total number of time-steps by counting the number of blocks. Let $C_1$ and
	$C_2$ be as defined in \autoref{lem:classes2}. Then we perform a case
	analysis by considering the two possible outcomes described in
	\autoref{lem:classes2} separately. 
	We define $Z_j = |\ids(S_{j \cdot (4 \tmix + \tau)}) \cap \mathcal{G}_2|$,
	\ie the number of walks remaining in $\mathcal{G}_2$ after $j$ blocks of time have
	passed. For any $j \geq 1$, we will show  that there exists a constant $c > 0 $ such
	that, 
	\[ 
		\E{Z_j ~|~ \F_{j - 1} } \leq Z_{j - 1} \cdot \left(1 - c \sqrt{p_\tau}
		\right). 
	\] 
	By using \autoref{lem:drift}, we get $\E{T_g} = O\left((4 \tmix + \tau)
	\cdot \frac{1}{\sqrt{p_{\tau}}} \cdot (\log |\mathcal{G}_2| - \log g)\right)$ (the
	factor $(4 \tmix + \tau)$ appears as the size of the block).  Recall
	that $\mathcal{F}_j$ is the filtration up to end of the $j$th block.
	In the remainder we show that we have indeed $\E{Z_j ~|~ \F_{j - 1} }
	\leq Z_{j - 1} \cdot \left(1 - c \sqrt{p_\tau} \right).$ 
	\medskip

	\noindent{\bf Case 1.} {\it $\Pr{ (X_t)_{t=0}^{\tau} \in C_1} \geq
	\frac{p_\tau}{3}$}: \smallskip \\
	Consider any random walk $r$ in $\mathcal{G}_1$ at the beginning of a \emph{block}.
	Using \autoref{lem:randomwalkcoupling}, after $4 \tmix$ steps we can couple
	the state of the random walk with a node drawn from $\pi$ with probability
	at least $(1 - e^{-1})$. Furthermore, conditioned on this coupling,
	the portion of the random walk between time-steps
	$4\tmix$ and $4\tmix + \tau$ of the walk is in class $C_1$ with probability
	at least $\frac{p_\tau}{3}$. Since $k \geq \frac{3}{p_\tau \cdot (1 - e^{-1})}$,
	w.p. $c_1>0$, in any block, there exists a walk in $\mathcal{G}_1$ that has the portion
	between time-steps $4 \tmix$ and $4 \tmix + \tau$ in $C_1$. 

	Fix a block and condition on the event that there is a walk in $\mathcal{G}_1$,
	denoted by $r_1$, whose portion between time-steps $4 \tmix$ and $4 \tmix +
	\tau$ is in $C_1$. Consider any walk in $\mathcal{G}_2$, denoted by $r_2$, at the
	beginning of the block.  We want to argue that this walk $r_2$ has a
	reasonable probability of intersecting some walk in $\mathcal{G}_1$ in this block of
	time-steps. First, consider (the possibly hypothetical continuation of $r_2$ ) walk $r_2^\prime$
	for the entire length of the block. The reason for this is that if $r_2$ and
	some walk from $\mathcal{G}_1$ are at the same node at the same time sometime in the
	block, $r_2$ will be eliminated in the process $\pimortal(S_0, k)$; however, we can
	consider its hypothetical extension to the entire length of the block.
	Using \autoref{lem:randomwalkcoupling} the state of the walk $r_2^\prime$ at
	time-step $4 \tmix$ can be coupled with a node drawn from $\pi$ with
	probability at least $c_2:=1 - e^{-1}$. Then conditioned on successful coupling,
	the probability that $r_2^\prime$ and $r_1$ collide during time-steps
	$4\tmix$ and $4\tmix + \tau$  is at least $\sqrt{p_\tau}$ (by definition of
	$C_1$ in \autoref{lem:classes2}). Thus, the probability that $r_2$ hits
	at least one
	walk in $\mathcal{G}_1$ is at least $c_1\cdot c_2\cdot \sqrt{p_\tau}$.
	Note that it is also possible for $r_2'$ to be eliminated by another walk from $\mathcal{G}_2$.
	In any case, we have that $r_2$ is eliminated w.p. at least $c \sqrt{p_\tau}$ and we get
		\[ \E{Z_j ~|~ \F_{j - 1} } \leq Z_{j - 1} \cdot \left(1 - c_1\cdot c_2 
	\sqrt{p_{\tau}} \right). \] 

	\medskip 
	\noindent{\bf Case 2.} {\it $\Pr{ (X_t)_{t=0}^{\tau} \in C_2} \geq
	\frac{\sqrt{p_\tau}}{3}$}: \smallskip \\
	Consider a walk in $\mathcal{G}_2$, denoted by $r_2$,
	at the beginning of a block; as in the previous case, we will consider a
	possibly hypothetical continuation $r_2^\prime$ of $r_2$. Using
	\autoref{lem:randomwalkcoupling} we can couple the state of $r^\prime_2$ at
	time-step $4\tmix$ with a node drawn from $\pi$ with probability at least $1
	- e^{-1}$. Furthermore, conditioned on the successful coupling, with
	probability at least $\frac{\sqrt{p_\tau}}{3}$ the trajectory of
	$r^\prime_2$ between the time-steps $4\tmix$ to $4 \tmix + \tau$ is in
	$C_2$.  Thus, with probability at least $p:=(1 - e^{-1})
	\frac{\sqrt{p_{\tau}}}{3}$, $r^\prime_2$ has a trajectory between time-steps
	$4 \tmix$ and $4 \tmix + \tau$ that lies in $C_2$.   Now consider any random walk $r_1 \in \mathcal{G}_1$ at the
	beginning of the block.  
	\fnote{we could possible get rid of the following if we change the $C_1$ $C_2$ lemma do work for any starting distribution}
	Again, using \autoref{lem:randomwalkcoupling} with
	probability at least $1 - e^{-1}$, we can couple the state of the random
	walk at time $4 \tmix$ with a node drawn from $\pi$. Conditioned on this
	between time-steps $4 \tmix$ to $4 \tmix + \tau$, this random walk hits any
	trajectory whose portion between time-steps $4 \tmix$ to $4 \tmix + \tau$
	lies in $C_2$ with probability at least $p_{\tau}/3$ (by definition of $C_2$
	in \autoref{lem:classes2}). Since $k = |\mathcal{G}_1| \geq \frac{3}{(1 - e^{-1})
	\cdot p_{\tau}}$, with at least constant probability $c_1 > 0$ there is some walk in
	$\mathcal{G}_1$ that intersects any fixed trajectory whose portion between time-steps
	$4 \tmix$ to $4 \tmix + \tau$ lies in $C_2$.  Since the random walks in
	$\mathcal{G}_1$ are independent, by the definition of the $\newprocess$, we have
	that any walk in $\mathcal{G}_2$ is eliminated by the end of the block with probability
	at least $c_1 \cdot p =c \sqrt{p_{\tau}}$ for some constant $c > 0$.
	Similarly as before, it is possible that $r_2$ is eliminated by at least one of the walks of $\mathcal{G}_2$, which only increases the probability for $r_2$ of being eliminated.
	We get
	\[ \E{Z_j ~|~ \F_{j - 1} } \leq Z_{j - 1} \cdot \left(1 - c
	\sqrt{p_{\tau}} \right). \] 
\end{proof}

In the following we bound the time $T$ required to reduce to
$2\ceil{\alpha}$ random walks.  The claim follows by applying
\autoref{claim:cat} to derive a bound on $\Timortal$ for
process$\pimortal$, and using the majorization of $T$ by $\Timortal$
(\autoref{lem:majormotal}).  

\begin{corollary} \label{lem:firstpart} 
	Consider the coalescence process starting with set $S_0$ and let $\alpha =
	\tmeet/\tmix$. Let $T_1 = \min \{ t \geq 0 ~|~ |S_t| \leq 2 \lceil \alpha
	\rceil \}$.  Then $\E{T_1} = O(\tmix \cdot \sqrt{\alpha} \cdot \log|S_0|)$. 
\end{corollary}
\begin{proof}
	We consider the process $P$ (defined in \autoref{sec:process}), which is
	identical to the coalescence process, but in addition also keeps track of
	$\id$s of random walks and that  allows only the walk with the smallest $\id$ to survive. We assume that the $\id$s are from the set $\{1, 2, \ldots, |S_0|
	\}$. Let  $S_0 = \{(v_1, 1),
	\ldots, (v_{|S_0|}, |S_0|) \}$ and $\bar{S}_0=\{ i \colon (v,i) \in S_0\}$.  We consider the process $\pimortal(S_0,
	k)$ and $k = \lceil \alpha \rceil$. Let $T^*_1$
	be the stopping time defined by $|\ids(\bar{S}_t) \cap \mathcal{G}_2| \leq \alpha$
	for the process $\pimortal(S_0, k)$. By definition of $\pimortal$ and \autoref{lem:majormotal}, it
	follows that $\Timortal$ stochastically dominates $T$. Thus, it suffices to
	bound $\E{\Timortal}$.
	W.l.o.g. we assume that $\alpha \geq 6\frac{e-1}{e}$, otherwise the claim follows directly from \autoref{lem:beer}.
	We apply \autoref{lem:ape} with $b = 6$ and derive that for some suitable constant $c$,
	\[ p=\Pr{\cross(X_{t\geq 0},Y_{t\geq 0}, 6 c \tmix)} \geq \frac{6}{\alpha},\] 
	Thus, we have 
	\[ \frac{3}{(1 - e^{-1}) \cdot p} \leq \frac{3}{\frac{1}{2} \cdot p} \leq \alpha \leq k \]
	Applying \autoref{claim:cat} with $g=\alpha$, $\tau=6c\tmix$ (where $c$ is a
	constant as given by \autoref{lem:ape}), $p_{\tau}=6/\alpha$, and observing
	that $k \geq \frac{3}{(1 - e^{-1}) \cdot p_{\tau}}$, we get the required
	result.  
	\end{proof}	

In the following we bound the time $T$ required to reduce from
$2\ceil{\alpha}$ random walks to a single random walk.  The proof uses
the same ideas as before (\autoref{lem:firstpart}) however, this time we
consider several phases and in each we reduce the number of random walks
by a constant factor. The expected time per phase is geometrically
increasing as the number of walks decreases and the overall time is
essentially dominated by the time for a constant number of random walks
to meet, which is $O(\tmeet)$. 

\begin{lemma} \label{lem:secondpart} 
	Consider the coalescence process starting with set $S_0$, satisfying $|S_0|
	\leq 4 \alpha \log \alpha$, where $\alpha = \tmeet/\tmix$.  Let $T_2 :=
	\min\{ t \geq 0 ~|~ |S_t| \leq 1 \}$. Then $\E{T_2} = O(\tmeet)$. 
\end{lemma}
\begin{proof}
	We will consider the coalescence process in phases. Let $\ell$ be the
	largest integer such that $|S_0| \geq \left(\frac{4}{3} \right)^{\ell}$. For
	$j \geq 1$, the $j\th$ phase ends when $|S_t| < \left(\frac{4}{3}
	\right)^{\ell - j + 1}$. The $(j + 1)\th$ phase begins as soon as the $j\th$
	phase ends. Note that it may be the case that some phases are empty. Let
	$T_2(j)$ denote the time for phase $j$ to last. We will only consider phases
	up to which $\ell - j + 1 \geq 32$.

	Now we focus on a particular phase $j$. Let $t_j$ be the time when the
	$j\th$ phase begins and let $S_{t_j}$ denote the corresponding set at that time. Thus, we have
	\begin{equation}\label{bbbiotech} \left(\frac{4}{3} \right)^{\ell - j + 1} \leq | S_{t_j} | < \left(\frac{4}{3} \right)^{\ell - j + 2} \end{equation}
	We consider the process $\pimortal$ defined in \autoref{theProcess} as
	follows. 
Define $n_j = |S_{t_j}|$.
Fix a phase $j$ and define 
	$S'_0 = \{ (v_1, 1), \ldots, (v_{n_j}, n_j) \}$ and
	$\bar{S'}_{0} = \{v_1, \ldots, v_{n_j} \}$. Then, consider again the set of occupied vertices (ignoring the labels) $\bar{S}_{t_j + t} = \{ v ~|~ \exists i \in \naturals, (v,
	i) \in S'_t \}$ with $t\in \naturals$. Thus, phase $j$ ends when $|S'_t| =|\bar{S}_{t_j + t} |< \left( \frac{4}{3}
	\right)^{\ell - j + 1}$. 	Let \[k_j := \left\lceil \frac{|S'_0|}{2} \right\rceil\] 
	be the size of $\mathcal{G}_1$ and consider the process
	$\pimortal(S'_0, k_j)$ as defined in \autoref{theProcess}. Let \[g_j
	:= \left\lfloor\frac{|S^\prime_0| - k_j}{3} \right\rfloor\] and 
	\[T^*_2(j) := \min\{ t ~|~ |\ids(S^\prime_t) \cap \mathcal{G}_2| \leq g_j \}.\]
	 We note that as long as
	$\ell - j + 1 \geq 32$, $g_j \geq 1$ and at time $T^*_2(j)$, 
	\[|S^\prime_t| \leq g_j +k_j \leq \frac{|S^\prime_0| - k_j}{3} +k_j =
	\frac{|S^\prime_0| }{3}  + \frac{2 k_j}{3} \leq  \frac{|S^\prime_0| }{3}
	+ \frac{|S^\prime_0| }{3}+\frac23 < \frac{3}{4} \cdot |S^\prime_0|.\] 
	By \autoref{lem:majormotal}, $T^*_2(j)$ stochastically dominates
	$T_2(j)$ and hence it suffices to bound $\E{T^*_2(j)}$.  In order to
	bound $\E{T^*_2(j)}$,  we define 
	\[ b_j := 32 \alpha \log(4/3)
	(\ell - j + 1) (3/4)^{\ell - j + 1} .\]
	Since we only consider phases with $j$ respecting $\ell - j + 1 \geq 32$ we have $b_j \leq b_{\ell - 31}
	\leq ((e-1)/e)\alpha$.
	Furthermore, we have $b_j \geq b_0 \geq 4\alpha \log \alpha (3/4)^\ell \geq 1$,
	where the last inequality follows from $(4/3)^\ell \leq |S_0| \leq 4\alpha \log \alpha,$ which in turn follows from definition of $\ell$ and the assumed bound on $|S_0|$.
	 Applying \autoref{lem:ape} with this value of $b_j$,
	we get that for 
	\[\tau_j := c b_j \tmix,\] for independent random walks
	$(X_t)_{t = 0}^{\tau_j}$, $(Y_t)_{t = 0}^{\tau_j}$,  $\Pr{\cross(X_t, Y_t,
	\tau_j) }\geq p_j$, where \[p_j := 32 \log(4/3)(\ell - j + 1) (3/4)^{\ell - j +
	1}.\] 
	We seek to apply \autoref{claim:cat} to bound $\E{T^*_2(j)}$. We first verify that the conditions of \autoref{claim:cat} are fulfilled. In particular, we 
	verify that $k_j \geq \frac{8 }{p_j}$; 
	 to see this consider the
	following: 
	\begin{align}
		\frac{8  }{p_j} 
		&=
		\frac{ 8  }{ 32 \log(4/3)(\ell - j + 1) }   (4/3)^{\ell - j +	1}\nonumber
\leq
		\frac{  1 }{ 4}  \cdot \left(
		\frac{4}{3} \right)^{\ell - j + 1} \nonumber
		\leq \frac{1}{2} \cdot |S^\prime_0| \leq k_j, \label{eqn:applyclaim}
	\end{align}
	where we used \eqref{bbbiotech} and $ |S^\prime_0|=| S_{t_j} |$ in the second-last inequality. 
	Thus we can apply \autoref{claim:cat} and derive  
	\[ \E{T^*_2(j)} \leq (\tau_j + 4\tmix)\cdot\frac{1}{\sqrt{p_j}}\cdot\left( \log|\ids(S^\prime_0) \cap \mathcal{G}_2|- \log g_j  \right) \]
	 and we continue by dissecting that bound.
	Since $b_j \geq 1$, there exists a suitably large constant $c_1$, so that $\tau_j + 4
	\tmix \leq c_1 b_j \tmix$.
	Furthermore, 
	\[ \frac{b_j}{\sqrt{p_j}} = \frac{32 \alpha \log(4/3)
	(\ell - j + 1) (3/4)^{\ell - j + 1}}{ \sqrt{32 \log(4/3)(\ell - j + 1) (3/4)^{\ell - j +
	1}}} =  O\left(  \alpha  \sqrt{\ell - j + 1 }  \cdot \left( \frac{3}{4} \right)^{(\ell - j + 1)/2}  \right).\] 
	Observe that, by definition, $|\ids(S^\prime_0) \cap \mathcal{G}_2| /g_j \leq 3$, hence $\log|\ids(S^\prime_0) \cap \mathcal{G}_2|- \log g_j \leq \log (3)$. Putting everything together, we get that there is a constant $c_2$ such that,
	\begin{align} 
		\E{T^*_2(j)} &\leq c_2 \cdot \tmix \cdot \alpha \cdot \sqrt{\ell - j + 1}
		\left( \frac{3}{4} \right)^{(\ell - j + 1)/2} \label{eqn:phase-j-bound}
	\end{align}

	Note that since we stop when $\ell - j + 1 < 32$, there are at most $\ell -
	30$ phases considered. Let $\tilde{T}$ be the random variable denoting the
	time-step when the last phase ends; at this point $|S_{\tilde{T}}| = O(1)$.
	Therefore, using \autoref{lem:beer}, $\E{T_2 - \tilde{T} \,|\,
	\tilde{T}} = O(\tmeet)$. But, clearly $\tilde{T}$ is stochastically
	dominated by $\sum_{j = 0}^{\ell - 30} T_2^*(j)$. Thus, we have
	\begin{align}
		\E{T_2} &= \E{\tilde{T}} + \E{\E{T_2 - \tilde{T} \,\mid\, \tilde{T}}}  \nonumber \\
		&\leq c_2 \cdot \tmix \cdot \alpha \sum_{j = 0}^{\ell - 30} \sqrt{ \ell - j  +
		1} \left(\frac{3}{4}\right)^{(\ell -j + 1)/2} + c_3 \tmeet
		\label{eqn:finalstep2} \\
		&\leq c_2 \cdot \tmix \cdot \alpha + c_3 \tmeet = O(\tmeet) \label{eqn:finalstep3}
	\end{align}
	Above, in~(\ref{eqn:finalstep2}) we used~(\ref{eqn:phase-j-bound}) and the
	fact that $\E{T_2 - \tilde{T} ~|~ \tilde{T}} \leq c_3 \tmeet$ for some
	constant $c_3 > 0$ and in step~(\ref{eqn:finalstep3}), we used the fact that
	$\sum_{j = 32}^{\infty} j c^j < 1$ for $c \leq \sqrt{3/4}$. 
\end{proof}

Thus, the first phase (\autoref{lem:firstpart}) and the second phase (\autoref{lem:secondpart}) take together $O(\sqrt{\alpha} \cdot \log n \cdot \tmix + \tmeet)$ time-steps, which yields
\autoref{thm:mixtradeoff}.

\subsection{Lower Bound - Proof of \autoref{thm:lowerboundgraph}} 
\label{sec:lowerbound}
	
In this section we give a construction of a graph family in order to establish lower bounds on $\tcoal(G)$ in terms of $\tmeet(G)$ and $\tmix(G)$ demonstrating that
 \autoref{thm:mixtradeoff} is asymptotically tight.
Additionally, our construction generalizes a claim of Aldous and Fill~\cite[Chapter 14]{AF14}:
 They mention that it is possible to construct regular graphs that mimic the $n$-star in the sense that the $\tmeet = o(\tavghit)$, without giving further details of the construction.
 Our construction shows that even the coalescence time can be significantly smaller than the average hitting time for almost-regular graphs.
 For our family of almost-regular graphs, there is a polynomial gap between $\tmeet$ and $\tavghit$. More importantly, we show that these almost-regular graphs have a gap  of $\sqrt{{\tmix}/{\tmeet}} \cdot \log n$ between coalescing and meeting time.
 This shows that the bound in \autoref{thm:mixtradeoff} is best possible, even if we constrain $G$ to be almost-regular.
We refer the reader to \autoref{sec:technical} for a high-level description  of the proof ideas.

More precisely, in the proof of \autoref{thm:lowerboundgraph} we shall give 
an explicit construction of a graph family $G=G_n$ with 
 $\tcoal=\Omega(\sqrt{\alpha_n}\cdot \log n \cdot \tmix)$, where $\alpha_n={\tmeet}/{\tmix}$.
For the remainder of this section, we will drop the dependence on $n$ and will simply use $G$ instead of $G_n$ and $\alpha$ instead $\alpha_n$. 


\begin{figure}
\centering
\includegraphics[page=2]{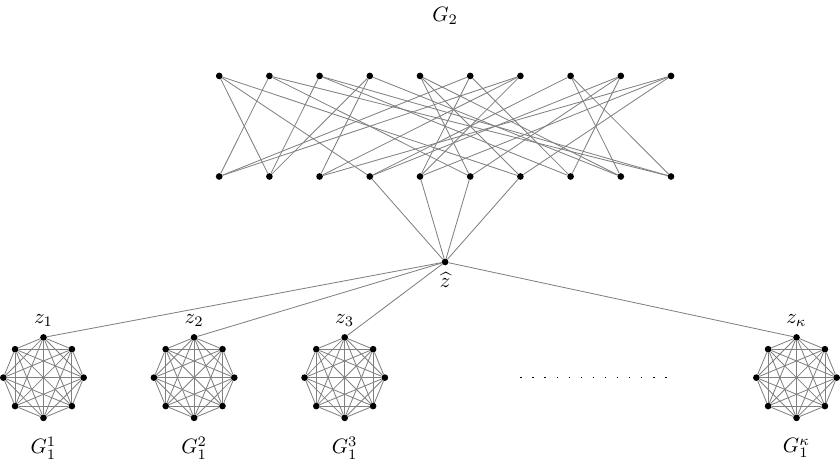}
\caption{The graph described in \autoref{sec:lowerbound}
with  $\tcoal = \Omega(\tmeet + \sqrt{\tmeet/\tmix}\cdot \log n \cdot \tmix)$.
}
\label{fig:holygraph}
\end{figure}

The construction of $G$ (see \autoref{fig:holygraph} for an illustration) is based on two building blocks, $G_1$ and $G_2$.
First, let $G_1=(V_1,E_1)$ be a clique of size $\sqrt{n}$.
	Let $G_2=(V_2,E_2)$ be a $\sqrt{n}$-regular bipartite Ramanujan Graph on $n/\sqrt{\alpha'}$ nodes \cite{MSS15}, where $\alpha' = \max\{\alpha,2^{20}\cdot C^2\}$, where $C>1$ is the universal constant of \autoref{cor:universialconstant}. The graph $G$ is made of one copy of $G_2$, $\kappa=\sqrt{n}$ copies of $G_1$ (denoted by $G_1^1, G_1^2, \dots, G_1^\kappa$), and a node $\widehat{z}$,
	which has an edge to  $\sqrt{n/\alpha'}$ distinct nodes of $G_2$ and to each of   
	the designated nodes $z^i \in V_1^i$ in  $G_1^i$  for $i\in[1,\kappa]$. It is not difficult to see that this graph is almost-regular, \ie maximum and minimum degree differ by at most a constant factor.

	In \autoref{lem:mix}, \autoref{lem:meet},  \autoref{lem:coal} and \autoref{lem:avghit} respectively we show that $\tmix = \Theta(n)  $, $\tmeet = \Theta(\alpha' n)$,  $\tcoal = \Omega(\sqrt{\alpha'}\cdot n  \log n)$, and 
	$\tavghit = \Omega(n^{3/2})$.
	We start with the following auxiliary lemma which shows that the walk restricted to $V_2$
	behaves similarly to the walk restricted to $V_2 \cup \{ \widehat{z}\}$, meaning that the walks have very similar $t$-step probabilities.

\begin{lemma}\label{lem:huehue}
Let $P$ denote the transition matrix of the random walk on $G$, $Q$ the transition matrix of the random walk on $G_2$
and $\widehat Q$ be the transition matrix of the random walk on the subgraph 
of $G$ induced by $V_2 \cup \{ \widehat{z} \}$. 
Let $p^t_{u,v},q^t_{u,v},\widehat q^t_{u,v}$ denote the corresponding transition probabilities for a walk starting at $u$ to end up at node $v$ after $t$ steps. Let $S^*= \{ u \in V_2 \cap N(\widehat{z})  \}$. 
Then the following statements hold:
\begin{enumerate}[(i)]
	\item For any $u, v\in V_2$ we have $\tvdist{p^t_{u,\cdot}-q^t_{u,\cdot} } \leq \sum_{i=1}^{t-1} p^i_{u,S^*}/(2\sqrt{n}) \leq  t/(2\sqrt{n})$.
	\item  For any $u, v\in V_2$ we have $\tvdist{ \widehat q^t_{u,\cdot}-q^t_{u,\cdot}}  \leq \sum_{i=1}^{t-1} p^i_{u,S^*}/(2\sqrt{n})  \leq t/(2\sqrt{n})$.
	\item 
For any $u,v\in V_2$ we  have that after $t=\tmix(G_2)$ time steps 
	$\tvdist{p^t_{u,\cdot} -p^t_{v,\cdot}  } \leq o(1) + 2/e.  $
\end{enumerate}
\end{lemma}
\begin{proof}

Let $(X_t)_{t\geq 0}$ be the Markov chain with transition matrix $P$ and let $(Y_t)_{t\geq 0}$ be the Markov chain with transition matrix $Q$. We will inductively couple these two random walks starting from $X_0=Y_0=u$.
Given that we coupled both chains up to time $t-1$, we can couple $(X_t,Y_t)$ such that $X_t=Y_t$ with an error probability
\begin{align*}
	\Pr{X_t \neq Y_t~|~X_{t-1}= Y_{t-1}}&=\Pr{X_t\neq Y_t~|~X_{t-1}= Y_{t-1}, X_{t-1}\in S^*} \cdot \Pr{ X_{t-1}\in S^*}\\
	&\phantom{000}+\Pr{X_t\neq Y_t~|~X_{t-1}= Y_{t-1}, X_{t-1} \in V_2\setminus S^*} \cdot \Pr{ X_{t-1}\in V_2\setminus S^*}\\
	&\leq p^{t-1}_{u,S^*}/(2\sqrt{n})+ 0.
\end{align*}
We have, by \cite[Proposition 4.7]{LPW06},
$
	\tvdist{p^t_{u,\cdot} -p^t_{v,\cdot}  } = \inf\{ \Pr{ X \neq Y} ~|~(X,Y)\text{ is  a coupling of $ p^t_{u,\cdot}$ and $p^t_{v,\cdot} $}\}.
$
Hence, by a union bound over $t$ steps,
\begin{align*}
	\tvdist{p^t_{u,\cdot} -p^t_{v,\cdot}  } &= \inf\{ \Pr{ X \neq Y} ~|~(X,Y)\text{ is  a coupling of $ p^t_{u,\cdot}$ and $p^t_{v,\cdot} $}\} \leq \Pr{X_t \neq Y_t} \\
	&\leq \sum_{i=1}^{t-1} p^i_{u ,S^*}/(2\sqrt{n}) \leq \frac{t}{2\sqrt{n}}.
\end{align*}	
To prove the second part we redefine
$(X_t)_{t\geq 0}$ to be the Markov chain with transition matrix $\widehat Q$ and the proof is identical.

We proceed with the last part.
For $u,v\in V_2$ we  have that after $t=\tmix(G_2)$ time steps, by the triangle inequality and using that $\tmix(G_2)=O(1)$, by \autoref{lem:BASF}, we get 
\begin{align*}
	\tvdist{ p^t_{u,\cdot} - \pi^{G_2}(\cdot) } &\leq  \tvdist{ p^t_{u,\cdot} - q^t_{u,\cdot} }  +  \tvdist{q^t_{u,\cdot} - \pi^{G_2}(\cdot) }  \\
	&\leq    \frac{\tmix(G_2)}{2\sqrt{n}}  + \tvdist{q^t_{u,\cdot} - \pi^{G_2}(\cdot) }\\
	&\leq o(1)  + \tvdist{q^t_{u,\cdot} - \pi^{G_2}(\cdot) } \leq o(1) + 1/e, 
\end{align*}
where the last inequality follows form the definition of mixing time.
	 Again, by the triangle inequality, 
	 $ \tvdist{ p^t_{u,\cdot} -p^t_{v,\cdot} } \leq o(1) + 2/e.  $
\end{proof}

Based on \autoref{lem:huehue}, we can now bound the hitting time to reach $\widehat{z}$, which will later be used to establish the bounds on the 
	mixing and meeting time of the whole graph $G$. But first, we prove that the mixing time of the graph $\widehat{G}$ induced by $V_2 \cup \{ \widehat{z}\}$ is constant and that after mixing on $\widehat{G}$, the random walk has a probability of $\Omega(1/n)$ to hit $\widehat{z}$ in a constant number of time steps.
	\begin{lemma}\label{lem:holyhit}
	 The following three statements hold.
			\begin{enumerate}[(i)]
			\item  Let $\widehat{G}$ be the induced graph by the vertices $V_2 \cup \{ \widehat{z}\}$. Then $\tmix(\widehat{G})=O(1)$.
				\item Let $u\in V \setminus \{ \widehat{z}\}$. Then there exists a constant $c \geq 1$ such that $\Pr{\hit(u,\widehat{z}) \geq n/c } \geq 1/2$. 

			\item Let $u\in V \setminus \{ \widehat{z}\}$. Then $   \thit(u,\widehat{z})=O({n})$.

		\end{enumerate}
	\end{lemma}
\begin{proof}

We prove the statements one by one.
\begin{enumerate}[(i)]

\item
Let $Q$ be the transition matrix of a random walk restricted to $G_2$.  Let $d^Q(t)$ be the total variation distance w.r.t. the transition matrix $Q$.
Further, let $\widehat Q$ be the transition matrix of a random walk restricted to $\widehat{G}$.
Recall that $\tmix(G_2)=O(1)$, by \autoref{lem:BASF}.

Fix an arbitrary  $t \in [ 2\tmix(G_2), 2\tmix(G_2)+7]$.
In the following we show $\tvdist{\widehat q_{u,\cdot}^{t} - \pi^{\widehat{G}}(\cdot )} \leq 1/e$. We first consider any start vertex $u \in V_2\setminus \{\widehat{z}\}$  
and afterwards the vertex $u=\widehat{z}$.
Let $\mathcal{D}$ be the set of distributions over $V(\widehat{G})=V_2 \cup \{ \widehat{z}\}$ assigning no probability mass to $\widehat{z}$, \ie 
\begin{align}\label{eq:thatnicedistribution}
\mathcal{D}= \{ D' \colon \text{ for $u~\sim D'$ we have $\Pr{ u = \widehat{z}}=0$} \}. 
\end{align}

For any such $D' \in \mathcal{D}$, we have, by definition of the total variation distance,
\begin{align*}
	\tvdist{\widehat q_{u\sim D',\cdot}^t -   \pi^{\widehat{G}}(\cdot) } &= 0+ \frac12 \sum_{v\in V_2} \left |\widehat q^t_{u\sim D',v} - \pi^{\widehat{G}}(v) \right| + \frac{1}{2} \left| \widehat q^t_{u\sim D',\widehat{z}} - \pi^{\widehat{G}}(\widehat{z})   \right| .
	\end{align*}
	
	For   $u\in V_2$ observe that $\pi^{\widehat{G}}(u) \in [\pi^{G_2}(u)(1-\zeta),\pi^{G_2}(u)(1+\zeta)]$ for some $\zeta = o(1)$.
	By \cite[Exercise 4.1]{LPW06} we have the following identity for $d^Q(t)$.
Let $\mathcal{D^*}$ be the set of all distributions over $V(G_2)$, then
\begin{align*}
	d^Q(t) = \max_{D\in \mathcal{D^*}} \tvdist{q_{u\sim D,\cdot}^t- \pi^{G_2}(\cdot)} \geq 
 \max_{D'\in \mathcal{D}} 	\tvdist{q_{u\sim D',\cdot}^t- \pi^{G_2}(\cdot)}.
\end{align*}
Thus, for   $\delta_v :=|\widehat q^t_{u,v}-q^t_{u,v}|$, we get by using triangle inequality,   
\begin{align}\label{eq:dagstuhl}
	 \frac12 \sum_{v\in V_2} \left| \widehat q^t_{u\sim D',v} - \pi^{\widehat{G}}(v) \right| 
	  &\leq \frac12 \sum_{v\in V_2} \left|\widehat q^t_{u\sim D',v} - \pi^{G_2}(v) \right| +
\frac12 \sum_{v\in V_2}|\pi^{G_2}(v)-\pi^{\widehat{G}}(v)| \notag\\
	 &\leq \frac12 \sum_{v\in V_2} \left|\widehat q^t_{u\sim D',v} - \pi^{G_2}(v) \right| +
\frac12 \sum_{v\in V_2}\pi^{G_2}(v)\zeta\notag\\
&\leq \frac12 \sum_{v\in V_2} \left|q^t_{u\sim D',v} - \pi^{G_2}(v) \right| 
+ \frac12 \sum_{v\in V_2} |\delta_v| 
+ 
\frac12 \sum_{v\in V_2}\pi^{G_2}(v)|\zeta| \notag\\
&\leq d^Q(t)+1/32+\frac{\zeta}{2}, \notag\\
&\leq d^Q(t)+1/32+1/32,
\end{align}
where the second-last inequality is due to \autoref{lem:huehue}.(ii), $\frac{1}{2}\sum_{v\in V} |\delta_v| \leq t/(2\sqrt{n})\leq \frac1{32}$. 
By definition of the $\tmix(G_2)$ and by sub-multiplicativity we have
$
d^Q(t)\leq d^Q(2\tmix(G_2))  \leq 1/e^2.
$

The above equation \eqref{eq:dagstuhl} only consider the variation distance w.r.t. $V_2$.
For $\widehat{z}$ we have 	$ \frac{1}{2} |\widehat q^t_{u\sim D',\widehat{z}} - \pi^{\widehat{G}}(\widehat{z})| \leq (2\tmix(G_2)+7)/\sqrt{n} \leq 1/32 $.

	Putting everything together we get
we get	
	\begin{align}\label{eq:firstbity}
	\tvdist{ \widehat q_{u\sim D',\cdot} ^t - \pi^{\widehat{G}}(\cdot) } &= \frac12 \sum_{v\in V_2} \left|\widehat q^t_{u\sim D',v}-\pi^{\widehat{G}}(v) \right| + \frac{1}{2} \left|\widehat q^t_{u\sim D',\widehat{z}} - \pi^{\widehat{G}}(\widehat{z}) \right|   \notag\\
	&\leq d^Q(t)+1/32+1/32 + 1/32\leq 1/e^2 + 3/32 \\
	&\leq 1/e. \label{eq:secondbity}
	\end{align}

Consider the random walk  starting at $\widehat{z}$ and let $(X_0,X_1, \dots)$ denote its trajectory.
Observe that at time $7$ we have \[\widehat q^7_{\widehat{z},\widehat{z}}\leq  \frac{1}{2^7} + \sum_{i\leq 7} \sum_{v \in N(\widehat{z})} \widehat q^{i-1}_{\widehat{z},v} \cdot \frac{1}{2(\sqrt{n}+1)} \leq \frac{1}{2^7}+ \frac{7^2}{\sqrt{ n}} \leq 1/32
.\]

The set of distribution for the position of the random walk at time $7$ conditioning on  $X_7\neq \widehat{z}$ gives
the same distribution $\mathcal{D}$ as defined in \eqref{eq:thatnicedistribution}.
Let $D_{\widehat{z}}\in \mathcal{D}$ be distribution of the random at time $7$ starting at $\widehat{z}$. 
 Hence, by\eqref{eq:firstbity}, we get 
\begin{align}
\tvdist{\widehat q_{\widehat{z},\cdot}^{ 2\tmix(G_2)+7} - \pi^{\widehat{G}}(\cdot )} &\leq \widehat q_{\widehat{z},V(\widehat{G})\setminus \{\widehat{z}\}} \cdot \tvdist{\widehat q_{u\sim D_{\widehat{z}},\cdot}^{ 2\tmix(G_2)} - \pi^{\widehat{G}}(\cdot)} + \widehat q_{\widehat{z},\widehat{z}}\cdot 1 \\
&\leq 1\cdot (1/e^2 + 3/32)  + 1/32\leq 1/e.	
\end{align}

Thus, for $t'=2\tmix(G_2)+7$ we have  $\tvdist{\widehat q_{\widehat{z},\cdot}^{t'} - \pi^{\widehat{G}}(\cdot )} \leq 1/e$.
Together with \eqref{eq:firstbity}, we conclude that for all $u\in V'$, $\tvdist{\widehat q_{u,\cdot}^{t'} - \pi^{\widehat{G}}(\cdot )} \leq 1/e$  and by definition of $\tmix$ and  we get $\tmix(\widehat{G})\leq 2\tmix +7 =O(1)$. 

\item
To prove $\Pr{\hit(u,\widehat{z}) \geq n/c }\geq 1/2$  for $u\in V_2$ we show that the random walk restricted to $\widehat{G}$ does not hit $\widehat{z}$ after $n/c_1$ steps w.c.p. for some large enough constant $c_1$.
By the Union bound, for some large  constants $c_1,c_2$ that 
\begin{align*}
\Pr{\hit^G(u,\widehat{z}) \leq n/c_1}
&=\Pr{\hit^{\widehat{G}}(u,\widehat{z}) \leq n/c_1}
 \leq \sum_{t= 1}^{n/c_1} \widehat q^{t}_{u,\widehat{z}}  \leq  
\sum_{t= 1}^{c_2 \log n} 1/\sqrt{n} 
+  \sum_{t= c_2 \log n}^{n/c_1} \widehat q^{t}_{u,\widehat{z}}\\
&\leq  o(1)+ n/c_1 \cdot (\pi^{\widehat{G}}(\widehat{z})+1/n^2) \leq 1/2,
\end{align*}
where we used 
$\widehat q_{u,\widehat{z}}^{t} \leq \pi^{\widehat{G}}(\widehat{z}) + \sqrt{\frac{\pi^{\widehat{G}}(\widehat{z})}{\pi^{\widehat{G}}(u)}} \lambda_2(\widehat{G})^{t}$ (\autoref{lem:loop}).

We proceed by bounding that $\Pr{\hit(u,\widehat{z}) \geq n/c_1 } \geq 1/2$ for $u\in V_1$.
Consider first a random walk $(\tilde{X}_t)_{t \geq 0}$ restricted to $G_1^1=G_1$ that starts at vertex $z^1$ and let $\tilde P$ denote the transition matrix.
Furthermore, in order to couple the random walk $\tilde{X}_t$ restricted to $G_1$ with a random walk in $G$, we will consider the random variable
$
  \tilde{Z}:= \sum_{t=0}^{\tsep^{G_1}} \mathbf{1}_{\tilde{X}_t = z_1}.
$ Since $G_1$ is a clique, $\tsep^{G_1} = O(1)$, and $\tilde p_{z_1,z_1}^t \leq \frac{1}{\sqrt{n}} + \lambda_2(G_1) ^ t$ by \autoref{lem:loop}, where $\lambda_2(G_1)$ is some constant bounded away form $1$. Therefore,
$
 \E{ \tilde{Z} } = \sum_{t=0}^{ n /c_1 } \tilde p_{z_1,z_1}^t \leq 2\sqrt{n}/c_1.
$
Let $\gamma := 4\cdot \E{ \tilde{Z} }$. Then, by Markov's inequality
\begin{align*}
  \Pr{ \tilde{Z} \geq \gamma } \leq  1/4. 
\end{align*}

Consider now the straightforward coupling between a random walk $(X_t)_{t \geq 1}$ in $G$ that starts at vertex $z^{1}$ and the random walk $(\tilde{X}_t)_{t \geq 1}$ restricted to $G_1^i$ that starts at the same vertex. Whenever the random walk $\tilde{X}_t$ is at a vertex different from $z^{1}$, then the random walk $X_t$ makes the same transition. If the random walk $\tilde{X}_t$ is at vertex $z^1$, then there is a coupling so that the random walk $X_t$ makes the same transition as $\tilde{X}_t$ with probability $ \frac{2\sqrt{n}-1}{2 \sqrt{n}}$. Conditional on the event $\tilde{Z} \leq \gamma   $ occurring, the random walk $\tilde{X}_t$ follows the random walk $X_t$ up until step $n/c_1$ with probability at least
\[
 p := \left(  \frac{2\sqrt{n}-1}{2 \sqrt{n}} \right)^{\gamma} \geq 3/4,
\] 
 since the random walk $\tilde{X}_t$ has at most $\gamma$ visits to $z^1$. Therefore, by the Union bound,
\[
 \Pr{\hit^G(u,\widehat{z}) \geq n/c_1} \geq \Pr{ \cup_{t=0}^{n/c_1} X_t = \tilde{X}_t} \geq 
 1-  \Pr{ \tilde{Z} \geq \gamma } - (1-p) \geq 1/2
\]
and the proof is complete.

\item 
We proceed by showing $\thit(u,\widehat{z}) = O(n)$ for $u\in V_2$.

Let $Q$ be the transition matrix of the random walk restricted to $G_2$.
Let $u\in V_2$ and $S^* = N(\widehat{z})$ be the neighbors of $\widehat{z}$ in $G_2$.
For every $v\in S^*$ we have $\pi^{G_2}(v)=\frac{\sqrt{n}+1}{\frac{n}{\sqrt{\alpha'}} \sqrt{n}+ \frac{\sqrt{n}}{\sqrt{\alpha'}} } \geq  \frac{\sqrt{\alpha'}}{1.2n }. $
Hence, after $t=\tsep(G_2)$ we have that 
\[q^t_{u ,S^*} := \sum_{v\in S^*} q^t_{u,v} \geq \sum_{v\in S^*}\pi^{G_2}(v)(1-e^{-1}) \geq \frac{\sqrt{n}}{\sqrt{\alpha'}} \cdot  \frac{\sqrt{\alpha'}}{1.2n }(1-e^{-1}) =\frac{1-e^{-1}}{1.2\sqrt{n}}. \]
By \autoref{lem:huehue}, we have for any $u\in V_2$ that $ \tvdist{ p^t_{u,\cdot}-q^t_{u,\cdot} } \leq \tsep(G_2)/(2\sqrt{n})$. 
To bound $\hit^{G}(u,\widehat{z})$ we show that after $\tsep+1=O(1)$ steps  the random walk hits $\widehat{z}$ w.p. $\Omega(1/n)$.

We distinguish between two cases.
\begin{enumerate}
\item For all $i\leq t$ we have $p^t_{u ,S^*} \leq 1/\tsep(G_2)$.
Thus, by \autoref{lem:huehue}.(i) 
\begin{align*}
p^t_{u ,S^*}=\sum_{v\in S^*} p^t_{u,v} &\geq 
q^t_{u,S^*} - \tvdist{ p^t_{u,\cdot}-q^t_{u,\cdot} } \\
&\geq 
\frac{1-e^{-1}}{1.2\sqrt{n}} - \sum_{i=1}^{t-1} p^i_{u,S^*}/(2\sqrt{n})\\
&\geq 
\frac{1-e^{-1}}{1.2\sqrt{n}} - \frac{\tsep(G_2)}{\tsep(G_2)2\sqrt{n}} 
= \Omega(1/\sqrt{n}). 
\end{align*}
Hence, the random walk hits $\widehat{z}$ after $\tsep(G_2)+1$ w.p. at least
$p^{t}_{u ,S^*}\cdot \min_{v \in S^*} \{  p_{v , \widehat{z}} \}=\Omega(1/n)$.
\item Otherwise there exists a $t^*$ such that $p^{t^*}_{u ,S^*} > 1/\tsep(G_2)$.
Thus  the random walk hits $\widehat{z}$ after $\tsep(G_2)+1$ w.p. at least
$p^{t^*}_{u ,S^*}\cdot \min_{v \in S^*} \{  p_{v , \widehat{z}} \}=\Omega(1/n)$.
\end{enumerate}
Thus after $O(1)$ steps the random walk hits $\widehat{z}$ w.p. $\Omega(1/n)$.

We now show a similar statement if $u\in V_1$. Let $(X_t)_{t\geq 0}$ be a random walk on $G$ starting on $u$. 
Observe that  $X_t$ (the walk on $G$) hits $\widehat{z}$ with probability  $p^1_{u,z^1} \cdot p^1_{z^1,\widehat{z}}=\Omega(1/n)$ in $2$ time steps.
Hence, for any $u\in V$ we $\Pr{\hit(u,\widehat{z}) = O(1) } = \Omega(1/n)$.
Thus, repeating this iteratively and using independence yields $\thit(u,\widehat{z}) = O(n)$ for $u\in V$.
\end{enumerate}
\end{proof}

To establish a bound on the mixing time of $G$, we will make use of the following result of Peres and Sousi.
\begin{theorem}[\cite{PS15}]\label{lem:PS15} 
	For any $\beta < 1/2$, let $\thit(\beta)= \max_{u, A: \pi(A)\geq \beta} \thit(u,A)  $.
	 Then there exist positive constants $c_\beta$ and $c'_\beta$ such that
	 \[ c'_\beta \cdot \thit(\beta) \leq \tmix(1/4) \leq c_\beta \cdot \thit(\beta). \]
\end{theorem}

In the following we show for any $\beta$ close enough to $1/2$, that any $A\subseteq V$ satisfying $\pi(A)\geq \beta$ must include at least a constant fraction of nodes from a constant fraction of  copies of $G_1$. 
\begin{claim}\label{lem:hashtagconstantfrac} 
	 Let $\beta=1/2-10^{-3}$. For any $A \subseteq V$ with $\pi(A) \geq \beta$, define 
		$
		H(A) = \{ i~|~|G_1^i \cap A| \geq |V_1|/(2e)  \}.
		$	
		Then,  $|H(A)| \geq \kappa/(2e)$.
\end{claim}
\begin{proof}
	This follows from a simple pigeon-hole argument: Suppose $|H(A)| < \kappa/(2e)$ was true. Then,
			\begin{align*}
			\pi(A)&\leq  |H(A)|\cdot  \pi(V_1) + (\kappa-|H(A)|)\cdot \left(\frac{\pi(V_1) }{2e} + \pi(z^i) \right) + \pi(V_2) + \pi(\widehat{z}) \\
			&<\frac{\kappa}{2e}   \cdot  \pi(V_1) + \kappa\cdot \left(\frac{\pi(V_1) }{2e} + \pi(z^i) \right) + 1/20 
< \beta \leq \pi(A),
			\end{align*}
		which is a contradiction and hence choice of $A$ must fulfill $|H(A)| \geq \kappa/(2e)$.
\end{proof}

We are now ready to determine the mixing time of $G$. The lower bound is a simple application of Cheeger's inequality, while the upper bound combines the previous lemmas with \autoref{lem:PS15}.
\begin{lemma}\label{lem:mix}
Let  $G$ be the graph described at the beginning of \autoref{sec:lowerbound}. We have 
$\tmix(G) = \Theta(n)$. 
\end{lemma}
\begin{proof}
First we show $\tmix = \Omega( n )$. The \emph{conductance} of $G=(V,E)$ is defined  by 
$\Phi(G)=\underset{
\substack{U\subseteq V,\\
0 < \vol(U) \leq \vol(V)/2}
}{\min} \frac{|E(U,V \setminus U)|
}{\vol(U)}.$
In particular, for $U=V_1$ we get that $\Phi(G)\leq \frac{ 4 }{ n }$.
Hence, by Cheeger's inequality and $\left(\frac{1}{1-\lambda_2(G)} -1\right)\cdot \log(\frac{e}{2})\leq \tmix(1/e)$  (see, \eg \cite[Chapter 12]{LPW06}),
\[
\frac{ n } {4 } \leq \frac{1}{\Phi(G)} \leq \frac{2}{1-\lambda_2(G)} =
\frac{2}{1-\lambda_2(G)}-2 +2
 \leq \frac{2\tmix}{ \log \left( \frac{e}{2}\right)} +2.
\]
Rearranging the terms yields $\tmix = \Omega(n )$.

We proceed with the upper bound on the mixing time. Let $\beta = 1/2-10^{-3}$ and let $A\subseteq V$ be an arbitrary set satisfying $\pi(A)\geq \beta$. 
First, we apply
\autoref{lem:hashtagconstantfrac} to conclude
that $|H(A)| \geq \kappa/(2e)$. This immediately implies that with
 $Z := \{ z^i \colon i \in H(A)\}$, $|Z| \geq \kappa/(2e)$.
The remainder of the proof is divided into the following three parts:
\begin{enumerate}[(i)]
 \item Starting from any vertex $u \in V$, with probability at least $1/2$, the random walk hits $z^{*}$ after  $2 \max_{u \in V} \thit(u,\widehat{z})=O(n)$ steps.
 \item With constant probability $p_1 > 0$, the random walk moves from $z^{*}$ to a vertex in $Z$.
 \item With constant probability $p_2 > 0$ a random walk starting from a vertex in $Z$ will hit $A$ after one step.
\end{enumerate}

It is clear that combining these three results shows that with constant probability $\frac{1}{2} p_1 p_2 > 0$, a random walk starting from an arbitrary vertex $u \in V$ hits a vertex in $A$ after $O(n) + 1 + 1$ time-steps. Iterating this and using independence shows that $\thit(u,A) = O({n})$, and hence by \autoref{lem:PS15}, $\tmix = O({n})$ as needed.  

\textbf{Part (i).} Consider  $\max_{u \in V} \thit(u,\widehat{z})$. For $u \in V$, \autoref{lem:holyhit}.$(iii)$ implies $\thit(u,\widehat{z}) = O({n})$.

\textbf{Part (ii).} If the random walk is on $z^{*}$, then since $\deg(z^{*}) = \kappa + \sqrt{n/\alpha'}$, $|Z| \geq \kappa/(2e)$, it follows that the random walk hits a vertex in $Z$ after one step with constant probability $p_1 := \frac{|Z|}{2(\kappa + \sqrt{n/\alpha'})}  > 0$.

\textbf{Part (iii).} Finally,
for any $z\in Z$ we have that $p_2= p_{z,A}=\frac{|V_1|/(2e)}{2\sqrt n} > 0$ 
and the proof is complete.
\end{proof}

In the following we establish the bound on the meeting time. As it turns out, any meeting is very likely to happen on $V_2$ and it takes about $\Theta(\alpha'n)$ time-steps until both walks reach $V_2$ simultaneously. The lower bound then follows from our common analysis method \eqref{eq:central}. The upper bound combines the mixing time bound of $O(n)$ (\autoref{lem:mix}), and that once a random walk reaches a copy of $G_1$, it says there for $\Theta(n)$ steps with constant probability \autoref{lem:holyhit}.$(ii)$.

\begin{lemma}\label{lem:meet}
Let  $G$ be the graph described at the beginning of \autoref{sec:lowerbound}. We have 
	$\tmeet(G)=\Theta(\alpha' {n})$.
\end{lemma} 
\begin{proof}
	We start by proving $\tmeet = \Omega(\alpha' {n})$:  
	Consider two non-interacting\lfnote{replace these by something}, random walks with starting positions drawn from the stationary distribution $\pi$. 
	Let $\ell=c'\alpha' n$, for some small enough constant $c'>0$.
	Let $Z_1$ be the number of collisions of the two random walks on the nodes in $V_1^1 \cup V_1^2 \cup \dots \cup V_1^\kappa$.
	Let $Z_2$  be the number of collisions of the two random walks on the nodes in $V_2$.
	Let $Z_*$  be the number of collisions of the two random walks on the node $\widehat{z}$.
	
	Let $Z$ be the number of collisions of the two walks during the first $\ell$ time steps, \ie $Z=Z_1 + Z_2+Z_*$.
	Using the Union bound  we derive
	
	\begin{align}\label{allianz}
		\Pr{Z\geq 1} &\leq \Pr{Z_1\geq 1} + \Pr{Z_2\geq 1} + \Pr{Z_*\geq 1} \notag\\
		&\leq \frac{ \E{Z_1} }{  \E{Z_1 \,|\, Z_1 \geq 1} } 
	+ \frac{ \E{Z_2} }{  \E{Z_2 \,|\, Z_2 \geq 1} } 
	+ \frac{ \E{Z_*} }{  \E{Z_* \,|\, Z_* \geq 1} }. 
	\end{align}
	
	We have $\E{Z_1} \leq \ell n \left(\frac{2}{n}\right)^2$, $\E{Z_2} \leq \ell \frac{n}{\sqrt{\alpha'}} \left(\frac{2}{n}\right)^2$, and
	$\E{Z_*} \leq \ell  \left(\frac{2}{n}\right)^2$, since $\max_{u} \pi(u) \leq 2/n$.
	Conditioning on $Z_1\geq 1$ and since both random walks start from the stationary distribution, we have, by \autoref{lem:fresenius}, that the first meeting 	happens in the first $\ell/2$ time steps w.p. at least $1/2$.
	
	Consider $\E{Z_1 \,|\, Z_1 \geq 1}$.
	Suppose the meeting occurred at node $u\in V_1$.
	Let $\mathcal{E}_1$ be the event that for  $u\in V_1$ we have $\hit({u,\widehat{z}}) \geq n/c$ for both walks, where $c > 0$ is a large enough constant.
	By \autoref{lem:holyhit}.$(ii)$, we have that $\Pr{\mathcal{E}_1} \geq (1/2)^2=1/4$ due to independence of the walks. 
	 For any $t < n/c$ let $\widehat p^t_{u,\cdot}$ be the distribution of 
	the random walk on $G_1$ starting on $u$ after $t$ time steps under the conditioning $\mathcal{E}_1$.
	Observe that  $\sum_{v\in V_1}  \widehat p^t_{u,v}=1$ implying that $\sum_{v\in V_1}   (\widehat p^t_{u,v})^2 \geq \sum_{v\in V_1} \left( \frac{1}{|V_1|}\right)^2 = 1/|V_1|$.
	Hence, we get  
	\[ \E{Z_1 \,|\, Z_1 \geq 1}  \geq \E{Z_1 \,|\, Z_1 \geq 1,\mathcal{E}_1} \cdot \Pr{\mathcal{E}_1} \geq \frac12 \min_{u \in V_1} \sum_{t=0}^{n/c-1} \sum_{v\in V_1} (\widehat p^t_{u,v})^2 \geq 
	\frac14  \sum_{t=0}^{n/c-1} 1/|V_1| 	= \frac{\sqrt{n}}{4c}.\]
	
	Using an exactly analogous  analysis for $Z_2$ we can upper bound $\E{Z_2 \,|\, Z_2 \geq 1}$ as follows:
		\[ \E{Z_2 \,|\, Z_2 \geq 1}  \geq \E{Z_2 \,|\, Z_2 \geq 1,\mathcal{E}_2} \cdot\Pr{\mathcal{E}_2} \geq \frac14 \min_{u \in V_2} \sum_{t=0}^{n/c-1} \sum_{v\in V_2} (\widehat p^t_{u,v})^2 \geq 
	\frac14  \sum_{t=0}^{n/c-1} 1/|V_2| 	=\frac{\sqrt{\alpha'}}{4c},\] 
		where $\mathcal{E}_2$ is the event that for  $u \in V_2$ we have $\hit({u,\widehat{z}}) \geq n/c$ for some large enough constant $c$.
	Plugging everything into \eqref{allianz} and using $\ell=c'\alpha' n$ yields
	\begin{align*}
		\Pr{Z\geq 1} &
		\leq \frac{ \E{Z_1} }{  \E{Z_1 \,|\, Z_1 \geq 1} } 
	+ \frac{ \E{Z_2} }{  \E{Z_2 \,|\, Z_2 \geq 1} } 
	+ \frac{ \E{Z_*} }{  \E{Z_* \,|\, Z_* \geq 1} }\\
	&\leq \frac{ \ell n \left(\frac{2}{n}\right)^2 }{ \frac{\sqrt{n}}{4c}} +
	\frac{\ell \frac{n}{\sqrt{\alpha'}} \left(\frac{2}{n}\right)^2}{  \frac{\sqrt{\alpha'}}{4c}  } + \frac{\ell \left(\frac{2}{n}\right)^2}{1}\\
	&\leq o(1) + 16 c\cdot  c' + o(1) \leq 1/2,
	\end{align*}
for any constant $c' \in (0, \frac{1}{33c}]$.  This finishes the proof of $\tmeet = \Omega(\alpha' {n})$.
	In the remainder we prove	
	 $\tmeet = O(\alpha' {n})$. 
	 Consider two independent walks $(X_t)_{t \geq 0}$ and $(Y_t)_{t \geq 0}$  on $G$, both starting from arbitrary nodes. Note $\tsep=\tsep(G) \leq 4 \tmix = O({n})$ by \autoref{lem:mix}, and
	 \[
	 p_0 :=\Pr{ \left\{ X_{\tsep} \in V_2 \right\} \cap \left\{ Y_{\tsep}\in V_2 \right\} } \geq \left(\sum_{u\in V_2} (1-e) \pi(u)  \right)^2 = \Omega\left(\left(\ifrac{1}{\sqrt{\alpha'}}\right)^2\right)=\Omega\left(\ifrac{1}{\alpha'}\right). 
	 \]
%
%
%
We assume in the following that
$\left\{ X_{\tsep} \in V_2 \right\} \cap \left\{ Y_{\tsep}\in V_2 \right\} $.
We have $\tmix(G_2)=O(1)$, by \autoref{lem:BASF}.
	 Consider  a random walk $(\tilde{X}_t)_{t \geq \tsep}$ restricted to $G_2$ that starts at vertex $X_{\tsep}\in V_2$ and let $\tilde P$ denote the transition matrix.
Furthermore, in order to couple the random walk $\tilde{X}_t$ restricted to $G_2$ with a random walk in $G$, we will consider the random variable
\[
  \tilde{Z}:= \sum_{t=\tsep}^{\tsep+n/c-1} \sum_{z\in N(\widehat{z})} \mathbf{1}_{\tilde{X}_t = z},
\]
for  $c=32$.
Thus, for any $z \in N(\widehat{z})$,
\begin{align*}
  \E{\tilde{Z}} &\leq \tmix(G_2) +   \sum_{t=\tsep+\tmix(G_2)+1}^{\tsep+n/c-1} |N(\widehat{z})| (\pi^{G_2}(z)+ d^{\tilde P}(t)   ) \\
  &\leq  \tmix(G_2) + |N(\widehat{z})| (n/c) + O(1) 
   	\leq (1+1/e)\sqrt{n}/c .	
\end{align*}
Let $\gamma := 8(1+1/e)\sqrt{n}/c $. Then, by Markov's inequality
\begin{align*}
  \Pr{ \tilde{Z} \geq \gamma } \leq  1/8. 
\end{align*}

Consider now the straightforward coupling between a random walk $(X_t)_{t \geq \tsep}$ in $G$ that starts at vertex $\tilde X_{\tsep}\in V_2$ and the random walk $(\tilde{X}_t)_{t \geq \tsep}$ restricted to $G_2$ that starts at the same vertex. Whenever the random walk $\tilde{X}_t$ is at a vertex in $V_2\setminus\{ N(\widehat{z})\}$, then the random walk $X_t$ makes the same transition. If the random walk $\tilde{X}_t$ is at vertex $z'\in N(\widehat{z})$, then there is a coupling so that the random walk $X_t$ makes the same transition as $\tilde{X}_t$ with probability $ \frac{2\sqrt{n}}{2 \sqrt{n}+2}$. Conditional on the event $\{ \tilde{Z} \leq \gamma \} $ occurring, the random walk $\tilde{X}_t$ follows the random walk $X_t$ up until step $n/c$ with probability at least
\[
 p_1 := \left(  \frac{2\sqrt{n}}{2 \sqrt{n}+2} \right)^{\gamma}=\left( 1- \frac{1}{ \sqrt{n}+1} \right)^{\gamma} \geq \frac{3}{4},
\] 
 since the random walk $\tilde{X}_t$ has at most $\gamma$ visits to $N(\widehat{z})$.
 Consider now the random walk $(\tilde{Y}_t)_{t\geq \tsep}$ using $\tilde{P}$ (\ie restricted to $V_2$) starting at  $Y_{\tsep}$, i.e., $\tilde{Y}_{\tsep}=Y_{\tsep}$.
 By an analogous argument as before we can couple
 $({Y}_t)_{t\geq \tsep}$ and $(\tilde{Y}_t)_{t\geq \tsep}$ for $n/c$ time steps w.p. at least $p_1$.
 
 Furthermore, after $\tsep(G_2)=O(1)$ steps we can couple $\tilde{X}_t$ and $\tilde{Y}_t$
 with nodes drawn independently from $\pi^{G_2}$. 
 Hence, 
 \[
 p_2 := \Pr{ \tilde{X}_{t+\tsep(G_2)} =\tilde{Y}_{t+\tsep(G_2)} ~|~ \mathcal{F}_{t}   } \geq  (1-1/e)^2 \twonorm{\pi^{G_2}}^2 \geq \frac{\sqrt{\alpha'}}{8n}.
 \]
 Recall that $\alpha' \geq 2^{20}\tsep(G_2)^2 $ by definition.
 Therefore, the probability that  
 $\tilde{X}_t$ and $\tilde{Y}_t$ do not meet in the time-interval $[\tsep(G_1), \tsep(G_1)+n/c-1]$ is  at most
 \[
 p_3 := (1- p_2)^{\floor{n/(\tsep(G_2) c)}} \leq  (1- p_2)^{\floor{2^{10} n/( \sqrt{\alpha'} c)}} \leq 1/4.
 \]
  Therefore, by the Union bound,
\[
 \Pr{ \cup_{t=0}^{\tsep(G_1)+n/c-1} X_t = {Y}_t } \geq 
p_0\cdot \left(1-   \Pr{ \tilde{Z} \geq \gamma } - 2\cdot(1- p_1) - p_3 \right) =\Omega(\alpha').
\]
Repeating this  $O(1/p_3)$ times and using the independence  yields that the expected meeting time is
$O((\tsep(G_1)+n/c-1)/p_3)=O(\alpha' n)$
and the proof is complete.


\end{proof}

Finally, we analyze the coalescing time of $G$. The proof idea is to consider $\sqrt[5]{n}$ random walks starting from $\pi$ and show that meetings only occur on $V_2$ and that at least one random walk requires $\Omega( \sqrt{\alpha' }\cdot n\log n)$ time-steps to reach $V_2$.
\begin{lemma}\label{lem:coal}
Let  $G$ be the graph described at the beginning of \autoref{sec:lowerbound}. We have 
	$\tcoal(G)=\Omega( \sqrt{\alpha' }\cdot n\log n)$.
\end{lemma}
\begin{proof}
Let $\varepsilon = 1/5$.
We  show
that even the coalescing time of $n^{\varepsilon}$ random walks requires 
$\Omega( \sqrt{\alpha' }\cdot n\log n)$ time-steps w.c.p..
Let $R$ be a collection of $n^\varepsilon$ independent, \ie non-interacting, random walks with starting positions drawn from the stationary distribution $\pi$. 
	We define the following three bad events: 
	\begin{enumerate}[(i)]

	\item
	Let $\mathcal{E}_1$ be the event that any of the $n^\varepsilon$ random walks meet on a node  $V\setminus V_2$ in $\sqrt{\alpha' }\cdot n \log^2 n$ steps.	
	\item
	Let  $\mathcal{E}_2$ be the event that fewer than $n^{\varepsilon}/4$   random walks start on copies of $G_1$, \ie on nodes $V\setminus (V_2\cup \widehat{z})$.
	\item
Let  $\mathcal{E}_3$ be the event that all random walks starting from a copy of $G_1$ require fewer than $c \cdot \sqrt{\alpha' }\cdot n \log n$ time-steps for leaving $V \setminus (V_2 \cup z^{*})$ for some constant $c>0$ to be determined later.
	\end{enumerate}
	
	In the following we show that $\Pr{\mathcal{E}_1} = o(1)$, $\Pr{\mathcal{E}_2 } = o(1)$, and $\Pr{\mathcal{E}_3 \, \mid \, \overline{ \mathcal{E}_2} } < 1/e$, which implies, by union bound, 
	\begin{align*}
	\Pr{ \overline{\mathcal{E}_1} \cap \overline{\mathcal{E}_2} \cap \overline{\mathcal{E}_3} } \geq \Pr{\overline{\mathcal{E}_1}} - ( 1 -\Pr{ \overline{\mathcal{E}_2} \cap \overline{\mathcal{E}_3 } }) 
	\geq 1 - o(1) - \left(1  - (1-o(1)) \cdot \left(1-\frac{1}{e}\right) \right) \geq 1 - \frac{1}{2e}.
	\end{align*}
	Conditioning on $\overline{\mathcal{E}_1} \cap \overline{\mathcal{E}_2} \cap \overline{\mathcal{E}_3}$, none of the independent random walks meet on any node  $V\setminus V_2$ and hence they are indistinguishable from coalescing random walks until they reach $V_2$. Therefore, it is necessary for all random walks to reach $G_2$ in order to coalesce.  Hence, we conclude  that $\tcoal(G)=\Omega( \sqrt{\alpha' }\cdot n\log n)$  yielding the lemma.

	\begin{enumerate}[(i)]
\item We now prove $\Pr{\mathcal{E}_1} = o(1)$. 
Consider any pair of the random walks $R$. Since both random walks start from the stationary distribution, the probability for them to meet on a node on $\widehat{z}$ in a fixed step $t \geq 0$ is at most
$O(1/n^2)$.

Hence, by the Union bound over $\binom{n^{\epsilon}}{2}$ pairs of random walks and $\sqrt{\alpha' }\cdot n \log^2 n \leq n \log^3 n$ steps, the probability of any two random walks meeting on $\widehat{z}$ is at most
\begin{align*}
p_1 := \binom{n^{\epsilon}}{2} \cdot n \log^3 n \cdot O(1/n^2) = o(1),
\end{align*}
since $\epsilon = \frac{1}{5}$.
 Furthermore, the probability that no two walks start on the same copy of $G_1$ is at most $p_2 := n^\varepsilon \cdot \frac{n^\varepsilon}{\sqrt{n}} = o(1)$ by the Union bound.
 
Moreover, using a Chernoff bound together with \autoref{lem:holyhit}.$(ii)$, it follows that a random walks visits the vertex $z^{*}$ at most $10 \log^3 n$ times during $n \log^3 n$ steps with probability at least $1-n^{-2}$. By the Union bound over all random walks, it follows that w.p.~at least $1-n^{-1}$, each random walk visits at most $10 \log^3 n$ different copies of $G_1$, and by construction of $G$ each such copy is chosen uniformly and independently at random among $G_1^1,G_1^2,\ldots,G_1^{\kappa}$. Therefore, the probability that there exists a copy of $G_1$ which is visited by at least two random walks in $n \log^3 n$ steps is at most 
	\begin{align}
	  p_3 :=   n^{-1} + n^\varepsilon (10 \log^3 n+1) \cdot \frac{n^\varepsilon(10 \log^3 n+1)}{\sqrt{n}} = o(1).
	  \label{eq:probone}
	\end{align}
Putting everything together, using union bound, yields $\Pr{\mathcal{E}_1} \leq p_1 + p_2 + p_3 =o(1)$.

\item We now prove $\Pr{\mathcal{E}_2} = o(1)$. The probability $p$ for each random walk to start on a node of $V\setminus (V_2\cup \widehat{z})$ is $\pi(V\setminus (V_2\cup \widehat{z}))\geq 1/2$.
For each of the random walks with label $1 \leq i \leq n^{\varepsilon}$ we define the indicator variable $X_i$ to be one, if that random walk starts on $V\setminus (V_2\cup \widehat{z})$. 
Let $X=\sum_{i=1}^{n^\varepsilon} X_i$. We have  $\E{X}=n^\varepsilon \cdot \E{X_i}\geq n^\varepsilon/2$.
Since the starting positions of the $n^\varepsilon$ random walks are drawn independently, by a Chernoff bound 
\[
\Pr{\mathcal{E}_2}=\Pr{X\leq \frac{1}{4} n^{\varepsilon}} \leq \Pr{X\leq \E{X}/2}\leq e^{-n^\varepsilon/16} = o(1).
\]

\item We now prove $\Pr{\mathcal{E}_3 \, \mid \, \overline{\mathcal{E}_2} } < 1/4$. From \autoref{lem:holyhit}.$(ii)$ we get that w.p. at least $1/2$ a random walk  starting at any node $u\in V_1$ does not leave $G_1$, \ie does not reach $z^{*}$, after $c_1 n$ time-steps for some constant $c_1 > 0$. It is easy to see that the number of visits to $\widehat{z}$ required before the random walk hits $G_2$ instead of returning to $G_1$ is w.c.p. at least $\sqrt{\alpha'}/2$; this is because the fraction of edges from $\widehat{z}$ to $G_2$ is $\sqrt{n/\alpha'}/(\sqrt{n/\alpha'}+ \sqrt{ n})$. Using a Chernoff bound, we conclude that any random walk starting at $G_1$ doesn't hit  $G_2$ during the first $T=c_1\cdot\sqrt{\alpha'}n/2$ time-steps with constant probability $p > 0$.
Thus the probability that a random walk does not reach $G_2$ after $\lambda \cdot T$ time-steps is at least $p^\lambda$, for any integer $\lambda \geq 1$.
Setting $\lambda=\epsilon \cdot \log(1/p) \cdot  \log (n/4)$, the probability that all of the at least $\frac{1}{4} n^{\varepsilon}$ random walks starting from $G_1$ reach $G_2$ within $\lambda \cdot T=\Omega( \sqrt{\alpha'}\cdot n\log n)$ steps is 
\[
\Pr{\mathcal{E}_3 \, \mid \, \overline{ \mathcal{E}_2 } } \leq (1-p^\lambda)^{\frac{1}{4} n^{\varepsilon}} \leq 1/e,
\]  
completing the proof.
	\end{enumerate}
		\end{proof}

The following lemma establishes a bound on the average hitting time. 
\begin{lemma}\label{lem:avghit}
Let  $G$ be the graph described at the beginning of \autoref{sec:lowerbound}. We have 
	$\tavghit = \Omega(n^{3/2})$
\end{lemma}
	\begin{proof} 
	Consider a random walk that starts from an arbitrary vertex $u \in V$. By \autoref{lem:holyhit}.$(ii)$, every time a vertex $z^i$ is visited, with probability at least $c>0$ it takes $\Omega(n)$ time-steps to visit another vertex $z^j$, $j \neq i$. 
Using a Chernoff bound, it follows that with probability larger than $1/2$ it takes at least $\Omega(n^{3/2})$ time-steps	
to visit at least half of the nodes in $\{z^1,z^2,\ldots,z^{\kappa}\}$.
By symmetry, it follows that for every vertex in a copy of $G_1$ there are $\Omega(n)$ vertices to which the hitting time is $\Omega(n^{3/2})$.
Thus, by symmetry,  $\tavghit= \sum_{u,v\in V} \pi(u)\cdot\pi(v)\cdot \thit(u,v)= 
\Omega( n^2 \frac{1}{n^2} n^{3/2})=\Omega(n^{3/2}) $.
\end{proof}

\section{Bounding $\tcoal=O(\thit)$ for Almost-Regular Graphs  }
\label{sec:hittingtime}

As mentioned in the introduction, the bound on $\tcoal$ in terms of $\thit$ will be based on the combination of two reduction results; the first result reduces the number of walks from $O(n)$ to $O(\log^3 n)$, while the second one reduces the number of walks  from $O(\log^3 n)$ to $(\Delta/d)^{100}$; both taking $O(\thit)$ time. In \autoref{sec:conc}, we first develop concentration inequalities that will be needed for these reductions. Then in \autoref{sec:simplered}, we present the first and technically simpler reduction to $O(\log^3 n)$ walks, which is stated in~\autoref{thm:mostgeneral}. The proof basically combines the concentration inequalities with our well-known formula \eq{central}. 

The derivation of the second reduction is done in \autoref{sec:trickyred}. It is based on identifying nearly-regular and dense subsets $S$, which will contain enough vertices visited by a random walk, even if the walks only run for $o(\thit)$ steps (\autoref{lem:structure}). The proof of this \autoref{lem:structure} also rests on the concentration inequalities we derive. The second reduction is then completed by \autoref{lem:stocknockout}, which uses the dense subsets $S$ provided by \autoref{lem:structure} in order to prove that random walks are likely to collide. A more detailed proof outline can be found at the beginning of \autoref{sec:trickyred}.

\subsection{Concentration Inequalities for Random Walks}\label{sec:conc}

In this part, we derive several concentration inequalities for random walks that are new to the best of our knowledge. We point out that existing concentration inequalities tend to fail in our setting, since the events we are considering (like visits to a certain vertex or expected collisions with an unexposed walk) may only appear a small number of times during $\thit$ steps. Therefore, we have to develop new concentration inequalities that are parameterized by $\thit$. Although the derivation is fairly elementary, the bounds are quite general and may complement existing bounds that are usually parameterized by the mixing time~\cite{CLLM12,Lez89}. 
In particular, our bounds are most useful when $\tmix$ and $\thit$ are close, which is precisely the challenging regime for proving $\tcoal = O(\thit)$. One limitation though is that our bounds only work for large deviations exceeding the expectation by a multiplicative factor.

\begin{lemma}\label{lem:return} 
Let $f:V \rightarrow [0,1]$ be any function over the vertices and $\overline{f} = \sum_{u \in V} f(u) \cdot \pi(u)$. Then for any random walk starting from an arbitrary vertex $X_0$ and any number of steps $T \geq 0$,
\[
  \E{ \sum_{t=0}^{T-1} f(X_t)} \leq 8\cdot T^{+} 
  \cdot \overline{f},
  \]
  where $T^{+} = \max\{ \thit, T\}$.
  Furthermore, for any integer $\lambda \geq 1$,
\begin{align*}
  \Pr{ \sum_{t=0}^{T-1} f(X_t) \geq \lambda \cdot \left( 16 \cdot T^{+} \cdot \overline{f} + 1 \right) } \leq 2^{-\lambda}.
\end{align*}
Moreover, suppose we have time-dependent functions, $f_t: V \rightarrow [0,1]$, $0 \leq t \leq T$, where $T$ may be any integer. Further assume that there is a universal bound $\Upsilon > 0$ so that for any $1 \leq s \leq T$ and any $w\in V$,\fnote{reviewer asks if for any s and any $w \in V$. I added $\in V$ in response. Please verify and delete this comment}

\begin{align*}
 \E{ \sum_{t=s}^{T-1} f_t(X_t) \,\mid\, X_s = w } \leq \Upsilon.
\end{align*}
Then, again for any integer $\lambda \geq 1$, 
\begin{align*}
 \Pr{	\sum_{t= s}^{T-1} f_t(X_t) \geq \lambda \cdot (2 \Upsilon+1)} \leq 2^{-\lambda}
\end{align*}

\end{lemma}

\begin{proof}
We first prove that
for all pairs of states $u,v \in V$ and any $T \geq \thit$ that,
\[
  \sum_{t=0}^{T-1} p_{u,v}^t \leq 8 \cdot T^+ \cdot \pi(v).
\]
Suppose for a sake of contradiction that $ \sum_{t=0}^{T-1} p_{u,v}^t > 8 \cdot T^{+} \cdot \pi(v)$. Then, for an arbitrary vertex $w \in V$, by Markov's inequality, 
$\Pr{ \hit(w,u) \leq 2 \,  \thit} \geq \frac{1}{2}$, where we recall that $\hit(w,u)$ is the first time step at which a random walk starting at $w$ hits $u$. 
We will use $N_{t}(u, v)=\sum_{i=0}^{t-1} \mathbf{1}_{X_t=v}$ to denote the number of visits to $v$ up step $t-1$  starting at $u$. 
Therefore,  
\begin{align*}
\E{ N_{3T^{+}}(X_0,v) \, \mid \, X_0 = w} &\geq
\Pr{ \hit(w,u) \leq 2 \,  \thit} \cdot \E{ N_{T^{+}}(X_0,v) \, \mid \, X_0 = u}
\\ &\geq
 \frac{1}{2} \cdot \sum_{t=0}^{T-1} p_{u,v}^t > 4 \cdot T^+ \cdot \pi(v).
 \end{align*} Since this holds for every vertex $w \in V$, we  conclude
$\E{ N_{3 T^+}(X_0,v) \, \mid \, X_0 \sim \pi} > 4 T^+ \cdot \pi(v)$. However, by definition of the stationary distribution, we also have $\E{ N_{3 T^+}(X_0, v) \, \mid \, X_0 \sim \pi} = 3 T^+ \cdot \pi(v)$, which yields the desired contradiction.
Now the first statement of the lemma follows simply by linearity of expectations:
\begin{align*}
  \E{ \sum_{t=0}^{T-1} f(X_t) \, \mid \, X_0=w} &=
  \E{ \sum_{t=0}^{T-1} \sum_{u \in V} \mathbf{1}_{X_t=u} \cdot f(u) \, \mid \, X_0 = w } \\
  &=    \sum_{u \in V} \sum_{t=0}^{T-1}  f(u) \cdot \E{\mathbf{1}_{X_t=u}  \, \mid \, X_0 = w } \\
  &= \sum_{u \in V} f(u) \cdot \sum_{t=0}^{T-1}  p_{w,u}^t \\
  &\leq \sum_{u \in V} f(u) \cdot 8 T^+ \cdot \pi(u) = 8 \cdot T^+ \cdot \overline{f}.
\end{align*}
We now prove the second statement. 
By Markov's inequality, for every $w \in V$,
\[
 \Pr{ \sum_{t=0}^{T-1} f(X_t)  \geq 16 \cdot  T^+ \cdot \overline{f} \, \mid \, X_0 = w} \leq \frac{1}{2}.
\] 
Hence with $\tau := \min \left\{ s \in \mathbb{N} \colon \sum_{t=0}^{s} f(X_t) \geq 16 \cdot T^+ \cdot \overline{f} \right\} $ we have for every $w \in V$,
\begin{align*}
 \Pr{ \tau \leq T -1\, \mid \, X_0 = w} &\leq \frac{1}{2}.
\end{align*}
Since $f$ is bounded by $1$, we get $\sum_{t=0}^{\tau} f(X_t) \leq 16 \cdot T^+ \cdot \overline{f} + 1$ 
and therefore,
\begin{align*}
 \Pr{ \sum_{t=0}^{T-1} f(X_t)  \geq \lambda \cdot (16 \cdot T^+ \cdot \overline{f} + 1)\, \mid \, X_0 = w} &\leq 
 \left( \max_{v \in V} \Pr{ \tau \leq T-1 \, \mid \, X_0 = v} \right)^{\lambda} \leq 
 2^{-\lambda}.
\end{align*}
The third statement is derived in exactly the same way we proved the second statement.
\end{proof}

The third statement of \autoref{lem:return} is very useful in that it can be used the following concentration inequality on $\tilde{Z}$. Notice that the variable random variable $\tilde{Z}$ is defined using only one random walk $(X_t)_{t \geq 0}$, but it can be viewed as the expected number of collisions on the vertex set $S$ of the random walk $(X_t)_{t \geq 0}$ with another (unexposed) random walk $(Y_t)_{t \geq 0}$, starting  from the same vertex $u$. 

\begin{lemma}\label{lem:deterministicpath}
Let $S$ be any subset of vertices such that the degree of any pair of vertices in $S$ differs by at most a factor of $\gamma$. Consider any random walk $(X_t)_{t \geq 0}$ that starts at an arbitrary vertex $u \in S$, and for any $T \geq 0$
\[
 \tilde{Z} := \sum_{t=0}^{T-1} \mathbf{1}_{X_t \in S} \cdot p_{u,X_t}^{t}.
\] 
Then with $T^{+} = \max\{\thit, T\}$ it holds that 
\[
  \E{\tilde{Z}} \leq 16 \gamma\cdot T^+ \cdot \max_{w \in S} \pi(w) =: \Upsilon. 
\]
Furthermore, for any $\lambda \geq 1$,
 \begin{align*}
 \Pr{\tilde{Z} \geq \lambda \cdot (2 \Upsilon+1)} &\leq 2^{-\lambda}.
\end{align*}
\end{lemma}
\begin{proof}
First note that $\tilde{Z}$ is a random variable over the walk $(X_0=u,X_1,X_2,\ldots,X_{T-1})$ with $u \in S$. 
Let us first upper bound the expectation of $\tilde{Z}$:
\begin{align*}
  \E{ \tilde{Z} } &\leq \sum_{t=0}^{T-1} \sum_{v \in S} p_{u,v}^{t} \cdot p_{u,v}^t \leq \gamma \cdot \sum_{t=0}^{T-1} \sum_{v \in S} p_{u,v}^{t} \cdot p_{v,u}^t \leq \gamma \cdot \sum_{t=0}^{T-1} p_{u,u}^{2t} \leq 8 \gamma \cdot T^+ \cdot \pi(u), 
\end{align*}
where the second inequality is due to reversibility, \ie $p_{u,s}^t \cdot \pi(u) = p_{s,u}^t \cdot \pi(s)$ and the fact that the degrees in $S$ differ by a factor of at most $\gamma$, and the fourth inequality uses $p_{u,u}^{2t} \leq \pi(u)$ which hold since $p_{u,u}^t$ is non-decreasing (\autoref{lem:loop}) and the first statement of \autoref{lem:return}.

Furthermore, suppose now that we condition on
the walk $(X_t)_{t \geq 0}$ being at an arbitrary vertex $w \in S$ at step $s$, where $1 \leq s \leq T-1$. Then the remaining contribution towards $\tilde{Z}$ is at most
\begin{align*}
 \E{ \sum_{t=s}^{T-1} \mathbf{1}_{X_t \in S} \cdot p_{u,X_t}^{t} \, \mid \, X_{s} = w} 
  &= \sum_{t=s}^{T-1} \sum_{v \in S} p_{w,v}^{t-s} \cdot p_{u,v}^{t} \\ 
  &\leq \gamma \cdot \sum_{t=s}^{T-1} \sum_{v \in S} p_{w,v}^{t-s} \cdot p_{v,u}^t \\
  &\leq \gamma \cdot \sum_{t=s}^{T-1} \sum_{v \in V} p_{w,v}^{t-s} \cdot p_{v,u}^t \\
  &= \gamma \cdot \sum_{t=s}^{T-1}  p_{u,w}^{2t-s} \\
  &\leq \gamma \cdot \sum_{t=0}^{2T-2}  p_{u,w}^{t}  \\
  &\leq \gamma \cdot 16 \cdot T^+ \cdot \pi(w) \leq \Upsilon,
\end{align*}
where the penultimate inequality is due to the first statement of \autoref{lem:return}, applied to the number of visits to $w$ of a random walk of length $T$, \ie $f(v):= \mathbf{1}_{v=w}$. 
 Finally, by the third statement of \autoref{lem:return}, applied to the functions $f_{t}(v) := \mathbf{1}_{v \in S} \cdot p_{u,v}^t$, $0 \leq t \leq T-1$,
\begin{align*}
 \Pr{\tilde{Z} \geq \lambda \cdot \left( 32 \gamma \cdot T^+ \cdot \max_{w \in S} \pi(w) + 1 \right) } &\leq 2^{-\lambda}.
\end{align*}
\end{proof}

\subsection{Reducing the Walks from $n$ to $O(\log^3 n)$ in $O(\thit)$}\label{sec:simplered}

We now present our first reduction result that reduces the number of walks from $n$ to $O(\log ^3 n)$ in $O(\thit)$ time.

\begin{theorem}\label{thm:mostgeneral}
Let $G=(V,E)$ be an arbitrary, possibly non-regular, graph. Then after $O(\thit)$ steps, the number of walks can be reduced from $n$ to $O(\log^3 n)$ with probability at least $1-n^{-1}$.
\end{theorem}

Thanks to $\tmeet \leq 4\thit$ (\autoref{pro:relatingmeetandhit}) and $\tcoal(S_0) = O(\tmeet \log|S_0|)$ (\autoref{lem:beer}), the result of \autoref{thm:mostgeneral} implies, among other things, a bound of $\tcoal =O(\thit \cdot \log \log n)$ for any graph. The proof idea is as follows. 
First, we use the concentration inequalities of the previous section to show for a given random walk $(X_t)_{t\geq 0}$, there exists w.h.p. a set $S'=S'((X_t)_{t\geq 0})$ of nodes where $(i)$ all nodes have up to a factor of $2$ the same degree, $(ii)$ the stationary mass of that set is   at least $\pi(S') \geq 1/\log^3 n$,
and $(iii)$ the nodes $S'$ receives at least $\Omega(\thit/\log n)$ visits during the interval $[\tsep, \tsep +2\thit]$.
From this we will be able conclude that any random walk collides with $(X_t)_{t\geq 0}$ w.p. at least $p=\Omega(1/\log^2 n)$.
Second, we consider the process $\pimortal$ of \autoref{sec:process} and make use of the majorization given by \autoref{lem:majormotal}.
We divide the walks into two sets $\mathcal{G}_1$ and $\mathcal{G}_2$ with $|\mathcal{G}_1|=\Theta(\log^3 n)$.
We show, using the first part, that w.h.p. each walks $(X_t)_{t\geq 0}$ of $\mathcal{G}_2$ will vanish due to its frequent visits to $S'$ and the fact that each independent random walk $(Y_t)_{t\geq 0}$ of $\mathcal{G}_1$ intersects with $(X_t)_{t\geq 0}$ on $S'$ 
w.p. at least $p$: Using independence, the probability for each walk of $\mathcal{G}_2$ to survive is $(1-p)^{|\mathcal{G}_1|} \leq n^{-2}$. The claim then follows by the Union bound.
\begin{proof}[Proof of \autoref{thm:mostgeneral}]
First consider any random walk $(X_t)_{t=0}^{\tsep+2\thit-1}$, 
that reaches an arbitrary vertex $u$ at time $\tsep$. Next divide all vertices in $V$ into buckets $S_i = \{ v \in V \colon \deg(u) \in [2^{i-1},2^i) \}$, \lfnote{T: I think the notation might be a bit confusing here, since now $i$ has nothing to do with the actual random walk, while later the index $i$ would be used to label the walk $i$.F: I agree, but I couldn't find any more suitable letter. We shouldn't call our random walks $i$, but now it's too late} where $1 \leq i \leq \log_2 n$. 
For any bucket $i$ with $\pi(S_i) \leq 1/\log^3 n$, let $Z_i := \sum_{t=\tsep}^{\tsep+2 \thit-1} \mathbf{1}_{X_t \in S_i}$ count the number of visits to $S_i$. Then by the first statement of \autoref{lem:return}, $\E{Z_i} \leq 16 \thit \cdot \pi(S_i) \leq 16 \thit / \log^3 n$. By the second statement of \autoref{lem:return}, it follows that 
\[
  \Pr{ Z_i \geq \thit / (2 \log_2 n)} \leq \Pr{ Z_i \geq \log^{1.5} n \cdot (32 \thit/\log^3 n + 1) } \leq 2^{-\log^{1.5} n} \leq n^{-2}/(2 \log_2 n),
\]
where we used the fact that $\thit \geq n$ (\autoref{lem:fridaygift}).
Hence with \[
S := \bigcup_{\substack{1 \leq i \leq \log_2 n\colon \\ \pi(S_i) \geq 1/\log^3 n}} S_i,
\] it follows by the Union bound that $S$ gets at least $\thit - \log_2 n \cdot \thit / (2 \log_2 n) = \thit / 2$ visits with probability at least $1-n^{-2}/2$.

Let us now consider any $S_i$ with $\pi(S_i) \geq 1/\log^3 n$, and define 
\[
 \tilde{Z}_i(s) := \sum_{t=s}^{\tsep+ 2 \thit-1} \mathbf{1}_{X_t \in S_i} \cdot p_{X_s,X_t}^{t-s}, \text{ for any $s\in [\tsep, \tsep+2\thit-1]$}. 
\]
Notice that by \autoref{lem:deterministicpath}, $\kappa=2$, setting $\Upsilon = 16 \cdot \kappa \cdot \max\{ \thit, 2 \thit -1-s \} \cdot \max_{w \in S_i} \pi(w) \leq 64 \thit \cdot \max_{w \in S_i} \pi(w)$ we have 
\begin{align*}
 \Pr{ \tilde{Z}_i(s) \geq 520 \log n \cdot \thit \cdot \max_{w \in S_i} \pi(w) } 
 \leq \Pr{ \tilde{Z}_i(s) \geq 8 \log n( 2\Upsilon + 1) } 
 \leq n^{-8},
\end{align*}
having used the fact that $\thit \geq 1/\min_{v \in V} \pi(v)$ due to \autoref{lem:fridaygift}.
 Since $\thit =O(n^3)$ and there are at most $\log_2 n$ buckets, by the Union bound,
\begin{align*}
 \Pr{ \bigcup_{\substack{1 \leq i \leq \log_2 n\colon \\ \pi(S_i) \geq 1/\log^3 n}} \bigcup_{\tsep \leq s \leq \tsep + 2 \thit-1} \left\{ \tilde{Z}_i(s) \geq c \log n  \cdot \thit \cdot \max_{w \in S_i} \pi(w) \right\} } \leq n^{-2}/2.
\end{align*}

Hence by the Union bound, with probability at least $1-n^{-2}$, the trajectory $(x_0,x_1,\ldots,x_{\tsep + 2 \thit-1})$ of $(X_t)_{t=0}^{\tsep+2\thit}$ is {\em good}, \ie its trajectory $(i)$ makes at least $\thit/2$ visits to $S$ during the steps $[\tsep,  \tsep+2\thit]$ and $(ii)$ all $\tilde{Z}_i(s)$ are bounded by $c \log n  \cdot \thit \cdot \max_{w \in S_i} \pi(w)$.

In the following, condition on $(X_t)_{t=0}^{\tsep+2\thit-1}$ being a good random walk, and let us denote by $(x_0,x_1,\ldots,x_{\tsep + 2 \thit-1})$ the deterministic trajectory.
Since $S$ gets at least $\thit/2$ visits in the time-interval $[\tsep,\tsep+2\thit-1]$, by the pigeonhole principle, there must be at least one bucket $S_j$ with $\pi(S_j) \geq 1/\log^3 n$ so that bucket $S_j$ gets at least 
$(1/8) \cdot \thit / \log n$ visits in that time-interval. We shall now prove that any other random walk $(Y_t)_{t=0}^{\tsep+2\thit-1}$, starting from an arbitrary vertex $w \in V$, collides with this deterministic trajectory on a vertex in $S_j$ in the time-interval $[\tsep, \tsep+2 \thit-1]$ with probability at least $\Omega(1/\log^2 n)$.
To this end, let us define
\[
Z:=\sum_{t=\tsep}^{\tsep+2\thit-1} \mathbf{1}_{x_t \in S_j} \cdot \mathbf{1}_{Y_t=x_t}.
\]
Since for any $t \geq \tsep$ we have $p_{w,v}^t \geq \frac{1}{2} \cdot \pi(v)$ for any pair of vertices $w,v \in V$, it follows that
\begin{align*}
 \E{ Z \, \mid \, (x_0,x_1,\ldots,x_{\tsep+2\thit-1}) } &\geq \sum_{t=\tsep}^{\tsep+2 \thit-1} \sum_{v \in S_j} \mathbf{1}_{x_t=v} \cdot p_{w,v}^t \\
 &\geq \sum_{v \in S_j} \sum_{t=\tsep}^{\tsep+2 \thit-1} \mathbf{1}_{x_t=v} \cdot\frac{1}{2} \cdot \pi(v) \\
 &\geq \frac{1}{2} \cdot \min_{w \in S_j} \pi(w) \cdot \sum_{v \in S_j} \sum_{t=\tsep}^{\tsep+2 \thit-1} \mathbf{1}_{x_t=v} \\
 &\geq \frac{1}{16} \cdot \min_{w \in S_j} \pi(w) \cdot \frac{1}{\log n} \cdot \thit,
\end{align*}
where the last inequality holds because the deterministic path $(x_0,x_1,\ldots,x_{\tsep+2\thit-1})$ makes at least $(1/8) \cdot \thit/\log n$ visits to $S_j$.

Furthermore, since the deterministic walk $(x_0,x_1,\ldots,x_{\tsep+2\thit-1})$ satisfies invariant (ii), conditional on $Y_t$ having its first collision with $X_t$ at step $s$ on a vertex $x_s \in S_j$,
\begin{align*}
   \E{ Z \, \mid \, (x_0,x_1,\ldots,x_{\tsep+2\thit-1}), Z \geq 1 } &\leq \max_{\tsep \leq s \leq \tsep+2 \thit-1} \sum_{t=s}^{\tsep+2\thit-1} \mathbf{1}_{x_{t} \in S_j} \cdot p_{x_{s},x_{t}}^{t-s} \\
    &= O(\log n \cdot \thit \cdot \max_{w \in S_j} \pi(w)),
\end{align*}
by part $(ii)$ of the definition of a good walk.
Combining our last two bounds yields
\begin{align}
 \Pr{ Z \geq 1 \, \mid \, (x_0,x_1,\ldots,x_{\tsep+2\thit-1}) } &= \frac{ \E{ Z \, \mid \, (x_0,x_1,\ldots,x_{\tsep+2\thit-1}) }}{\E{ Z \, \mid \, (x_0,x_1,\ldots,x_{\tsep+2\thit-1}) , Z \geq 1 }} \notag \\ &= \Omega(1/\log^2 n) =: p \label{eq:goodd}
\end{align}

To complete the proof of the theorem, divide the $k > 2/p \cdot \ln n$ random walk arbitrarily into two disjoint groups $\mathcal{G}_1$ and $\mathcal{G}_2$ such that $|\mathcal{G}_1| = 2/p \cdot \ln n$. 
We will analyze the $\pimortal$ process defined in \autoref{sec:process} in which 
random walks from $\mathcal{G}_1$ are immortal. By making use of the majorization given in \autoref{lem:majormotal}, to show the claim it suffices to bound
the time it takes in $\pimortal$ for all walks of $\mathcal{G}_2$ be eliminated.  

By the above argument, any random walk $(X_t)_{t=0}^{\tsep+2 \thit-1}$ in $\mathcal{G}_2$ will be good with probability at least $1-n^{-2}$. Hence by Markov's inequality, all trajectories of the random walks in $\mathcal{G}_2$ are good with probability at least $1-|\mathcal{G}_2|n^{-2}$. Conditioning on this event, \eq{goodd} shows that any from the random walks in $\mathcal{G}_1$ collides with the trajectory of any fixed good random walk in $\mathcal{G}_2$ in the time-interval $[\tsep, \tsep + 2 \thit-1]$ with probability at least $p$. By independence of these events across random walks in $\mathcal{G}_1$, random walk $(X_t)_{t=0}^{\tsep+2\thit-1}$ is not eliminated with probability at most
\[
  \left(1 - p \right)^{|\mathcal{G}_1|}  = \left(1 - p \right)^{2 \frac{1}{p} \cdot \ln n} \leq n^{-2}.
\]
Note, that we neglected the fact that random walks of $\mathcal{G}_2$ can eliminate each other, which only further decreases the probability of a walks of $\mathcal{G}_2$ to survive.
Combining everything, and using Union bound, it follows that with probability at least $1- n^{-1}$ all random walks in $\mathcal{G}_2$ are eliminated. 
The claim follows by noting that the number of steps used is $\tsep + 2\thit =O(\tmix+\thit)=O(\thit)$.
\end{proof}

\subsection{Reducing the Walks from  $\log^4 n$ to $(\Delta/d)^{O(1)}$ in $O(\thit)$}\label{sec:trickyred}

\begin{theorem}\label{thm:keylemma}
Let $G=(V,E)$ be any graph with maximum degree $\Delta$ and average degree $d$. Then the expected time to reduce the number of walks from $\log^4 n$ to $(\Delta/d)^{100}$ is at most $O(\thit)$. 
\end{theorem}

\subsubsection{Proof Overview}

In comparison with \autoref{thm:mostgeneral}, the reduction in \autoref{thm:keylemma} is more subtle, as there might be a sub-logarithmic number of random walks preventing us from using the simple bucketing-argument into ``nearly-regular'' partitions used in \autoref{thm:mostgeneral}. Furthermore, \autoref{thm:mostgeneral} achieved the reduction in just a single phase of $\Theta(\thit)$ steps:
All random walks have w.h.p.  a  distribution of visits to nodes  which is reasonably close to the expectation of visits to these nodes when starting from the stationary distribution.

Here however, we are only able to prove an ``exponential'' progress and
consider periods which can be much shorter than the mixing time.
This means that we need to cope with random walks whose distribution may be far from the stationary distribution. Specifically, if there are $k$ random walks left, we will analyze a phase of length $\approx \thit/\kappa$, with $k=\kappa^{100}$, and show that a constant fraction of random walks will be eliminated. To account for the fact that the random walks are not mixed, we will identify certain ``dense'' subsets $D_0$ having the crucial property that each node in $D_0$ has a sufficiently large stationary mass and all nodes together have a stationary mass which is close to $1$.

We then show the existence of a subset of $D_1(i) \subseteq D_0$ which random walk $i$ will pay enough visits to within $\thit/\kappa$ steps (see first part of the proof of \autoref{lem:structure}). This is derived via our new concentration inequality (\autoref{lem:return}) to show that $(i)$ random walk $i$ does not spend too many steps outside $D_0$ and $(ii)$ most vertices do not receive much more visits than predicted by the stationary distribution. Thus we end up in a favorable situation where for most walks $1 \leq i \leq k$, we have a subset $D_1(i) \subseteq D_0$ with $|D_1(i)| \geq 2 n / \kappa^{8}$. Since we have $\approx k=\kappa^{100}$ of such walks, an overwhelmingly large fraction of these subsets $S(i)$, $1 \leq i \leq k$, must overlap. 

Unfortunately, we are still not done since in order to reduce the number of random walks, we also need to consider {\em when} the visits to $D_1(i)$ occur. Specifically, if random walk $i$ makes a visit to a vertex $u \in D_1(i)$ at time, say, $t$, then we need to ensure that there are enough other random walks $j$ which could potentially also visit vertex $u$ at time $t$. To ensure this, we will discard ``surprising'' visits, which are visits to vertices when the probability for this to happen at this step or before is at most $\kappa^{-23}$. It is worth pointing out that the property of a visit to $u$ being surprising, depends not only on the vertex but also on the start vertex of the walk. The second part of the proof of \autoref{lem:structure} deals with this issue and shows that for most walks, there is a subset $D_2(i) \subseteq D_1(i)$ with $|D_2(i)| \geq n/\kappa^8$ containing only vertices which receive enough ``unsurprising'' visits.

Equipped with these subsets $D_2(i)$, we regard the ``unsurprising'' visits as a balls-into-bins configuration, where each ball on a bin (vertex) is associated to a walk $i$ which may visit this vertex (we refer to \autoref{fig:ballsandbins} for an illustration). Through a series of counting arguments \autoref{lem:stocknockout}, we establish that for most random walk $i$ there is a subset $D_4(i) \subseteq D_2$ of vertices, so that each vertex receives enough visits and for each such visit at some time $t$, there are sufficiently many other walks $j$ which have a probability of at least $\kappa^{-23}$ each to visit the same vertex at time $t$. 

After all these preparations, we analyze the coalescing process and achieve the desired reduction in the number of the random walks in \autoref{lem:thelastone}. Similarly to previous analyses, we use a division of random walks into groups $\mathcal{G}_1$ and $\mathcal{G}_2$. The roles of $\mathcal{G}_1$ and $\mathcal{G}_2$ are as before; walks in $\mathcal{G}_1$ are merely used to eliminate walks in $\mathcal{G}_2$. This time, however, the division into $\mathcal{G}_1$ and $\mathcal{G}_2$ is completely uniformly at random, in particular, this means that $\mathcal{G}_1$ and $\mathcal{G}_2$ are roughly of the same size. We establish that for most fixed random walks $1 \leq i \leq k$, conditional on being in group $\mathcal{G}_2$, there is a constant probability of picking a trajectory that will likely lead to an intersection with any of the other $k-1$ unexposed random walk.

Combining the two steps of the proof, the structural result in \autoref{lem:structure} with the probabilistic analysis in \autoref{lem:stocknockout}, it immediately follows that the number of walks can be reduced by a constant factor within $O(\thit/\kappa	)$ steps, yielding \autoref{thm:keylemma}.

\subsubsection{Definitions and Lemmas required to prove \autoref{thm:keylemma}}

Before giving the formal proof of \autoref{thm:regular}, we introduce additional notation. Recall that $k$ is the number of random walks at a certain time, w.l.o.g., say $t=0$. Consider a fixed random walk $(X_t)_{t=0}^{\infty}$ with label $1 \leq i \leq k$, where $k = \kappa^{100}$ that is run for 
\[\tau:=4 \thit/\kappa\] steps, and starts at an arbitrary vertex $X_0=u_i$. 
Since we seek to reduce the number of random walks to 
$(\Delta/d)^{100}$, we assume the following.
\begin{assumption} \label{ass:funnyk}
	Throughout this section we assume $k^{1/100} = \kappa > \max\{ 2^{10},\Delta/d\}$.
\end{assumption} 
Note that if $\kappa \leq \max\{ 2^{10},\Delta/d\}$ the claim follows immediately from \autoref{lem:beer}.
We define a ``dense'' subset of nodes as
\[
  D_0 := \left\{ v \in V \colon \pi(v) \geq 1/(n \kappa^2)  \right\}.
\]
Clearly, 
\begin{align}
\pi(D_0) \geq 1 - |V \setminus D_0| \cdot 1/(\kappa^2 n) \geq 1 - \frac{1}{\kappa^2}. \label{eq:SAP}
\end{align}

\begin{align}
\pi(u) \leq \pi_{\max} \leq \frac{\Delta}{nd} \leq k^{1/100} \cdot \frac{1}{n} = \kappa \cdot \frac{1}{n}, \label{eq:oracle}
\end{align}
where the penultimate inequality holds since $k^{1/100} > \Delta/d$ by assumption.
Hence the degree of any two vertices in $D_0$ differ by at most a factor of $\kappa^3$.

Before proceeding further, we introduce another piece of notation. For any random walk $1 \leq i \leq k$ denoted by $(X_t)_{t \geq 0}$, we call a time-step $s \in [0,\tau-1]$ is {\em bad} if $X_s \in D_0$ and additionally, \GNOTE{T: The exponent $10$ below is extremely generous...}
\[
 \sum_{t=s}^{\tau-1} \mathbf{1}_{X_{t} \in D_0} \cdot p_{X_s,X_{t}}^{t-s} \geq 16 \cdot \kappa^{10} \cdot \thit/n.
\] 
Intuitively, a time-step $s$ is bad, if the expected number of collisions for another random walk starting at vertex $X_s \in D_0$ at step $s$ with the walk $(X_s,X_{s+1},\ldots)$ is too large. 

\begin{lemma}\label{lem:goodsteps}
Consider a random walk $(X_t)_{t=0}^{\tau-1}$ of length $\tau=4\thit/\kappa$.
 Then with probability at least $1-2^{-\kappa}$, there are at most $\tau \cdot 2^{-\kappa}$ bad time-steps $t \in [0,\tau]$.
 Consequently, for a collection of $\kappa^{100}$ random walks with $\kappa \geq 2^{10}$, all of these walks
 have at most  $\tau \cdot 2^{-\kappa}$ bad time-steps $t \in [0,\tau-1]$ w.p. at least $1-1/\kappa$.
\end{lemma} 
\begin{proof}
First, let us fix any step $s \in [0,\tau-1]$, and following the notation of \autoref{lem:deterministicpath}, let
\begin{align*}
  \tilde{Z}(s) := \sum_{t=s}^{\tau-1} \mathbf{1}_{X_t \in D_0}  \cdot p_{X_s,X_{t}}^{t-s}.
\end{align*}
Then, since $\tau - s \leq \thit$, by \autoref{lem:deterministicpath}, for any vertex $u \in D_0$, 
\begin{align*}
 \E{ \tilde{Z}(s) \, \mid \, X_s = u } \leq 8 \kappa^3 \cdot \thit \cdot \max_{w \in D_0} \pi(w) \leq 8 \kappa^4 \cdot \thit/n =: \Upsilon.
\end{align*}
Since $\thit \geq n$ by \autoref{lem:fridaygift}, we have $\Upsilon \geq 2$ and the concentration inequality in \autoref{lem:deterministicpath} implies
\begin{align*}
 \Pr{ \tilde{Z}(s) \geq 16 \cdot \kappa^{10} \cdot \thit / n } &\leq \Pr{ \tilde{Z}(s) \geq 2 \kappa \cdot (2 \Upsilon + 1) } \leq 2^{-2 \kappa}.
\end{align*} 
Now let $B$ denote the number of bad time-steps, \ie
\begin{align*}
 B := \left| \left\{ 0 \leq s \leq \tau -1 \colon \tilde{Z}(s) \geq 16 \cdot \kappa^{10} \cdot \thit/n \right\} \right|.
\end{align*}
Then, by linearity of expectation
\begin{align*}
 \E{ B }=\sum_{s=0}^{\tau-1}  \Pr{ \tilde{Z}(s) \geq 16 \cdot \kappa^{10} \cdot \thit / n }  \leq \tau \cdot 2^{-2\kappa},
\end{align*}
and a simple application of Markov's inequality implies the first part of the claim. For the second part we simply take Union bound over all $\kappa^{100}$ walks and using that 
$1- \kappa^{100}/2^{\kappa} \geq 1-1/\kappa$.
\end{proof}

Let $N_t(i,v)$ denote the number of visits of the random walk $i$ to $v$ within $t$ time-steps.
Let
\[
D_1(i) := \left\{ v \in D_0\colon N_{\tau}(i, v) \geq 2\tau / \kappa^4  \cdot \pi(v)\right\},
\]
 \ie $D_1$ are all vertices $v \in D_0$ that are visited at least $2\tau/\kappa^4 \cdot \pi(v)$ times before time-step $\tau$. Notice that $D_1$  is a (random) set that depends on the realization of the walk.

Furthermore,
let 
\[ t(i,v) := \min \left\{ t \colon \Pr{ \hit(u_i,v) \leq t } \geq \kappa^{-23} \right\}. \]
Basically $t(i,v)$ is the ``smallest'' step $t$ so that the probability that random walk $i$ visits vertex $v$ at step $t$ or earlier is bounded below by $\kappa^{-23}$. Note that $t(i,v)$ is a deterministic integer that does not depend on the realization of the walk. With reference to our proof outline, we regard any visit before $t(i,v)$ as a ``surprising'' visit, while visits at step $t(i,v)$ or later as an ``unsurprising'' visit.

Let $D_2(i)$  be the set of vertices $v$ that get at least $\tau/\kappa^4 \cdot \pi(v)$ visits between the time steps $t(i,v) $ and $\tau$; in symbols 
\begin{equation}\label{eq:defS1} 
D_2(i) := \left\{ v \in D_1(i)\colon  N_{\tau}(i,v) - N_{t(i,v)}(i,v) \geq \tau/\kappa^4 \cdot \pi(v)
 \right\}.
\end{equation}

 We now to state a structural lemma, providing lower bounds on the size of $D_2(i)$.
 Recall that $D_2(i)$ is a subset of the  ``dense'' set $D_0$ that has a large stationary mass and contains only vertices with sufficiently high degree. This ``projection'' is not required on regular graphs, where we could simply work with all vertices, \ie $D_0=V$. However, for non-regular graphs, the projection on $D_0$ is essential since on the set $D_0$, the random walk will behave sufficiently similar to a random walk on a regular graph. 
\begin{lemma}\label{lem:structure}
Let $G=(V,E)$ be an arbitrary graph and 
 let $\kappa \geq 2^{10}$ be any integer.
Consider any random walk $(X_t)_{t \geq 0}$ with label $i$. Then, we have that $|D_2(i)| \geq n/\kappa^8$
w.p.~at least $1-6/\kappa$.
\end{lemma}

\begin{proof}[Proof of \autoref{lem:structure}]
\newcommand{\pimax}{\pi_{\max}}
First we bound the number of visits of walk $i$ to $V \setminus D_0$. To this end, let $\tilde{C}$ be the hits
from $i$  to the vertices which are not in $D_0$ before time step $\tau$, in symbols,
\[
\tilde{C} := \sum_{t=0}^{\tau-1} \mathbf{1}_{X_t \not\in D_0}.
\]
By \autoref{lem:return} and using $\tau \leq \thit$, we derive 
\begin{align*}
 \E{\tilde{C}} \leq 8 \max\{\thit,\tau\} \cdot \pi(V \setminus D_0) \leq 8 \thit/\kappa^2 = 2 \tau/\kappa,
\end{align*}
where we used the fact that $\pi(V \setminus D_0) \leq 1 / \kappa^2$ by \eq{SAP}.
Hence by Markov's inequality,
\begin{align}\label{eq:asd1}
 \Pr{ \tilde{C} \geq \tau/2} \leq \Pr{ \tilde{C} \geq \kappa/4 \cdot \E{\tilde{C}}} \leq \frac{4}{\kappa}.
\end{align}

Next for any fixed vertex $v \in V$, we know that for any $\lambda \geq 1$ the probability that the random walk makes more than $32 \lambda \cdot \thit \cdot\pi(v) $ visits to $v$ is 
\begin{equation}\label{eq:toursaintjaques}
	\Pr{N_\tau(i,v) \geq 32\lambda \cdot \thit \cdot \pi(v)}  \leq 	\Pr{N_\tau(i,v) \geq \lambda \cdot (16 \cdot \thit \cdot \pi(v) +1 )}  \leq  2^{-\lambda},
\end{equation}
  where we used the facts that $\thit \cdot \pi(v) \geq 1$ (\autoref{lem:fridaygift}) and the second statement of  \autoref{lem:return}.
Recall that $N_t(i,v)$ denotes the number of visits of the random walk $i$ to $v$ within $t$ time-steps. Define
\[
  B := \sum_{v \in V} \mathbf{1}_{N_{\tau}(i, v) \geq \kappa^2 \cdot \thit\cdot \pi(v) } \cdot N_{\tau}(i, v).
\]
Then  

\begin{align*}
 \E{B} &= \sum_{v \in V} \sum_{j=\kappa^2 \cdot \thit\cdot\pi(v)}^{\infty} j \cdot \Pr{ N_{\tau}(i, v) = j} \\
&\leq \sum_{v \in V} \sum_{\sigma= \kappa^2}^\infty 
 ( \sigma + 1)\cdot \thit\cdot\pi(v) \cdot
\sum_{j= \sigma\cdot \thit\cdot\pi(v)}^{( \sigma + 1)\cdot \thit\cdot\pi(v)-1} \Pr{ N_{\tau}(i, v) = j} \\
 &\leq \sum_{v \in V} \sum_{\sigma=\kappa^2}^{\infty}  (\sigma+1) \cdot \thit \cdot \pi(v) \cdot \Pr{N_{\tau}(i, v) \geq \sigma \cdot \thit \cdot \pi(v)} \\
  &= \sum_{v \in V} \thit \cdot \pi(v) \cdot\sum_{\sigma=\kappa^2}^{\infty}  (\sigma+1) \cdot  \Pr{N_{\tau}(i, v) \geq \sigma \cdot \thit \cdot \pi(v)} \\
 &\stackrel{\eqref{eq:toursaintjaques}}{\leq} \sum_{v \in V} \thit \cdot \pi(v) \cdot \sum_{\sigma=\kappa^2}^{\infty }2 \sigma \cdot 2^{- \sigma / 32}   \\
  &=  \thit \cdot \sum_{\sigma=\kappa^2}^{\infty }2 \sigma \cdot 2^{- \sigma / 32}   \\
 &\leq  2 \thit \cdot \kappa^{-10} \leq \tau\cdot \kappa^{-9},
\end{align*}
where we used that $\kappa\geq 2^{10}$ and $\sum_{v\in V}\pi(v)=1$.
Hence, by Markov's inequality,
\begin{align}\label{eq:asd2}
 \Pr{B \geq \tau \cdot \kappa^{-8}} \leq \frac{1}{\kappa}.
\end{align}
Suppose now that $\left\{ \tilde{C} \leq \tau/2 \right\}$ and $\left\{ B \leq \tau \cdot \kappa^{-8} \right\}$ both occur.
Conditioning on $B \leq \tau / \kappa^{8}$ and $\tilde{C} \leq \tau/2$, we will show by pigeonhole principle \begin{equation}\label{eq:asdas}|D_1(i)| \geq 2n/\kappa^8. \end{equation}
Suppose for the sake of contradiction that $|D_1(i)|<2n/\kappa^8$. Then, using $ \pi_{\max} \leq \kappa/n$ by \eq{oracle},  we have that 
 the total number of visits to nodes in $D$ is at most 
 \begin{align*}
 \sum_{u \in D_0 } N_{\tau}(i,u) &=
 \sum_{u \in D_0 \setminus D_1(i)} N_{\tau}(i,u) + \sum_{u \in D_1(i)} N_{\tau}(i,u)\\ 
 &\leq  \sum_{u \in D_0 \setminus D_1(i)} 2\pi(u) \cdot \tau/\kappa^4  + \sum_{u \in D_1(i) \colon N_{\tau}(i,u) \leq \kappa^2 \cdot \thit \cdot \pi(u)} N_{\tau}(i,u) + \sum_{u \in D_1(i) \colon N_{\tau}(i,u) > \kappa^2 \cdot \thit \cdot \pi(u)} N_{\tau}(i,u) \\
 &\leq\sum_{u \in D_0 \setminus D_1(i)} 2\pi(u) \cdot \tau/\kappa^4  +  |D_1(i)|\cdot \max_{u \in D_0} \pi(u) \cdot\thit\cdot \kappa^2 + B  \\ 
 &\leq 2(|D_0|-|D_1(i)|) \cdot \max_{u \in D_0} \pi(u) \cdot \tau/\kappa^4 +  |D_1(i)|\cdot \max_{u \in D_0} \pi(u) \cdot\thit\cdot \kappa^2 + B  \\ 
&\leq 2(|D_0|-|D_1(i)|) \cdot (\kappa/n) \cdot \tau/\kappa^4 +  |D_1(i)|\cdot (\kappa/n) \cdot\thit\cdot \kappa^2 + B  \\ 
&< 2(|D_0|-|D_1(i)|)\cdot \tau/(n\cdot \kappa) + 2n/\kappa^8 \cdot (\kappa/n) \cdot \thit \cdot \kappa^2+B\\
&\leq 2\tau/\kappa + 2n/\kappa^8 \cdot (\kappa/n) \cdot \tau \cdot \kappa \cdot \kappa^2
+ \tau / \kappa^8  \\ 
&= 2\tau/\kappa + \tau/\kappa^2 + \tau/\kappa^8 \leq \frac{1}{2} \tau.
 \end{align*}
Thus, $\tilde{C} > \tau - \frac{1}{2} \tau = \frac{1}{2} \tau$ which is a contradiction to the assumption that the event $\left\{ \tilde{C} \leq \tau/2 \right\}$ occurs.

Finally, we will upper bound the number of ``surprising'' visits, which are visits to vertices that happen too early. That is, we will upper bound the number of visits to vertices $v$ before time $t(i,v)$; in symbols, \begin{align*}
 \tilde{B} := \sum_{v \in V} \sum_{0 \leq t < t(i,v)} \mathbf{1}_{X_t = v}.
\end{align*}
By definition of $t(i,v)$, with probability at most $\kappa^{-23}$ the random walk visits the vertex $v$ before $t(i,v)$. Conditional on this event occurring, the expected number of visits is at most $\sum_{t=0}^{t(i,v)} p_{v,v}^t$.
Hence by linearity of expectations,
\begin{align*}
 \E{ \tilde{B} } \leq  \sum_{v \in V} \left( \kappa^{-23} \cdot \sum_{t=0}^{t(i,v)-1} p_{v,v}^{t}.
 \right)
\end{align*}
Since $t(i,v) \leq 2 \thit$, it follows by the first statement of \autoref{lem:return} that $\sum_{t=0}^{t(i,v)-1} p_{v,v}^{t} \leq 16 \thit \cdot \pi(v)$ and hence \GNOTE{T: in this part of the proof, we are extremely wasteful with powers of $\kappa$...}
\begin{align*}
 \E{ \tilde{B} } &\leq \kappa^{-23} \cdot 16  \thit \cdot \sum_{v\in V}\pi(v) \leq  \tau/2 \cdot \kappa^{-16},
\end{align*}
and thus by Markov's inequality,
\begin{align}\label{eq:asd3}
 \Pr{ \tilde{B} \geq \tau/2 \cdot \kappa^{-15} } \leq \Pr{ \tilde{B} \geq \kappa \cdot \E{ \tilde{B} } } \leq \frac{1}{\kappa}. 
\end{align}
Hence, the total number of visits to vertices in $V$ before time $t(i,v)$ is at most $\tau/2 \cdot \kappa^{-15}$ with probability at least $1-1/\kappa$. 

 Hence, by \eqref{eq:asd1}, \eqref{eq:asd2}, and \eqref{eq:asd3} and by taking Union bound, we have \begin{align*} p&:=\ \Pr{ \left\{ \tilde{C} \geq \tau/2 \right\} \cap \left\{ B \geq \tau \cdot \kappa^{-8} \right\} \cap  \left\{  \tilde{B} \geq \tau/2 \cdot \kappa^{-15} \right\}     } \\
 &\geq 1 - \Pr{ \left\{ \tilde{C} \geq \tau/2 \right\} \cup \left\{ B \geq \tau \cdot \kappa^{-8} \right\} \cup  \left\{  \tilde{B} \geq \tau/2 \cdot \kappa^{-15} \right\}     } \\
 &\geq 1- \frac{6}{\kappa}	.
 \end{align*}
In particular, w.p. $p$ and by \eqref{eq:asdas},  we have $|D_1(i)| \geq 2n/\kappa^8$. 
Observe that, by definition of the sets, each vertex $v \in D_1(i) \setminus  D_2(i)$ is visited at least 
$2 \tau/\kappa^4 \cdot \pi(v) - \tau/\kappa^4 \cdot \pi(v) = \pi(v) \cdot \tau/\kappa^4$ times before time-step $t(i,v)$
and thus,
\[|D_2(i)| = |D_1(i)| - | D_1(i) \setminus D_2(i) | 
\geq |D_1(i)| -  \frac{\tilde{B}}{  \min_{w \in S}\pi(w) \cdot \tau/\kappa^4  }  \geq \frac{2n}{\kappa^8}-\frac{ \tau/2 \cdot \kappa^{-15}}{1/(\kappa^3 n) \cdot \tau/\kappa^4  } \geq \frac{n}{\kappa^8}. \]  
\end{proof}

The next step in the proof is to elaborate on the sets $D_2(i)$ from \autoref{lem:structure} in order to analyze collisions on this set. Before doing this, we need to introduce additional notation in order to define a balls-into-bins configuration.

Let us denote the random walk with label $i$ by $(X^i_t)_{t=0}^{\tau}$. The random walk may start from an arbitrary vertex $X^i_0=u_i$ and is run for $\tau= 4 \thit / \kappa$ steps.
Recall 
\[ t(i,v) := \min \left\{ t \colon \Pr{ \hit(u_i,v) \leq t } \geq \kappa^{-23} \right\}.\]

We now consider the following balls-into-bins configuration, where we emphasize that the balls-into-bins configuration is completely deterministic (for fixed start vertices at time $0$) and does not depend on the realization of any of the random walks. Every vertex in $V$ corresponds to a bin. For every walk $j$ and $u \in D_0(j)$, we place deterministically a ball with label $(j,t(j,u))$ into bin $u$. We call a ball with label $(j,t(j,u))$ in bin $u$ {\em bad} if there are fewer than $\kappa^{55}$ other balls $i$ in the same bin such that either $(i)$ $t(i,u) < t(j,u)$ or $(ii)$ $t(i,u) = t(j,u)$ and $i<j$. 
Next define a random walk $j$ to be \emph{bad} if at least $n/\kappa^{9}$ bad balls have label $j$ and otherwise we call $j$ \emph{good}. Since there are at most $  \kappa^{55} \cdot n$ bad balls, it follows that the number of {\em bad} walks $j \in \{1,\ldots,k\}$  is at most 
\begin{align}\label{eq:numberofbadballs}
   \frac{\kappa^{55} \cdot n}{n/\kappa^{9}} = \kappa^{64}.	
\end{align}

\begin{figure}
	\rowcolors{2}{}{testcolor!50}
\begin{center}
	\begin{tikzpicture}[scale=1.7,knoten/.style={fill=black,circle,scale=0.34},edge/.style={thick,draw=black!50,scale=0.1},redge/.style={line width=2pt,draw=blue,scale=0.5}, rknoten/.style={fill=red!50!white, circle, scale=0.31}, gknoten/.style={fill=green!50!white, circle, scale=0.31}]

\draw (-0.19,2) to (-0.15,1) to (0.15,1) to (0.19,2);
\node () at (0,0.75) {$v_1$};

\draw (0.31,2) to (0.35,1) to (0.65,1) to (0.69,2);
\node () at (0.5,0.75) {$v_2$};

\draw (0.81,2) to (0.85,1) to (1.15,1) to (1.19,2);
\node () at (1,0.75) {$v_3$};

\draw (1.31,2) to (1.35,1) to (1.65,1) to (1.69,2);
\node () at (1.5,0.75) {$v_4$};

\draw (1.81,2) to (1.85,1) to (2.15,1) to (2.19,2);
\node () at (2,0.75) {$v_5$};

\draw (2.31,2) to (2.35,1) to (2.65,1) to (2.69,2);
\node () at (2.5,0.75) {$v_6$};

\draw [dashed] (-0.22,1.25) -- (2.76,1.25);

\node[rknoten] (1) at (0,1.125) {$(1,4)$};
\node[gknoten] (1) at (0,1.375) {$(2,4)$};
\node[gknoten] (1) at (0,1.625) {$(3,5)$};

\node[rknoten] (1) at (0.5,1.125) {$(4,2)$};
\node[gknoten] (1) at (0.5,1.375) {$(1,3)$};
\node[gknoten] (1) at (0.5,1.625) {$(3,3)$};

\node[rknoten] (1) at (1,1.125) {$(4,5)$};

\node[rknoten] (1) at (1.5,1.125) {$(3,2)$};
\node[gknoten] (1) at (1.5,1.375) {$(1,3)$};
\node[gknoten] (1) at (1.5,1.625) {$(4,3)$};
\node[gknoten] (1) at (1.5,1.875) {$(2,4)$};

\node[rknoten] (1) at (2,1.125) {$(4,4)$};
\node[gknoten] (1) at (2,1.375) {$(2,5)$};
\node[gknoten] (1) at (2,1.625) {$(1,6)$};

\node[] at (6.05,1.5) {
\begin{minipage}{0.6\textwidth}
\scalebox{0.85}{
\begin{tabular}{|c|c|c|c|c|c|c|}
\hline
$i$ & $t(i,v_1)$ & $t(i,v_2)$ & $t(i,v_3)$ & $t(i,v_4)$ & $t(i,v_5)$ & $t(i,v_6)$ \\
\hline
1 & 4 & 3 & $>\tau$ & 3 & 6 & $>\tau$ \\
\hline
2 & 4 & $>\tau$ & $>\tau$ & 4 & 5 & $>\tau$ \\
\hline
3 & 5 & 3 & $>\tau$ & 2 & $>\tau$ & $>\tau$ \\
\hline
4 & $>\tau$ & 2 & 5 & 3 & 4 & $>\tau$ \\
\hline
\end{tabular}
}

\end{minipage}
};

\end{tikzpicture}
\end{center}
\caption{Illustration of the balls-into-bins configuration of $k=4$ walks labeled $1,2,3,4$ into $6$ vertices $v_1,v_2,\ldots,v_6$. In the illustration, a ball $(i,t(i,u))$ on bin $u$ is good if there is at least one other ball with $t(i,u) < t(j,u)$ (or $t(i,u)=t(j,u)$ and $i<j$). For each of the random walks $1$, $2$ and $3$, there are at least two good balls, while for random walk $4$ only one ball is good.}
\label{fig:ballsandbins}
\end{figure}
In the following, we will focus on the $\kappa^{100}-\kappa^{64}$ good walks and ignore all other walks. Recall that any fixed good walk $i$ has at most $n/\kappa^{9}$ bad balls. 
We now make another central definition of a random subset:
\[
 D_4(i) := \left\{ v \in D_0 \colon \text{the ball  $(i,t(i,v))$ is good and $N^G(i,v) \geq \tau/(\kappa^9 n)$} \right\},
\]
where $N^G(i,v)$ denotes the number of times $v \in D_0$ is visited by walk $i$ on a good time step in the interval $[t(i,v),\tau-1]$. Intuitively, every such visit of a random walk $i$ to a vertex $v \in D_4(i)$ at a time $t$ is very helpful for the following reason: Since the ball $(i,t(i,v))$ is good, there are at least $\kappa^{55}$ other random walks $j \neq i$ with $t(j,v) \leq t(i,v) \leq t$, and thus each walk $j$ has a probability of at least $\kappa^{-23}$ to visit vertex $v$ at a time $t(j,v)$ and potentially collide with random walk $i$ at time $t$ later.\GNOTE{T: Tried to make the use of the time steps more clear. Also the probability for a collision is of course not $\kappa^{-23}$!!} The next lemma provides a lower bound on the size of $|D_4(i)|$.

\begin{lemma}\label{lem:stocknockout} 
Consider any random walk $(X_t)_{t \geq 0}$ with label $i$. Then $\Pr{ |D_4(i)| \geq n/(4 \kappa^8)} \geq 1-8/\kappa$.
\end{lemma}
\begin{proof}
Recall that, (see  \eqref{eq:defS1}) 
\[ D_2(i) = \left\{ v \in D_1(i) \colon  N_{\tau}(i, v) - N_{t(i,v)}(i,v) \geq \tau/\kappa^4 \cdot \pi(v) \right\}.\]
Let us now define
\begin{equation}\label{eq:defofS2}
D_3(i) = \{ v \in D_2(i) \colon \text{the ball  $(i,t(i,v))$ is good }  \}.
\end{equation}
  By \autoref{lem:structure}, we have $\Pr{|D_2(i)| \geq n/\kappa^8   } \geq 1-6/\kappa$. In the following assume that the event $\left\{ |D_2(i)| \geq n/\kappa^8 \right \}$ occurs.
 Since by definition a good walk has fewer than $n/\kappa^9$ bad balls, we have 
\begin{equation}\label{eq:sizeofS2}
 |D_3(i)|=	|D_2(i)| - \text{number of bad balls of $i$} \geq n/\kappa^8 - n/\kappa^9 \geq \frac{n}{2\kappa^8} 
\end{equation}
Further, by definition of $D_2(i)$   each vertex in $v \in D_3(i)$ is visited at least 
\begin{equation}\label{eq:visitsofS2}
	\tau/\kappa^4 \cdot \pi(v) \geq \tau/(\kappa^7 n)
\end{equation}
 times during the interval $[t(i,v),\tau-1]$, where the inequality is due to the definition of $D_0$ and the fact that $D_2(i) \subseteq D_0$.
 We now define the following random variables. Let 
\begin{enumerate}
	\item $N^G(i,u)$ be the number of times  $u\in D_3(i)$ is visited by walk $i$  on a good time step in the interval $[t(i,u),\tau-1]$ (as defined previously).
 \item $N^B(i,u)$  be the number of times  $u\in D_3(i)$ is visited by walk $i$ on a bad time step in the interval $[t(i,u),\tau-1]$.
	\item 	 $N^E(i)$ be the set of nodes  $u\in D_3(i)$ that are visited by walk $i$  prior to $t(i,u)$, \ie 
 $N^E(i) := \{ u \in D_3(i) \colon \hit(u_i,v) < t(i,v) \}$.	
\end{enumerate}
 We have \begin{align*}
 	 \E{N^E(i)} &= \E{ \sum_{v\in D_3(i)} \mathbf{1}_{\hit(u_i,v) < t(i,v) } } 
 	 = \sum_{v\in D_3(i)} \E{ \mathbf{1}_{\hit(u_i,v) < t(i,v)}}\\
 	&= \sum_{v\in D_3(i)} \Pr{ \hit(u_i,v) < t(i,v) }
 	 \leq \sum_{v\in D_3(i)} 1/\kappa^{23} = |D_3(i)|/\kappa^{23},\\
 \end{align*}
 where the inequality comes from the definition of $t(i,v)$. By Markov inequality,
 \begin{equation}\label{eq:boundearlyvisits}
 	\Pr{N^E(i) \leq \frac{|D_3(i)|}{\kappa^{22}} } \geq 1-\frac{1}{\kappa}.
 \end{equation}

  In the remainder we condition on
  $N^E(i) \leq \frac{|D_3(i)|}{\kappa^{22}} $.
By \autoref{lem:goodsteps},  with probability at least $1-\kappa$,
all random walks have at most $ \tau \cdot 2^{-\kappa} \leq \thit/\kappa^{20}$ bad time-steps $s \in [0,\tau]$,
where we recall
 $t$ is \emph{bad} if $X_s \in D_0$ and $
 \sum_{t=s}^{\tau} \mathbf{1}_{X_{t} \in D_0} \cdot p_{X_s,X_{t}}^{t-s} \geq 16 \cdot \kappa^{10} \cdot \thit/n.$
In the following we condition on the number of bad time steps being bounded by  $\thit/\kappa^{20}$.
   We claim that \GNOTE{T: Since I have restructured everything, notice that now $D_4(i)$ is no longer a subset of $D_3(i)$ and $D_2(i)$!}
   \GNOTE{F: Then the index is confusing. Are you worried about the implicit conditioning or is the redefinition due to the
   fact that $D_3$ is only defined in the proof? If the latter is the case, then I opt for defining $D_3$ before and to keep to subset structure}
\begin{equation}\label{eq:sizeS3claim} |D_4(i)|\geq |D_3(i)|/2.\end{equation}
 Assume, for the sake of contradiction, that 
$|D_4(i)| < \frac{1}{2} |D_3(i)|$.
 Let 
 \[
 \tilde{D} := \{ v \in D_3(i) \setminus D_4(i)  \colon \hit(u_i,v) \geq t(i,v)\}.
 \]
We have, using  $|D_3(i) \setminus D_4(i)| \geq |D_3(i)|-|D_4(i)| \geq \frac{1}{2} |D_3(i)|$ that
\[|\tilde{D} | \geq |D_3(i) \setminus D_4(i) | - N^E(i) \geq \frac{1}{2} |D_3(i)|- \frac{|D_3(i)|}{\kappa^{22}} \geq  \frac{1}{4} |D_3(i)|.\]%
For each vertex $u \in \tilde{D} \subseteq D_3(i) \setminus D_4(i) $, we have
$N^G(i,u) < \tau/(2n\kappa^9) $ and thus 
\begin{align*}
 N^B(i,u) &\stackrel{u\in \tilde{D}}{=} N_\tau(i,u) - N_{t(i,v)}(i,u) - N^G(i,u) \\& \geq \tau / \kappa^4 \cdot \pi(u) - \tau(2n \kappa^9) \geq \tau/(\kappa^7 n)-  \tau/(2n\kappa^9) \geq \tau/(2n\kappa^7),\end{align*}
where the first inequality follows from \eqref{eq:visitsofS2} and  $u\in D_2(i)$.
In words,
  at least $\tau/(2n\kappa^7)$ visits to $u$ happened on a bad time step during the interval $[t(i,u), \tau]$. 
Thus, the total number of visits to nodes of $\tilde{D}$ at bad time steps is at least 
\begin{align*}
	\sum_{v\in \tilde{D}} N^B(i,v)  &\geq
	|\tilde{D}|\cdot \tau/(2n\kappa^7) 
	\geq
  \frac{1}{8} |D_3(i)| \cdot \frac{\tau}{n\kappa^7} 
  \stackrel{\eqref{eq:sizeofS2}}{>}\frac{\thit}{\kappa^{20}}.
\end{align*}
This contradicts the assumption that there are at most $\thit/\kappa^{20}$ bad time steps in total.
Thus, \eqref{eq:sizeS3claim} holds and we derive using \eqref{eq:sizeofS2}
\begin{equation}\label{eq:sizeS3} |D_4(i)|\geq n/(4\kappa^{8}).\end{equation}
As shown above, this  lower bound on $|D_4(i)|$ holds whenever the following three events all occur: $(i)$ $|D_2(i)| \geq n/\kappa^8$ (which holds with probability at least $1-6/\kappa$ by \autoref{lem:structure}), $(ii)$  $N^E(i) \leq \frac{|D_3(i)|}{\kappa^{22}} $ occurs (which holds with probability $1-1/\kappa$ by \eqref{eq:boundearlyvisits}) and $(iii)$ the number of bad time steps of random walk $i$ is at most $\thit/\kappa^{20}$ (which holds with probability $1-1/\kappa$ by \autoref{lem:goodsteps}). Hence by the Union bound, 
\[
 \Pr{ |D_4(i)| \geq n/(4 \kappa^8) } \geq 1 - 6 /\kappa - 1 / \kappa - 1 / \kappa = 1  - 8 /\kappa.
\]
\end{proof}

The previous lemma established that with reasonably large probability, any fixed good random walk $i$ satisfies $|D_4(i)| \geq n/\kappa^{10}$. In the next lemma we show that, conditioning on this event occurring, that random walk $i$ is eliminated by any of the other random walks with some constant probability $>0$.
\begin{lemma}\label{lem:thelastone}
Assume that a good random walk $i$ has a trajectory $(x_0=u_i,x_1,\ldots,x_{\tau-1})$ satisfying $|D_4(i)| \geq n/\kappa^{10}$. Then random walk $i$ will be eliminated before time $\tau$ with probability at least $1/10$.
\end{lemma}
\begin{proof}
Consider now another random walk $j \neq i$ starting from an arbitrary vertex $u_j$. Define
\begin{equation}\label{eq:defS4}
	D_5(i,j) := \{ u \in D_4(i) \colon  t(j,u) \leq t(i,u) \}.
\end{equation}
With reference to \autoref{fig:ballsandbins}, $u \in D_5(i,j)$ if the green ball
$(i,t(i,u))$ lies above the  ball $(j,t(j,u))$.
Intuitively, $D_5(i,j)$ contains all the vertices in $u \in D_4(i)$ so that each time random walk $i$ visits $u$, also random walk $j$ could visit that vertex with sufficiently large probability.
 For  $u\in D_4(i)$ we have that $(u,t(i,u))$ is good
 and for each such good ball, by definition there  at least $\kappa^{55} $ other random walks $j$ such that $t(j,u) < t(i,u)$ (or $t(j,u)=t(i,u)$ and $j < i $). Hence by considering all bins we conclude
\begin{align}\label{eq:thomas}
\sum_{j=1,j \neq i}^{\kappa^{100}} |D_5(i,j)| \geq \kappa^{55} \cdot |D_4(i)| \geq \kappa^{55} \cdot n/(4\kappa^{8}) = n \kappa^{47}/4,
\end{align}
where the second inequality holds by our assumption $|D_4(i)| \geq n/(4 \kappa^8)$.

We are now in a position to apply our common analysis method.
We consider the $\pimortal$ process of \autoref{sec:process} and make use
of the majorization of \autoref{lem:majormotal}.
We assign each random walk $1 \leq i \leq k=\kappa^{100}$ into group $\mathcal{G}_1$ and $\mathcal{G}_2$ independently and uniformly at random. Recall that walks of $\mathcal{G}_1$ cannot be eliminated.
In the following, we define the following event $\mathcal{E}_i$ for walk $i$: $ \mathcal{E}_i = \{ i \in \mathcal{G}_2\} $. Clearly, $\Pr{ \mathcal{E}_i } = 1/2$. We will prove that conditional on this event occurring, random walk $i$ is eliminated by one of the walks in $\mathcal{G}_1$ with at least constant probability $>0$. 
 Let $Z(i,j)$ denote the number of collisions between random walk $i$, denoted by $(X_t)_{t\geq 0}$, and $j$ denoted by $(Y_t)_{t\geq 0}$, that happen on a vertex in $D_5(i,j)$ at a good time step, in symbols, 

\begin{align*}
Z(i,j) :=  \mathbf{1}_{i \in \mathcal{G}_2}\cdot  \mathbf{1}_{j \in \mathcal{G}_1} \cdot \sum_{u \in D_5(i,j)} \sum_{\substack{t(i,u) \leq t \leq \tau -1\colon \\ \text{$t$ is good  } }} \mathbf{1}_{X_t = u}\cdot \mathbf{1}_{Y_t =u}
\end{align*}
By conditioning on $\mathcal{E}_i$ and the trajectory of $(x_0,x_1,\ldots,x_{\tau-1})$ of the good random walk $i$,
 \begin{align*}
  \E{Z(i,j) \, \mid \, \text{trajectory $i=$}(x_0,x_1,\ldots,x_{\tau-1}), \mathcal{E}_i } &= \frac{1}{2} \cdot \sum_{u \in D_5(i,j)} \sum_{\substack{t(i,u) \leq t \leq \tau -1 \colon \\ \text{$t$ is good and } x_t=u }} p_{u_j,u}^{t}. 
\end{align*}
We now would like to derive a lower bound on $p_{u_j,u}^t$, where $t(i,u)\leq t\leq \tau-1$ is a good time-step with $X_t^i=u$, $u\in D_5(i,j)$.
By definition of $D_5(i,j)$ we have $t(i,u) \geq t(j,u)$.
By conditioning on the first visit of random walk $j$ to $u$, we obtain
\begin{align*}
 p_{u_j,u}^t &= \sum_{s=0}^{t} \Pr{ \hit(u_j,u) = s } \cdot p_{u,u}^{t-s}  \stackrel{\text{\autoref{lem:loop}}}{\geq} \pi(u) \cdot \sum_{s=0}^{t} \Pr{ \hit(u_j,u) = s } \\ 
 &\stackrel{u \in D_0}{\geq} 1/(\kappa^3 \cdot n) \cdot \Pr{ \hit(u_j,u) \leq t } 
 \stackrel{t \geq t(i,u) \geq t(j,u)}{\geq} 1/(\kappa^3 \cdot n) \cdot \Pr{ \hit(u_j,u) \leq t(j,u) } \\
 &\stackrel{\text{def.\ of $t(j,u)$ }}{\geq} 1/(\kappa^{26} \cdot n).
\end{align*} 
By definition of $D_4(i)$, for any vertex $u \in D_5(i,j) \subseteq D_4(i)$, $u$ is visited at least $\tau/(\kappa^9 n)$ times during the interval $[t(i,u), \tau-1]$. 
Therefore, 
	\begin{align*}
& \E{Z(i,j) \, \mid \, \text{trajectory $i=$}(x_0,x_1,\ldots,x_{\tau-1}),\mathcal{E}_i} \\
&\phantom{000}\geq \frac12 \cdot |D_5(i,j)| \cdot \frac{\tau}{\kappa^9 \cdot n} \cdot \frac{1}{\kappa^{26} \cdot n }  = \frac12 \cdot |D_5(i,j)| \cdot \frac{\tau}{  \kappa^{35}} \cdot \frac{1}{n^2}.
\end{align*}
Recall that if a time step
$s \in [0,\tau-1]$ is good (\ie not bad), then  $X_s \in D_0$ implies 
\[
 \sum_{t=s}^{\tau-1} \mathbf{1}_{X_{t} \in D_0} \cdot p_{X_s,X_{t}}^{t-s} < 16 \cdot \kappa^{10} \cdot \thit/n. 
\] 
Since $Z(i,j)$ sums only over good time steps and using $D_5(i,j)\subseteq D_0$ we conclude that
\begin{align*}
 &\E{Z(i,j) \, \mid \, \text{trajectory $i=$}(x_0,x_1,\ldots,x_{\tau-1}),\mathcal{E}_i, Z(i,j) \geq 1} \\
 &\leq \max_{s \colon X_s \in D_0} \E{\sum_{\substack{s\leq t \leq \tau -1\colon \\ \text{$t$ is good  } }} \mathbf{1}_{X_t \in D_5(i,j)} \cdot p^{t-s}_{X_s,X_t}} \leq 16 \kappa^{10} \cdot \frac{\thit}{n} 
\end{align*}
Combining the last two inequalities yields 
\begin{align*}
 &\Pr{ Z(i,j) \geq 1 \, \mid \, \text{trajectory $i=$}(x_0,x_1,\ldots,x_{\tau-1}),\mathcal{E}_i} 
 \\&\phantom{000}= \frac{\E{Z(i,j) \, \mid \, \text{trajectory $i=$}(x_0,x_1,\ldots,x_{\tau-1}),\mathcal{E}_i} }{ \E{Z(i,j) \, \mid \, \text{trajectory $i=$}(x_0,x_1,\ldots,x_{\tau-1}),\mathcal{E}_i, Z \geq 1}    }  \\
 &\phantom{000}\geq \frac{|D_5(i,j)|}{n \cdot \kappa^{47}}
\end{align*}
We are interested in the probability for $i$ being eliminated.
Neglecting the possibility that $i$ might even be eliminated by another rand walk of $\mathcal{G}_2$, which can only increase the probability of $i$ being eliminated, we derive
\begin{align*}
&\Pr{\text{Walk $i$ is eliminated}  \, \mid \, \text{ trajectory $i=$}(x_0,x_1,\ldots,x_{\tau-1}),\mathcal{E}_i}\\
&\geq \Pr{ \cup_{j=1, j \neq i}^{\kappa^{100}} \left\{ Z(i,j) \geq 1 \right\} \, \mid \, \text{trajectory $i=$}(x_0,x_1,\ldots,x_{\tau-1}),\mathcal{E}_i} \\
&\phantom{000}\geq 1 - \prod_{j=1, j \neq i}^{\kappa^{100}} \left( 1 - \frac{|D_5(i,j)|}{n \cdot \kappa^{47}}   \right)  \\
&\phantom{000}\geq 1 - \exp\left(-  \frac{1}{n \cdot \kappa^{47}} \cdot \sum_{j=1, j \neq i}^{\kappa^{100}} |D_5(i,j)| \right) \\
&\phantom{00}\stackrel{\eqref{eq:thomas}}{\geq} 1- \exp\left(- 1/4 \right).
\end{align*}
Note that the above derivation was conditional on $\mathcal{E}_i$, but this event holds with probability $1/2$. Hence with probability at least $1/2 \cdot (1- \exp\left(- 1/4 \right)) > 1/10$, the trajectory of $i$ meets with that of a random walk in $\mathcal{G}_1$ and hence the random walk $i$ is eliminated before time step $\tau$. 
\end{proof}

\subsubsection{Completing the Proof of \autoref{thm:keylemma}}

We are now ready to complete the proof of \autoref{thm:keylemma} by combining 
\autoref{lem:structure}, 
\autoref{lem:stocknockout} and \autoref{lem:thelastone}.
\begin{proof}[Proof~of~\autoref{thm:keylemma}]
Let $k=\kappa^{100}$ be the number of random walks. 
\GNOTE{T: Frederik can you add a reference to Assumption 4.5 here?}
Since we seek to reduce the number of random walks to 
$(\Delta/d)^{100}$, we assume in the following that that
	$k^{1/100} = \kappa > \Delta/d$ and $\kappa \geq 2^{10}$ (cf.~\autoref{ass:funnyk}).
Otherwise, if $\kappa < 2^{10}$, then with $V_0$ denoting the set of start vertices of the $k$ walks, we have  
 $\tcoal(V_0) = O(\tmeet\cdot\log(|V_0|))=O(\thit)$, by \autoref{lem:beer} and \autoref{pro:relatingmeetandhit}.
As derived in \eqref{eq:numberofbadballs} we have that out of the $\kappa^{100}$ random walks at least $\kappa^{100}-\kappa^{64}$ random walks are good.
By \autoref{lem:stocknockout}, any good random walk $(X_t)_{t \geq 0}$ with label $i$ satisfies $\Pr{ |D_4(i)| \geq n/(4 \kappa^8)} \geq 1-8/\kappa$. Conditioning on the trajectory $(x_0,x_1,\ldots,x_{\tau-1})$ satisfying $|D_4(i)| \geq n/(4 \kappa^8)$, \autoref{lem:thelastone} shows that with probability at least $1/10$ the random walk $i$ will be eliminated before time step $\tau=O(\thit/\kappa)$. Hence a constant fraction of all $k=\kappa^{100}$ random walks are  eliminated in a single phase of $O(\thit/\kappa)$ steps with constant probability $>0$, provided that $\kappa > \Delta/d$. 

In conclusion, for any $k' > (\Delta/d)^{100}$, there exists a constant $c > 0$ such that the expected time required to reduce the number of walks from $k'$ to $ \max\{ k'/2, (\Delta/d)^{100}\}$ is bounded by $c\cdot \thit/\sqrt[100]{k'}$, by \autoref{lem:drift}.
 Therefore, the expected time to reduce the number of walks from $k'\leq \log^4 n$ to $(\Delta/d)^{100}$ is upper bounded by
 \[
 \sum_{i=0}^{\log( \log^4 n)}    \frac{c\cdot \thit}{\sqrt[100]{2^i}} \leq c\cdot \thit\sum_{i=0}^{\infty }    \left(\frac{1}{\sqrt[100]{2} }\right)^i  =\frac{c \cdot \thit}{1-\sqrt[100]{1/2}}=O(\thit). 
 \]
 
\end{proof}

\subsection{Bounding $\tcoal$ in terms of $\thit$}

\fnote{Check the proof of Theorem 1.3. The reviewer had some smaller things and I rewrote parts of it.}

In this subsection we prove the following theorem relating $\tcoal$ to $\thit$.
Recall that $\tmeet \leq 4 \thit$ (\autoref{pro:relatingmeetandhit})  for any graph.

\begin{theorem}\label{thm:regular}
Let $G=(V,E)$ be any graph with maximum degree $\Delta$ and average degree $d$. Then $\tcoal = O(\thit + \tmeet \cdot \log(\Delta/d))$. So in particular, for any almost-regular graph, $\tcoal=O(\thit)$. 
\end{theorem}

\autoref{thm:regular} follows almost immediately from the previous two reductions in \thmref{mostgeneral} and \thmref{keylemma}. 

\begin{proof}[Proof of \thmref{regular}]
By \thmref{mostgeneral}, we can reduce the number of walks from $n$ to $O(\log^3 n)$ in $O(\thit)$ steps with probability at least $1-n^{-1}$. Then, using \thmref{keylemma}, we can reduce the number of walks from $O(\log^ 3 n)$ to $(\Delta/d)^{100}$ in $O(\thit)$ expected time. Finally, we apply \autoref{lem:beer} to reduce the number of walks from $(\Delta/d)^{100}$ to $1$ in $O(\tmeet \cdot \log (\Delta/d))$ expected time to obtain the result.
\end{proof}

\subsection{Proof of \autoref{thm:hittingtime}}
Part~(i) follows from  \autoref{thm:mostgeneral} together with $\tmeet \leq 4\thit$ (\autoref{pro:relatingmeetandhit}) and $\tcoal(S_0) = O(\tmeet \log|S_0|)$ (\autoref{lem:beer}).

Part~(ii) is the statement of \autoref{thm:regular}.
To prove Part~(iii) follows from the following three facts.
First, $\thit = \Theta(\tmeet)$ by \autoref{pro:relatingmeetandhit}.
Second,  $\tcoal = O(\thit + \tmeet)$ by Part~(ii).
Third, $\tcoal \geq \tmeet$.
Finally, Part~(iv) follows from the results presented in \autoref{sec:special}.

\subsection{Conjecture and a Possible Improvement for Non-Regular Graphs}

Before concluding this section, we mention an intriguing conjecture that might be useful to improve our bound on $\tcoal$ when $\Delta \gg d$.

\begin{conjecture}\label{conj}
There exists a universal constant $C>0$ so that for any graph $G=(V,E)$, any vertex $u \in V$ and any path of vertices $(x_0,x_1,\ldots,x_{\thit})$, i.e., either $x_i=x_{i+1}$ or $\{x_i,x_{i+1}\} \in E(G)$,
\[
  \sum_{t=0}^{\thit} p_{u,x_t}^t \leq C \cdot \sum_{t=0}^{\thit} \pi(x_t).
\]
\end{conjecture}
Note that the inequality is a stronger version than the one given in the first statement of~\autoref{lem:return} or \autoref{lem:deterministicpath}.
We do not know whether the conjecture is actually true in this generality. However, if it is true, it would imply that any random walk of length $4 \thit$ starting from an arbitrary vertex meets with any deterministic path of length $4 \thit$ with constant probability $>0$. This would then result in a simple proof that $\tcoal = O(\thit \cdot \log^{*} n)$ for any graph, since each phase of $O(\thit)$ steps would reduce the number of walks from $k$ to $O(\log k)$.

\autoref{conj} can be also seen as the optimization problem of ``predicting'' a random walk $(Y_t)_{t \geq 0}$ for $\thit$ time steps. More precisely, we are given the start vertex of the random walk $Y_0 = u$ and for each time step $1 \leq t \leq \thit$, we have to specify a vertex $x_t$ that acts as a predictor of the random location of the random walk at step $t$. The goal is to maximize the (expected) number of correct predictions, which is equal to
\[
  \sum_{t=0}^{\thit} p_{u,x_t}^t.
\]
The conjecture states that regardless which prediction, \ie which path $(x_0,x_1,\ldots,x_{\thit})$ is picked, the expected number of correct predictions cannot be made larger than in the setting where the random walks starts from stationarity (and the start vertex is unknown).

One specific strategy would be to choose $x_0=x_1=\cdots=x_{\thit} =v$ for some vertex $v$. In that case we know by \autoref{lem:return} that
\[
  \sum_{t=0}^{\thit} p_{u,v}^t \leq 4 \cdot (\thit+1) \cdot \pi(v), 
\]
so the conjecture holds in this case.

It is also worth mentioning that we cannot replace $\thit$ by a smaller value, say, $\tmix$. Indeed if $G$ is a two-dimensional grid, then $\tmix=\Theta(n)$ and choosing $x_0=x_1=\cdots=x_{\thit} =u$, we obtain $\sum_{t=0}^{\tmix} p_{u,u}^t = \Omega(\log n)$, while, $\sum_{t=0}^{\tmix} \pi(x_t) = O(1)$.

Finally, there is some resemblance to the meeting-time-lemma in the continuous-time setting~\cite{O12}, however, one important difference is that in \autoref{conj}, the right hand side depends on the actual path $(x_0,x_1,\ldots,x_{\thit})$.

\section{Bounding $\tcoal \in [\Omega(\log n), O(n^3)]$ }\label{sec:worst-case-bounds}

Given that worst-case upper and lower bounds have long been known for $\tmix, \thit$ and $\tcov$, it is very natural to pose the same question for $\tmeet$ and $\tcoal$. In the following we determine the correct asymptotic worst-case upper and lower bounds for $\tmeet$ and $\tcoal$ on (i) general graphs, (ii) regular graphs and (iii) vertex-transitive graphs. We refer to \autoref{lessertable} for 
an overview.

\begin{table}
	\rowcolors{2}{testcolor!50}{}

	\resizebox{\linewidth}{!}{ 
	\begin{tabular}{L{4.4cm}L{2.3cm}L{3.9cm}llL{1.5cm}L{1.0cm}} 
		\toprule
		Graph  & \twolines{$\tmeet$}{} & \twolines{}{} &
		\twolines{$\tcoal$}{} & \twolines{}{}  \\
		\midrule
		General Graphs   & $\Omega(1)$, $O(n^3)$ & 
		\ripref{Thm.}{thm:meetingtime}& 
		$\Omega(\log n)$, $O(n^3)$  & \ripref{Lem.}{lem:lognlower} $\&$ \ripref{Thm.}{thm:hittingtime} \\
		Regular Graphs   & $\Omega(n)$, $O(n^2)$ & 
		\ripref{Thm.}{thm:meetingtime} $\&$ \ripref{Thm.}{thm:hittingtime}& 
		$\Omega(n)$, $O(n^2)$  & \ripref{Thm.}{thm:meetingtime} $\&$ \ripref{Thm.}{thm:hittingtime}& \\
	Vertex-Trans. Graphs   & $\Omega(n)$, $O(n^2)$ & 
		\ripref{Thm.}{thm:meetingtime} $\&$ \ripref{Thm.}{thm:hittingtime}& 
		$\Omega(n)$, $O(n^2)$  & \ripref{Thm.}{thm:meetingtime} $\&$ \ripref{Thm.}{thm:hittingtime}& \\
		\bottomrule
	\end{tabular}}%
	\begin{flushleft}
		\begin{small}
			\caption{\label{lessertable}A summary of bounds on the  meeting and coalescence 
			 times graph classes.  All bounds are easily shown to be tight:
			 For general graphs the meeting time and coalescence bounds are matched by the star and the barbell graph.	For vertex-transitive and regular graphs the bounds are matched by the clique and the cycle.
			 }
		\end{small}
	\end{flushleft}
\end{table}

\subsection{General Upper Bound $\tcoal=O(n^3)$}
In this section we establish that $\tcoal = O(n^3)$ on all graphs, which is matched for instance by the Barbell graph.

\begin{theorem}\label{thm:worstcasecoal}
For any graph $G$ we have $\tcoal = O(n \cdot |E| \cdot \log(|E|/n))$, so in particular,
$\tcoal=O(n^3)$.
\end{theorem}
\begin{proof}[Proof of \autoref{thm:worstcasecoal}]
It is  well-known that $\thit \leq n \cdot 2|E|$ (cf.~\cite{AKLLR79}). From \autoref{pro:relatingmeetandhit} and \autoref{thm:regular} we derive 
\begin{align*}
 \tcoal &= O(\thit + \tmeet \cdot  \log(\Delta/d)) \\  
 &= O(\thit + \thit  \log(n^2/|E|)) \\
 &= O(n \cdot |E| + n \cdot |E| \cdot \log (n^2/|E|)) = O(n^3),
\end{align*}
where the last inequality holds since $|E| < n^2$.
\end{proof}

\subsection{General Lower Bound $\tcoal = \Omega(\log n)$}\label{lognlower}
In this section, we prove that  the coalescing time of any graph is $\Omega(\log
n)$. 
We consider a process $P'$ where there is exactly one
random walk starting at each node in the graph. 
For every node $u\in V$ and every time step $t\in \naturals$ we 
draw an independent random variable $Z_{u,t} \in \{0,1\}$ with
$\Pr{Z_{u,t}=1}=1/2$ and  $\Pr{Z_{u,t}=0}=1/2$.
If  $Z_{u,t}=1$, then the random walk on $u$ at time $t$ (if there is any),  moves to a neighboring node chosen u.a.r.. Otherwise ($Z_{u,t}=0$),  the random walk on $u$ at time $t$ (if there is any) stays on the same node. 
It is
straightforward to show that the set of nodes which have an active random walk
according to this process can be coupled with the coalescence process defined
in \autoref{sec:notation}.

We show that after $c \log n$ steps, for a sufficiently small
$c$, there are at least two surviving walks in this process. In order to do
this, we simply argue that there must be at least two walks that have not left
their starting position. Note that there is no way for these walks to be
eliminated, because even if other walks visited one of their starting  nodes, there
are two nodes from which no walks can have left. 
The formal proof follows.

\begin{lemma}\label{lem:lognlower}
	For any graph $G=(V,E)$, $|V|=n$ we have $\tcoal=\Omega(\log n)$. 
\end{lemma}
\begin{proof}
Consider the process $P'$ defined above.
Let $T$ be the coalescence time.
Note that coalescence at time $\tau$ in $P'$ requires that for $n-1$ nodes $u\in V$ there exists $t_u\leq \tau$ such
$Z_{u,t_u}=1$. In symbols, let $T$ be the first point in time where all walks coalesced, then
$T\geq T'$, with $T' := \min \{t'\in \naturals \colon  |\{ u \colon \exists t_u\leq t' \text{ s.t. } Z_{u,t_u}=1  \}| \geq n-1 \}$. 
 Let $Y_u$ be the indicator variable which is $1$ 
 if $Z_{u,t}=0$ for all $t\leq \tau := \log n/2$.
The process ensures independence of the $Y_u$. 
Due to  the laziness of the random walk, $\ \Pr{Y_u=1 }  = 1/2^\tau=1/\sqrt{n}.$
Thus, using the independence of the $Y_u$,
\[
\Pr{T \geq \tau} \geq \Pr{ T' \geq \tau } \geq \Pr{ \sum_{u \in V}  Y_{u} \geq 2 }  =\Pr{\operatorname{Binomial}(n,1/\sqrt{n}) \geq 2} = 1-o(1) ,
\]
where $\operatorname{Binomial}(n,p)$ denotes the binomial distribution with parameters $n$ and $p$.
We conclude that $\E{T}=\Omega(\log n)$ which yields the claim. 
\end{proof}

\subsection{Proof of \autoref{thm:graphclasses}}
\label{sec:trivialproof}

We are now ready to put all the pieces together.
The upper bound on general graphs follows directly from $\tmeet\leq \tcoal = O(n^3)$, by \autoref{thm:worstcasecoal}.
The lower bound on the meeting time holds by definition and the lower bound on the coalescing time follows from \autoref{lem:lognlower}.
For the upper bound on regular graphs  we have $\tmeet \leq \tcoal =O( \thit) = O(n^2)$ due to
\autoref{thm:hittingtime}, having used the standard bound $\thit=O(n^2)$ for regular graphs~(see \cite{AF14}).
The lower bound follows from $\tcoal \geq \tmeet \geq \tmeet^{\pi} = \Omega(n)$, by  \autoref{thm:meetingtime}.

\fi

\subsubsection*{Acknowledgments.} 
The authors would like to thank Petra Berenbrink, Robert Els{\"a}sser, Nikolaos Fountoulakis, Peter Kling, Roberto Oliveira and Perla Sousi for helpful discussions and in particular Yuval Peres for pointing out how to further improve the bound on the hitting time (\autoref{thm:spectral}).
Moreover, the authors would like to thank the anonymous reviewer for pointing out a mistake in the domination described in Section~\ref{sec:process}\ifEA~of the full paper\fi~in an earlier version of this work.

\printbibliography

\ifEA
\else

\appendix

\section{Basic Results about Markov Chains}
\label{app:lazyrw}

We will frequently use the following basic fact about lazy random walks, which in fact also holds for arbitrary reversible Markov chains:

\begin{lemma}[cf.~{\cite[Chapter~12]{LPW06}}]\label{lem:loop}
Let $P$ be the transition matrix of a reversible Markov chain with state space $\Omega$. Then the following statements hold:
\begin{enumerate}[(i)]	
\item If $P$ is irreducible, then for any two states $x,y \in \Omega$, 
\begin{align*}
 p_{x,y}^t \leq \pi(y) +  \sqrt{ \frac{\pi(y)}{\pi(x)}} \cdot \lambda^t,
\end{align*}
where $\lambda:=\max\{\lambda_2,|\lambda_n|\}$ and $\lambda_1 \geq \lambda_2 \geq \cdots \geq \lambda_n$ are the $n$ real eigenvalues of the matrix $P$.
\item If the Markov chain is a non-lazy random walk on a bipartite regular graph with two partitions $V_1$ and $V_2$, then for any pair of states $x,y$ in the same partition
\begin{align*}
p_{x,y}^t \leq \frac{{ 2}}{n} \cdot \left(1 + (-1)^{t-1} \right) + {2}\left(\max\{ \lambda_2, |\lambda_{n-1}| \} \right)^{ t}.
\end{align*}
Similarly, if $x$ and $y$ are in opposite partitions,%
		\ifdraft
		{\color{red} Perhaps we could make the statements slightly more general
		by bounding $|p_{x,y}^t - 2/n \cdot (1+(-1)^t)| \leq ...$}
		\fi
\begin{align*}
p_{x,y}^t \leq \frac{{ 2}}{n} \cdot \left(1 + (-1)^{t} \right)  + {2}\left(\max\{ \lambda_2, |\lambda_{n-1}| \} \right)^{ t}.
\end{align*}

\item If the Markov chain is lazy, then for any state $x \in \Omega$, $p_{x,x}^t$ is non-increasing in $t$.
In particular, $p^t_{x,x} \geq \pi(u)$.
\end{enumerate}
\end{lemma}
\begin{proof}
The first statement can be found in \cite[Equation~12.11]{LPW06}. 

For the second statement, recall the spectral representation \cite[Lemma~12.2~(iii)]{LPW06}
\begin{align}
 p_{x,y}^t &= \pi(y) + \pi(y) \cdot \sum_{k=2}^{n} u_k(x) \cdot u_k(y) \cdot \lambda_k^t, \label{eq:spectral}
\end{align}
where $u_k$ is the corresponding eigenvector to $\lambda_k$. Since all eigenvalues are non-negative, we conclude from \eqref{eq:spectral} that $p_{x,x}^t$ is non-increasing in $t$ as needed.
Since $G$ is bipartite and regular, it is not difficult to verify that $\lambda_{n}=-1$ and $u_{n}(x) = \sqrt{1/n}$ if $x \in V_1$ and $u_{n}(x) = -\sqrt{1/n}$ if $x \in V_2$ is the corresponding eigenvector. Hence, 
\begin{align*}
 \left| p_{x,y}^t - \frac{2}{n} \cdot \left(1 + (-1)^{t-1} \right) \right| &\leq \pi(y) \cdot \left| \sum_{k=2}^{n-1} u_k(x) \cdot u_k(y) \cdot \lambda_k^{t} \right| \\
 &\leq \frac{2}{n} \cdot \max_{2 \leq k \leq n-1} \left| \lambda_k^{t} \right|
 \cdot \frac{1}{n} \cdot \sum_{k=2}^{n-1} \left| u_k(x) \cdot u_k(y)  \right|  \\
 &\leq \frac{2}{n} \cdot  \max_{2 \leq k \leq n-1} \left| \lambda_k^{t} \right| \cdot 
 \sqrt{ \sum_{k=2}^{n-1} u_k(x)^2 \cdot \sum_{k=2}^{n-1} u_k(y)^2 }
 \end{align*}
 As in \cite[Proof of Theorem 12.3]{LPW06}, using the orthonormality of the eigenvectors, we have 
 \[
   \sum_{k=2}^{n-1} u_k(x)^2 \leq \sum_{k=2}^{n} u_k(x)^2 \leq n,
 \]
 and the second statement follows if $u$ and $v$ are in the same partition. The case where $u$ and $v$ are in different partitions follows analogously.

For the third statement, first note that by \cite[Exercise~12.3]{LPW06}, all eigenvalues of the transition matrix $P$ are non-negative.
Since all eigenvalues are non-negative, we conclude from \eqref{eq:spectral} that $p_{x,x}^t$ is non-increasing in $t$ as needed.
Due to this and the fact that $p_{x,x}^t$ converges to $\pi_x$, we get that $p^t_{x,x} \geq \pi_x$. 
\end{proof}

The following is a simple corollary from a recent work by \citet{MSS15} on the existence of Ramanujan graphs.
	\begin{lemma}[cf.~\citet{MSS15}]\label{lem:BASF} 	For any integer $d \geq 3$, there are $d$-regular bipartite Ramanujan graph $H=(V,E)$ with $\tmix=O(\log n/ \log d )$. 
		\end{lemma}
	\begin{proof}
	\citet{MSS15}  show that the existence of a $d$-regular bipartite Ramanujan graph $H$ such that $\max \{ \lambda_2(\widehat Q), | \lambda_{n-1}(\widehat Q)|\} = O(1/\sqrt{d})$, where $\widehat Q = \frac{1}{d} {A}$  is the transition matrix of a non-lazy random walk where $A$ is the adjacency matrix. 
	 By the second statement of \autoref{lem:loop}, for any pair of states $x,y$ in the same partition
\begin{align*}
\hat q_{x,y}^t \leq \frac{{ 2}}{n} \cdot \left(1 + (-1)^{t-1} \right) + {2}\left(\max\{ \lambda_2, |\lambda_{n-1}| \} \right)^{ t}.
\end{align*}
Similarly, $x$ and $y$ are in opposite partitions,
\begin{align*}
\hat q_{x,y}^t \leq \frac{{ 2}}{n} \cdot \left(1 + (-1)^{t} \right) + {2}\left(\max\{ \lambda_2, |\lambda_{n-1}| \} \right)^{ t}.
\end{align*}
Furthermore note that $q_{x,y}^t \geq 2/n$ due to \autoref{lem:loop}.(iii) for even (or odd) $t$
depending on whether $x$ and $y$ are in the same partitions.

	Fix $t=O(\log n/\log d)$ such that 
	$2\left(\max \{ \lambda_2(\widehat Q), | \lambda_{n-1}(\widehat Q)|\}\right)^t \leq \frac{1}{20n}$, where we note that such a t exists due to $\max \{ \lambda_2(\widehat Q), | \lambda_{n-1}(\widehat Q)|\} = O(1/\sqrt{d})$.
	We choose $s$ to be the smallest odd integer being greater than $20t$.
	To translate from the non-lazy random walk $\widehat Q$ to a lazy-random walk $P$, 
let $Z$ denote the number of non-loops performed by a lazy random walk of length $s$.
Since, the probability for a self-loop is $1/2$ and the number of self-loops is binomially distributed, we have 
\begin{align*}
 \Pr{ Z \geq t } \geq 19/20.
\end{align*}
%
By symmetry and the fact that $s$ is odd, 
$
 \Pr{ \mbox{$Z$ is even} } = \frac{1}{2}.
$
Hence, by the Union bound,
\begin{align*}
 \Pr{ \mbox{$Z$ is even} \, \mid \,  Z \geq t } \geq \Pr{ \mbox{$Z$ is even} \cap  Z \geq t } \geq \Pr{ \mbox{$Z$ is even}} - \Pr{Z < t} \geq \frac{9}{20},
\end{align*}
and similarly, $
 \Pr{ \mbox{$Z$ is odd} \, \mid \,  Z \geq t } \geq \frac{9}{20}$.
Let $V_1$ and $V_2$ be the bipartite partition of $V$. 
	\begin{align*}
	  \tvdist{ p_{u,\cdot}^{s} - \pi} &\leq
	  \Pr{ Z < t} \cdot 1 + \Pr{ Z \geq t} 
	  \cdot \left(
	  \sum_{v \in V_1} \left| \frac{11}{20} \hat q^t_{u,v}  - \frac{1}{n} \right| +	 
	  \sum_{v \in V_2} \left| \frac{11}{20} \hat q^t_{u,v} - \frac{1}{n} \right| 
	  \right) \\
	  &\leq \Pr{ Z < t} \cdot 1 + \Pr{ Z \geq t} 
	  \cdot \left(
	  \sum_{v \in V} \left| \frac{11}{20} \left( \frac{2}{n}+\frac{1}{20n}\right)  - \frac{1}{n} \right| 	  \right) \\
	   &\leq \Pr{ Z < t} \cdot 1 + \Pr{ Z \geq t} 
	  \cdot \left(
	  \sum_{v \in V} \left| \frac{22}{20n}  - \frac{1}{n} \right| + \frac{11}{400} \right)
	   \\
	  &\leq \frac{1}{20}\cdot 1+ \frac{19}{20}\left( \frac{2}{20}  + \frac{11}{400}\right)   < 1/e,
	\end{align*}
	where the first inequality follows from the equations for $p_{x,y}^t$ above.
	\end{proof}

\begin{corollary}\label{cor:universialconstant}
Let $n_0$ be a sufficiently large constant.
Let $H_n$ be the graph of   \autoref{lem:BASF} with $n$ nodes and $d=\ceil{\sqrt{n}}$ for $n\geq n_0$.
	There exists a universal constant $C$ such that $ \max_{n\geq n_0} \{ \tsep(H_n)\} \leq C$.
\end{corollary}
The corollary follows directly from \autoref{lem:BASF} and $\tsep \leq 4 \tmix$. 

The following lemma will be helpful to define a coupling between  distributions that are close to the stationary distribution and the exact stationary distribution. (A very similar lemma has been derived in \cite[Lemma~2.8]{ES11})
\begin{lemma}\label{lem:coupling}
Let $\varepsilon \in (0,1]$ be an arbitrary value. Let $Z_1$ and $Z_2$ be two probability distributions over $\{1,\ldots,n\}$ so that $\Pr{ Z_1 = i} \geq \epsilon \cdot \Pr{ Z_2 = i}$ for every $1 \leq i \leq n$. Then, there is a coupling $(\tilde{Z}_1,\tilde{Z}_2)$ of $(Z_1,Z_2)$  and an event $\mathcal{E}$ with $\Pr{ \mathcal{E} } \geq \epsilon$ so that
\begin{align*}
   \Pr{ \tilde{Z}_1 = i  \, \mid \, \mathcal{E} } = \Pr{ \tilde{Z}_2 = i} \qquad \mbox{ for every $1 \leq i \leq n$.}
\end{align*}
\end{lemma}
\begin{proof}
Let $U \in [0,1]$ be a uniform random variable. We next define our coupling $(\tilde{Z}_1,\tilde{Z}_2)$ of $Z_1$ and $Z_2$ that will depend on the outcome of $U$.
First, if $U \in [0, \epsilon)$, then we set
\[
 \tilde{Z}_1 = \tilde{Z}_2 = i, \qquad \mbox{if $i$ satisfies $ \epsilon \sum_{k=1}^{i-1} \Pr{Z_2 = k} \leq U < \epsilon \sum_{k=1}^{i} \Pr{Z_2 = k}$.}
\]
For the case where $U \in (\epsilon,1)$, it is clear that the definition of $U$ can be extended 
in a way so that $\tilde{Z}_1$ has the same distribution as $Z_1$, and $\tilde{Z}_2$ has the same distribution as $Z_2$. Furthermore, notice that if $U \in [0,\epsilon)$ happens, then $\tilde{Z}_1$ has the same distribution as $Z_2$, and $\tilde{Z}_1= \tilde{Z}_2$. Observing that $\Pr{ U \in [0,\epsilon)} = \epsilon$ completes the proof.
\end{proof}

The following lemma is an immediate consequence of \autoref{lem:coupling}.

\begin{lemma}\label{lem:randomwalkcoupling}
Consider a random walk $(X_t)_{t \geq 0}$, starting from an arbitrary but fixed vertex $x_0$. Then with probability at least $1-1/e$, we can couple $X_{4\tmix}$ with the stationary distribution.
\end{lemma}
\begin{proof}
Consider the random walk $(X_t)_{t \geq 0}$ after step $s:=\tsep \leq 4 \tmix$. 
By definition of $\tsep$, $p_{u,v}^s \geq (1-1/e) \pi(v)$. Applying \autoref{lem:coupling}, where $Z_1$ is the distribution given by $p_{u,v}^t$ and $Z_2$ is the stationary distribution shows that with probability at least $1-1/e$, $X_{s}$ has the same distribution as $\pi$. If this is the case, then the same holds for $X_{4 \tmix}$ as well. 
\end{proof}

The lemma above shows that for $\tmeet$ and $\tcoal$ it suffices to consider the stationary case:
\begin{lemma}\label{lem:avgmeet} 
For any graph $G$,
\[
  \max\{ (1/e) \tmix, \tavgmeet \} \leq \tmeet \leq \frac{2}{(1-1/e)^2} \cdot \left( 4 \tmix + 2 \tmeet^{\pi} \right),
\]
and similarly, $
  \tcoal \leq 4 \cdot (4 \tmix + 2 \tcoal^{\pi} ).$
\end{lemma}
\begin{proof}
We begin by proving the lower bound on $\tmeet$. First, consider two independent random walks $(X_t)_{t\geq 0}$ and $(Y_t)_{t \geq 0}$ that are run for $t = e \cdot \tmeet$ time-steps.
Then, we have 
\[
  \bar d(t)=\max_{u,v} \tvdist{p_{u,\cdot}^t - p_{v,\cdot}^t} \leq \Pr{ \cup_{s=0}^t X_s = Y_s } \leq \frac{1}{e},
\]  
where the first inequality is due to the coupling method \cite[Theorem 5.3]{LPW06} and the second inequality follows by Markov's inequality. The above inequality implies $\tmix \leq e \cdot \tmeet$. Furthermore, $\tavgmeet \leq \tmeet$ holds by definition, and the lower bound follows.

For the upper bound, we divide the two random walks into consecutive epochs of length $\ell:=4 \tmix + 2 \tmeet^{\pi}$. For the statement it suffices to prove that in each such epoch, regardless of the start vertices of the two random walks, a meeting occurs with probability at least $(1-1/e)^2 \cdot 1/2$.

Consider the first random walk $(X_t)_{t \geq 0}$ starting from an arbitrary vertex after $s:=4\tmix$ steps. 
By \autoref{lem:randomwalkcoupling}, we obtain that with probability at least $1-1/e$, the distribution of $X_s$ is equal to that of a stationary random walk. Similarly, we obtain that with probability at least $1-1/e$, the distribution of $Y_s$ is equal to that of a stationary distribution. Hence with probability $(1-1/e)^2$, $X_s$ and $Y_s$ are drawn independently from the stationary distribution. In this case, it follows by Markov's inequality that the two random walks meet before step $s+2 \tmeet^{\pi}$ with probability at least $1/2$. Overall, we have shown that with probability at least $(1-1/e)^2 \cdot 1/2$, a meeting occurs in a single epoch. Since this lower bound holds for every epoch, independent of the outcomes in previous epochs, the upper bound on the expected time $\tmeet$ follows.
The upper bound on $\tcoal$ in terms of $\tcoal^{\pi}$ is shown in exactly the same way.
\end{proof}

\begin{lemma}\label{lem:fridaygift}
For a lazy random walk on an $n$-vertex graph with $n \geq 2$, we have $\thit \geq  \frac{2}{\pi_{\min}} - 2.$ 
In particular for $n\geq 2$, we have  $\thit \geq  1/\pi_{\min}\geq n$.
\end{lemma}
Note that $\thit \geq  \frac{2}{\pi_{\min}} - 2$ is tight in the sense that the hitting time
of the clique is indeed $2(n-1)= \frac{2}{\pi_{\min}} - 2$ since the random walk moves w.p. $1/2$ and when it moves the probability to hit the target node is $1/(n-1)$ (assuming that the random walk is not on the target node).
\begin{proof}
Let $u$ be a vertex attaining $\pi_{\min} = \pi(u)$. 
Consider the random walks $(X_t)_{t\geq 0}$ starting at $u$.
Then it is well-known (cf. \cite{AF14}) that for the first return $\tau^+(u,u) := \min \{ t > 0: X_t = u, X_0 = u\}$, we have $\E{ \tau^+(u,u) }= 1/\pi(u) = 1/\pi_{\min}$. By conditioning on the first step of the random walk, we obtain 
\begin{align*}
  \frac{1}{\pi_{\min}} = \E{ \tau^+(u,u) } &= 1 + \frac{1}{2} \cdot 0 + \frac{1}{2} \sum_{v \in N(u)} \frac{1}{\deg(u)} \cdot \thit(v,u),
\end{align*}
and rearranging yields
\begin{align*}
 \frac{1}{\deg(u)} \cdot \sum_{v \in N(u)} \thit(v,u) &= \frac{2}{\pi_{\min}} - 2.
\end{align*}
 Now by the pigeonhole principle there exists a vertex $v \in N(u)$ with $\thit(v,u) \geq \frac{2}{\pi_{\min}} - 2$, and the first claim follows. The second part follows from observing that if $n \geq 2$ we have $\pi_{\min} \leq 1/2$ and thus
 $\thit \geq  \frac{2}{\pi_{\min}} - 2\geq \frac{1}{\pi_{\min}}\geq n$,
 where the last inequality follows from the simple pigeon hole principle. 
\end{proof}

\begin{observation}\label{lem:fresenius}
	Consider two random walks $(X_t)_{t \geq 0}$ and $(Y_t)_{t \geq 0}$ starting on nodes drawn from the stationary distribution.
	Fix an arbitrary $t\in \mathbb{N}$.
	Define the collision-counting random variables $Z_1=\sum_{i=0}^{\ceil{t/2}} \mathbf{1}_{X_t=Y_t}$,
	$Z_2=\sum_{i=\ceil{t/2}+1}^{t} \mathbf{1}_{X_t=Y_t}$,
	and $Z=Z_1 + Z_2$.
	Then $\Pr{Z_1 \geq 1 \, \mid \, Z \geq 1} \geq \frac{1}{2}$.
\end{observation}
\begin{proof}
	  Since both nodes start from the stationary distribution, $\Pr{Z_1 \geq 1} \geq  \Pr{Z_2 \geq 1}$. By the Union bound,
 $\Pr{Z \geq 1} \leq \Pr{Z_1 \geq 1} + \Pr{Z_2 \geq 1} \leq 2 \cdot \Pr{Z_1 \geq 1}.$
By law of total probability, $\Pr{Z_1}=\Pr{Z_1 \geq 1 ~|~ Z\geq 1}\cdot \Pr{Z\geq 1}$.
Putting everything together yields $\Pr{Z_1 \geq 1 \, \mid \, Z \geq 1} \geq \frac{1}{2}$.
\end{proof}
\begin{lemma}\label{lem:drift} 
Let $(X_t)_{t\geq 0}$ be a stochastic process satisfying (i)
$\E{X_{t} | \mathcal{F}_{t-1} } \leq  \beta \cdot X_{t-1}$, for some $\beta < 1$, and (ii) $X_t \geq 0$ for all $t \geq 0$.
Let $\tau(g) := \min\{ t\geq 0|  X_t \leq g \}$  	for $g \in (0,|X_0|)$, then
\[
\E{\tau(g) } \leq   2 \cdot \ceil{\log_{\beta}(g/(2X_0))}. 
\]
\end{lemma}
\begin{proof}
By the iterative law of expectation, we have
\begin{align*}
 \E{X_t} &\leq \beta^{t} \cdot X_0.
\end{align*}
Furthermore, by Markov's inequality, for any $\lambda \geq 1$
\begin{align*}
 \Pr{ \tau(g) > \lambda \cdot \ceil{\log_{\beta}(g/(2X_0)}} \leq \Pr{ X_{ \lambda \cdot \ceil{\log_{\beta}(g/(2X_0)}   } > g} \leq 
 \frac{ \E{	X_{ \lambda \cdot \ceil{\log_{\beta}(g/(2X_0)}}}}{g} \leq 2^{-\lambda}.
\end{align*}
Therefore,
\begin{align*}
 \E{ \tau(g) } &= \sum_{i=1}^{\infty} \Pr{ \tau(g) \geq i} \\
 &\leq \ceil{\log_{\beta}(g/(2X_0))} + \sum_{\lambda=1}^{\infty} \ceil{\log_{\beta}(g/(2X_0))} \cdot \Pr{ \tau(g) > \lambda \cdot \ceil{\log_{\beta}(g/(2X_0)}} \\
 &\leq 2 \cdot \ceil{\log_{\beta}(g/(2X_0))}.
\end{align*}

\end{proof}

\section{Bounding  $\tmeet$ and Implications for $\tcoal, \thit$ and $\tcov$}
\label{sec:meeting}

Although the focus of this work is on understanding the coalescence time, in
order to apply our general results, we need to devise some tools to obtain
lower and upper bounds on $\tmeet$. In 
\autoref{thm:meetingtime} (\autoref{sec:meetpi}) we establish upper and lower
bounds on the meeting time in terms of $\twonorm{\pi}^2=\sum_{u \in V} \pi(u)^2$. 
\autoref{sec:meetpi} contains several additional upper bounds on $\tmeet$ and $\thit$. Through combination with other results, we also obtain new bounds on $\tcoal$ and $\tcov$. A common feature of many of these bounds is a sub-linear dependence on the spectral gap $1/(1-\lambda_2)$, which we obtain by an application of short-term bounds on the $t$-step transition probabilities.

In
\autoref{pro:relatingmeetandhit} (\autoref{sec:meetvt}) we establish a
discrete-time counterpart of \cite[Proposition 14.5]{AF14}, albeit with worse
constants, stating that the meeting time is at most of the order of the hitting
time; on vertex transitive graphs these quantities are asymptotically of the
same order. 

\subsection{Relating Meeting Time to $\tmix$ and $\frac{1}{1-\lambda_2}$}\label{sec:meetpi} 
We first state some basic bounds on $\thit$ and $\tmeet$, which mostly follow directly from \eqref{eq:central} and its counterpart for the hitting times (cf.~Cooper, Frieze~\cite{CF05}). In these bounds, we will use the following notation: 
\begin{align*}
C_{\max} &:= \max_{u \in V}\sum_{t=0}^{\tmix-1} \sum_{v \in V} \left( p_{u,v}^{t} \right)^2, \\
C_{\min} &:= \min_{u \in V}\sum_{t=0}^{\tmix-1} \sum_{v \in V} \left( p_{u,v}^{t} \right)^2.
\end{align*}
Note that $C_{\max}$ and $C_{\min}$ provide worst-case upper respective lower bounds on the expected collisions of two independent random walks of length $\tmix$, starting from the same vertex $u$.
Similarly, we define
\begin{align*}
 R_{\max} := \max_{u \in V} \sum_{t=0}^{\tmix-1} p_{u,u}^t.
\end{align*}
Note that $R_{\max}$ is the number of expected returns of a random walk to $u$ during $\tmix$ steps. This quantity is more convenient to bound than $C_{\max}$, for instance, it can be easily bounded by $\max_{u \in V} \pi(u) \cdot \tmix + \frac{1}{1-\lambda_2}$ (cf.~\autoref{lem:loop}, or also~\cite{CEOR13}).

\begin{theorem} \label{thm:meetingtime}
For any graph $G=(V,E)$, the following statements hold:
\begin{enumerate}[(i)]
\item For any pair of vertices $u,v \in V$, 
\[
  \thit(u,v) \leq \frac{5e \cdot (\sum_{t=0}^{\tmix-1} p_{v,v}^t)}{\pi(v)}.
\]
In particular, if the graph $G$ is $\gap$-approximative regular, then 
$
  \thit(u,v) \leq 5e \cdot \gap \cdot n \cdot \sum_{t=0}^{\tmix-1} p_{v,v}^t.
$
	\item For any pair of vertices $u,v \in V$,
\[
 \tmeet(u,v) \leq \frac{5e^2 \cdot C_{\max}}{\| \pi \|_2^2}.
\] 
In particular, if the graph $G$ is $\gap$-approximative regular, then
\[
  \tmeet(u,v) \leq \frac{10 e^2 \cdot (4 + \log_2 (\gap)) \cdot R_{\max}}{\| \pi \|_2^2}.
\]
	\item It holds that,
\[
\tavgmeet \geq \frac{ C_{\min} }{64 \| \pi \|_2^2 }.
 \]
 In particular, if the graph $G$ is $\gap$-approximative regular, then $\tavgmeet=\Omega(n/\gap)$.
\end{enumerate}
\end{theorem} 
Since $C_{\min} \geq 1$, the last statement of the lemma implies also $ \tavgmeet =\Omega( \frac{1}{ \| \pi \|_2^2} ).$
We remark that the second upper bound on $\tmeet(u,v)$ depends only logarithmically on $\gap$. 

\begin{proof} 
We begin by proving the first part. Consider one random walk $(X_t)_{t \geq 0}$, starting from an arbitrary vertex.
Divide the time-interval into consecutive epochs of length $\tsep + \tmix$, and let \[ Z:=\sum_{t=\tsep}^{\tsep + \tmix-1} \mathbf{1}_{X_t=v}\] denote the number of visits. Then, by the separation time, $\E{Z} \geq  \tmix \cdot \frac{1}{e} \cdot \pi(v)$, and (\ref{eq:central}) yields
\begin{align*}
   \Pr{Z \geq 1 } &\geq \frac{\tmix \cdot \frac{\pi(v)}{e}}{\E{Z \, \mid \, Z \geq 1 }}.
\end{align*} 
Clearly, $\E{Z \, \mid \, Z \geq 1} \leq \max_{t=0}^{\tmix-1} \sum_{s=t}^{\tmix} p_{v,v}^{s-t} \leq \sum_{t=0}^{\tmix-1} p_{v,v}^t $.
Hence,
\begin{align*}
  \Pr{Z \geq 1} &\geq \frac{\tmix \cdot \frac{\pi(v)}{e}}{ \sum_{t=0}^{\tmix-1} p_{v,v}^t } =: p.
\end{align*}
This means that in every epoch of length $\tsep+\tmix \leq 5 \tmix$, the random walk has a probability of at least $p$ to visit vertex $v$, and this is independent of any previous epoch. Therefore, the expected number of steps until $v$ is visited is upper bounded by 
\begin{align*}
   \thit(u,v) \leq 5 \tmix \cdot \frac{1}{p} \leq \frac{5e \cdot \sum_{t=0}^{\tmix-1} p_{v,v}^t}{\pi(v)}.
\end{align*}
The claim for $\gap$-approximative regular graphs follows from the observation that $\min_{u \in V} \pi(u) \geq 1/(\gap n)$.
We continue with the second part. Consider two independent random walks, $(X_t)_{t \geq 0}$, $(Y_t)_{t
	\geq 0}$ of length $\tsep + \tmix$ with arbitrary start vertices.  Let $Z$
	be the random variable counting the number of collisions between steps $\tsep$ and $\tsep + \tmix-1$, \ie
	\begin{align*}
	  Z :&= \sum_{t=\tsep}^{\tsep + \tmix-1} \mathbf{1}_{X_t = Y_t}.
	\end{align*}
	By linearity of expectation,
	\begin{align}
	  	\E{Z} &= \sum_{t=\tsep}^{\tsep + \tmix-1} \sum_{u \in V} \Pr{X_t= u} \cdot \Pr{Y_t=u} \geq
		\tmix \cdot \frac{1}{e^2} \cdot \| \pi \|_2^2. \label{eq:first}
	\end{align}
	Let us now consider $\E{ Z \, \mid \, Z \geq 1}$ and recall that conditioning on $Z \geq 1$ can be regarded as jumping to the first step $
	\tau := \min \{ t \colon \tsep \leq t \leq \tmix-1, X_{\tau} = Y_{\tau}\}$ without knowing anything about the future steps $t > \tau$ of both walks.
	 Therefore, 
	\begin{align*}
	 \E{Z \, \mid \, Z \geq 1} \leq \max_{u \in V}\sum_{t=0}^{\tmix-1} \sum_{v \in V} \left( p_{u,v}^{t} \right)^2 = C_{\max}.
	\end{align*}
	Plugging this into \eqref{eq:central} and using \eq{first} we finally	 arrive at
	\begin{align*}
		\Pr{ Z \geq 1} &= \frac{\E{Z}}{\E{Z \, \mid \, Z \geq 1}} \geq \frac{\tmix \cdot
		\frac{1}{e^2} \cdot \| \pi \|_2^2}{C_{\max}} =: p.
	\end{align*}
	Hence,
	\begin{align*}
	  \tmeet(u,v) \leq 5 \tmix \cdot \frac{1}{p} \leq \frac{5e^2 \cdot C_{\max}}{\| \pi \|_2^2}.
	\end{align*}
	Let us derive the result for $\gap$-approximate regular graphs. To this end, define 
	\[
	S_i := \left\{ u \in V \colon \deg(u) \in [2^{i-1},2^{i}) \right\},
	\] and note that $S_0,\ldots,S_{\log_2 n}$ forms a partition of $V$. Since the graph is $\gap$-approximate regular, at most $4 + \log_2(\gap)$ of the $S_i$'s are non-empty.
	 Hence there exists a set $S_j$ with
	\[
	   \sum_{j \in S_i} \pi(j)^2 \geq \frac{1}{4 + \log_2(\gap)} \cdot \| \pi \|_2^2.
	\]
	Let us now by $Z_j$ denote the collisions on the set $S_j$, \ie
		\begin{align*}
	  Z_j :&= \sum_{t=\tsep}^{\tsep + \tmix-1} \mathbf{1}_{\{X_t = Y_t\} \cap \{X_t \in S_j\}}.
	\end{align*}
	Then,
	\[
	  \E{Z_j} = 
	  \sum_{t=\tsep}^{\tsep + \tmix} \sum_{u \in S_j} \Pr{X_t= u} \cdot \Pr{Y_t=u} 
	  \geq \tmix \cdot \sum_{j \in S_i} \frac{1}{e^2} \cdot \pi(j)^2 \geq \frac{\tmix}{e^2} \cdot \frac{1}{  4 + \log_2(\gap)} \cdot \|\pi\|_2^2.
	\]  
	Furthermore, 
	\begin{align*}
	   \E{Z_j \, \mid \, Z_j \geq 1} &\leq \max_{u \in S_i}\sum_{t=0}^{\tmix} \sum_{v \in S_i} \left( p_{u,v}^{t} \right)^2 \leq \max_{u \in S_i}\sum_{t=0}^{\tmix} \sum_{v \in S_i}  p_{u,v}^{t} \cdot 2 p_{v,u}^t,
	\end{align*}
	having used reversibility, \ie $p_{u,v}^t \pi(u) = p_{v,u}^t \pi(v)$ and $\pi(v) / \pi(u) \leq 2$ by definition of $S_i$. Further,
	\begin{align*}
	 \E{Z_j \, \mid \, Z_j \geq 1} \leq 2 \cdot \max_{u \in S_i} \sum_{t=0}^{\tmix -1} \sum_{v \in V}  p_{u,v}^{t} \cdot p_{v,u}^t \leq 2 \cdot \max_{u \in V} \sum_{t=0}^{\tmix -1} p_{u,u}^{2t} \leq 2 \cdot \max_{u \in V} \sum_{t=0}^{\tmix -1 }  p_{u,u}^t = 2 \cdot R_{\max}, 
	\end{align*}
	where the last inequality holds since $p_{u,u}^t$ is non-increasing by \autoref{lem:loop}. 
	Hence, similarly as before,
	\begin{align*}
	  \tmeet(u,v) \leq 5 \tmix \cdot \frac{2 \cdot R_{\max}}{\frac{\tmix}{e^2} \cdot \frac{1}{  4 + \log_2(\gap)} \cdot \|\pi\|_2^2} \leq \frac{10 e^2 \cdot (4 + \log_2 (\gap)) \cdot R_{\max}}{\| \pi \|_2^2}. 
	\end{align*}
	Finally, for the third statement, let $(X_t)_{t \geq 0}$, $(Y_t)_{t\geq 0}$ be two random walk starting from stationarity. 
 Let $\tilde{Z}$
	be the random variable counting the number of collisions between steps $0$ and $2 \tmix$, \ie
	\begin{align*}
	  \tilde{Z} &:= \sum_{t=0}^{2\tmix-1} \mathbf{1}_{X_t = Y_t}.
	\end{align*}
	Then,
	\begin{align*}
	 \E{\tilde{Z}}&= 2 \tmix \cdot \| \pi \|_2^2.
	\end{align*}
	In order to lower bound $\E{\tilde{Z} \, \mid \, \tilde{Z} \geq 1}$, let us write $\tilde{Z}=\tilde{Z}_1+\tilde{Z}_2$ with $\tilde{Z}_1 := \sum_{t=0}^{\tmix} \mathbf{1}_{X_t = Y_t}$ and $\tilde{Z}_2 := \sum_{t=\tmix+1}^{2\tmix-1} \mathbf{1}_{X_t = Y_t}$. 
	By \autoref{lem:fresenius}, $\Pr{\tilde{Z}_1 \geq 1 \, \mid \, \tilde{Z} \geq 1} \geq \frac{1}{2}$. Therefore, by law of total expectation,
	\begin{align*}
	\E{\tilde{Z} \, \mid \, \tilde{Z} \geq 1} &\geq \frac{1}{2} \cdot \E{\tilde{Z} \, \mid \, \tilde{Z}_1 \geq 1} \geq \frac{1}{2} \cdot \min_{u \in V}\sum_{t=0}^{\tmix} \sum_{v \in V} \left( p_{u,v}^{t} \right)^2 = \frac{1}{2} \cdot C_{\min}.
	\end{align*}
	Hence,
	\begin{align*}
	  \Pr{\tilde{Z} \geq 1} &\leq \frac{2 \tmix \cdot \| \pi \|_2^2}{\frac{1}{2} \cdot C_{\min}} =: p.
	\end{align*}
	Consider now $\lceil \frac{C_{\min}}{16 \|\pi\|_2^2 \cdot \tmix} \rceil$ consecutive time-intervals of length $2 \tmix$ each. Note that if $2 \tmix \geq \frac{1}{2} \cdot \frac{C_{\min}}{16 \|\pi\|_2^2}$, then we have
	\[
	  \tmeet^{\pi} \geq \tmix \geq \frac{C_{\min}}{64 \|\pi\|_2^2},
	\]
	and the claim follows immediately.
	Hence we may assume for the remainder of the proof that $2 \tmix < \frac{1}{2} \cdot \frac{C_{\min}}{16 \|\pi\|_2^2}$ and we conclude that, if $B$ denotes the total number of collisions between the walks across all the intervals,
	\[
	\E{B} \leq  \left\lceil \frac{C_{\min}}{16 \|\pi\|_2^2 \cdot \tmix} \right\rceil \cdot p \leq \frac{C_{\min}}{8 \|\pi\|_2^2 \cdot \tmix}
 \cdot \frac{2 \tmix \cdot \| \pi \|_2^2}{\frac{1}{2} \cdot C_{\min}} 
	\leq \frac{1}{2}.
	\]

	Hence by Markov's inequality, $\Pr{B \geq 1} \leq \frac{1}{2}$ and thus $\tmeet^{\pi} \geq \frac{1}{2} \cdot \frac{C_{\min}}{8 \|\pi\|_2^2}$ in this case. The claim for $\gap$-approximate regular graph follows immediately since $\|\pi\|_2^2 \leq \max_{u \in V} \pi(u) \leq \Gamma/n$. Together with $C_{\min} \geq 1$ this completes the proof of the theorem.
	\end{proof}
	It is interesting to compare the upper bound on $\tmeet$ in
	\autoref{thm:meetingtime} with the bound $\tmeet = O( \frac{1}{1-\lambda_2}
	\cdot ( \frac{1}{\| \pi \|_2^2} + \log n))$ from Cooper et
	al.~\cite[Theorem 2]{CEOR13}. Using the trivial bound
	$C_{\max} \leq \tmix$ and $\tmix=O( \frac{\log n}{1-\lambda_2})$, we
	obtain $\tmeet = O( \frac{1}{1-\lambda_2} \cdot  \frac{\log
	n}{\| \pi \|_2^2})$, which is at most a $\log n$-factor worse. However, for certain graphs like grids or tori one may have a better control on the $t$-step probabilities, so that $C_{\max} \ll \tmix$ could be established.

\begin{proposition}\label{pro:frederik}
Combining the upper bound on $\tmeet$ in
	\autoref{thm:meetingtime} with the bound $\tmeet = O( \frac{1}{1-\lambda_2}
	\cdot (\frac{1}{\| \pi \|_2^2  } + \log n))$ from Cooper et
	al.~\cite[Theorem 2]{CEOR13} together with \autoref{thm:mixtradeoff} we derive
	\[
	\tcoal = O\left( \frac{1}{1-\lambda_2} \cdot\frac{1}{\| \pi \|_2^2}  + \tmix \cdot \log^3 n  \right),
	\]
	which is at least as good as the bound of \cite[Theorem 1]{CEOR13}
	and equally good if one uses the trivial bound  $\tmix=O\left( \frac{\log n}{1-\lambda_2}\right)$.
\end{proposition} 
\lfnote{Consider a bin tree on $n/2$ nodes where the leaves are connected to a star (possibly with another random matching on the leaves) This one has $\tmeet=\Omega(\log n)$ which shows that the bound they give is slightly off. Or properties diam: $\log n/(1-\lambda_2)$ one star to make sure that the pi squared is constant 
}
\begin{proof}
	First assume $\tmeet/\tmix \geq \log^2 n$. In this case, by \autoref{thm:mixtradeoff}, 
	\[\tcoal=O(\tmeet)=O\left( \frac{1}{1-\lambda_2}
	\cdot \left( \frac{1}{\| \pi \|_2^2 } + \log n \right)\right)
	\]
	 follows immediately.
	Next assume $\tmeet/\tmix \leq \log^2 n$, so $\tmeet \leq \tmix \cdot  \log^2 n$.
	By \autoref{lem:beer}, $\tcoal=O(\tmeet \log n)=O\left(\tmix \log^3 n \right).$ 
	Using $\tmix=O\left( \frac{\log n}{1-\lambda_2}\right)$ we derive indeed 
	 $O( \frac{1}{1-\lambda_2}
	\cdot ( \frac{1}{\| \pi \|_2^2} + \log^4 n))$
	the same bound as   \cite[Theorem 1]{CEOR13}.

\end{proof}

In the following, we will try to get more concrete estimates than the ones in \autoref{thm:meetingtime} by expressing the number of expected returns or $C_{\max}$ through $1-\lambda_2$ and $\tmix$. To this end, we define
  \[
  \beta:= \min \left\{ \frac{\log (1/(1-\lambda_2))}{1-\lambda_2}, \tmix \right\}.
\]
Note that since $\tmix = \Omega(1/(1-\lambda_2))$ (\eg \cite{AF14}), we have $\beta \geq \Omega(1/(1-\lambda_2))$. Further, $\beta \leq \tmix = O(\log n/(1-\lambda_2))$. Hence $\beta$ is always sandwiched between the relaxation time $1/(1-\lambda_2)$ and mixing time.

We will frequently make use of the following result, which is a straightforward generalization of a result in the textbook by Aldous and Fill~\cite{AF14} from regular to $\gap$-approximate regular graph.
\begin{lemma}[{\cite[Proposition~6.16~(iii)]{AF14}}]\label{lem:AF}
Let $G$ be any $\gap$-approximate regular graph. Then for any $\tau \leq 5 n^2$,
\[
  \sum_{t=0}^{\tau-1} p_{u,u}^t
\leq 2 \, \gap \cdot \sqrt{5 \tau} .
\]
Since $p_{u,u}^t$ is non-increasing, this implies for any $1 \leq t \leq 5n^2$,
$
  p_{u,u}^t \leq \frac{20 \gap}{\sqrt{t}}.
$
\end{lemma}

\begin{theorem}\label{thm:awesomeregularhit}
For any regular graph we have
\[ \thit = O\left( \frac{n}{\sqrt{1-\lambda_2}}\right) ,\]
and by Cheeger's inequality we obtain $\thit = O (\frac{n}{\Phi})$, where $\Phi$ is the conductance of $G$.
Furthermore, for any non-regular graph with maximum degree $\Delta$, average degree $d$ and minimum degree $\delta$, we have
\[
  \thit = O \left(  \frac{\sqrt{\Delta d}}{\delta} \cdot \frac{n}{\sqrt{1-\lambda_2}} \right) 
\]
\end{theorem}
\begin{proof}
	By \cite[Lemma 10.2]{LPW06} and \cite[Proposition 10.19]{LPW06},\NOTE{Th: Add more specific reference}\fnote{you need to compile with biber then you'll see specific references ;)}
\begin{align} \label{eq:buffalo}
\thit &\leq 2 \max_u \sum_v \thit(v,u)\pi(v) = \frac{2}{\pi(u)} \sum_{t=0}^\infty \left(p_{u,u}^t -\pi(u)\right)  \nonumber\\
&\leq\frac{2}{\pi(u)}  \sum_{t=0}^{2\trel} \left(p_{u,u}^t\right) +
\frac{2}{\pi(u)}\sum_{t= 2\trel +1}^\infty \left(p_{u,u}^t -\pi(u)\right) \nonumber\\
&\leq \frac{O(1)}{\pi(u)}\sqrt{\trel} +
\frac{2}{\pi(u)} 2\trel\sum_{k=0}^\infty \left(p_{u,u}^{2(\sqrt{\trel}+k\trel)} -\pi(u)\right),
\end{align}
where the bound on the first term of the last inequality follows from $ p_{u,u}^{\sqrt{\trel}} \leq O\left(\frac{1}{\sqrt{\trel}}\right)$ (\autoref{lem:AF}).

We can bound the sum as follows using that $\pi(v) = 1/n$ and that for regular graphs any $\tau$ it holds that  $p_{u,u}^{2\tau}=\sum_{v\in V}(p_{u,v}^\tau)^2$ as follows
\begin{align}\label{eq:trail}
	&\sum_{k=0}^\infty \left(p_{u,u}^{2(\sqrt{\trel}+k\trel)} -\frac1n\right) 
	=
	\sum_{k=0}^\infty \left(\sum_{v\in V}
	\left(p_{u,v}^{\sqrt{\trel}+k\trel} \right)^2 -\frac{1}{n}\right) \nonumber\\
	&=	\sum_{k=0}^\infty \left(\sum_{v\in V}
	\left(p_{u,v}^{\sqrt{\trel}+k\trel} \right)^2 -2\sum_{v\in V}p_{u,v}^{\sqrt{\trel}+k\trel} \cdot \frac{1}{n} + \sum_{v\in V} \frac{1}{n^2}\right) \nonumber\\
	 &=	\sum_{k=0}^\infty \norm{p_{u,\cdot}^{\sqrt{\trel}+k\trel} -\frac{1}{n}}^2
	 =	\sum_{k=0}^\infty \norm{P^{ k\cdot\trel}p_{u,\cdot}^{\sqrt{\trel}} -\frac{1}{n}}^2
	 \leq \sum_{k=0}^\infty \lambda_2^{k \cdot \trel}
	 \norm{p_{u,\cdot}^{\sqrt{\trel}} -\frac{1}{n}}^2\nonumber \\
	 &= \sum_{k=0}^\infty \lambda_2^{k \cdot \trel}
	 \left(  \sum_{v\in V} p_{u,v}^{2\sqrt{\trel}} -2 \sum_{v\in V} p_{u,v}^{\sqrt{\trel}} +  \sum_{v\in V} \frac{1}{n^2} \right)\nonumber \\
	 &\leq	\sum_{k=0}^\infty \lambda_2^{k \cdot \trel} p_{u,u}^{2\sqrt{\trel}}\stackrel{(*)}{=}O\left( 
	 \sum_{k=0}^\infty \lambda_2^{k \cdot \trel} \frac{1}{\sqrt{\trel}}\right)\stackrel{(**)}{=}O\left( 
	 \frac{1}{\sqrt{\trel}}
	 \right),
\end{align}
where $(*)$ follows from  $ p_{u,u}^{\sqrt{\trel}} \leq O\left(\frac{1}{\sqrt{\trel}}\right)$ (\autoref{lem:AF})
and $(**)$ follows since $f(y) = y^{1/(1-y)}$ is bounded from above by $1/e$ for any $y \in (0,1)$ and hence the sum is a geometric series.
Combining \eqref{eq:buffalo} into \eqref{eq:trail} yields the claim.

To obtain the result for non-regular graphs, we consider the modified Markov chain with transition matrix $Q$ where the loop probability of every vertex is $1-\frac{\deg(u)}{\Delta}$. As a result, every transition of the walk to another vertex is made with probability $1/\Delta$. Thus $Q$ is symmetric and the stationary distribution $\pi_{Q}$ is uniform. We can apply the result from the first statement to $Q$ and it only remains to relate $\lambda_2(Q)$ to $\lambda_2(P)$.  The variational characterization of $\lambda_2(P)$ gives:
\begin{align*}
  1 - \lambda_2(Q) &= \inf_{\varphi: V \rightarrow \mathbb{R}, \varphi~\text{non-constant}} \frac{ \sum_{u,v \in V} (\varphi(u) - \varphi(v))^2 \pi_Q(u) Q_{u,v}}{ \sum_{u,v \in V} (\varphi(u) - \varphi(v))^2 \pi_Q(u) \pi_Q(v) } \\
  &= \inf_{\varphi: V \rightarrow \mathbb{R}, \varphi~\text{non-constant}} \frac{ \sum_{u,v \in V} (\varphi(u) - \varphi(v))^2 \frac{1}{n \Delta}}{ \sum_{u,v \in V} (\varphi(u) - \varphi(v))^2 \frac{1}{n^2} }
\end{align*}
Similarly,
\[
  1 - \lambda_2(P) = \inf_{\varphi: V \rightarrow \mathbb{R}, \varphi~\text{non-constant}} \frac{ \sum_{u,v \in V} (\varphi(u) - \varphi(v))^2 \frac{1}{2|E|} }{ \sum_{u,v \in V} (\varphi(u) - \varphi(v))^2 \cdot \frac{\deg(u)}{2|E|} \cdot \frac{\deg(v)}{2|E|} }.
\]
Comparing the two equations, we can see that
\[
 1- \lambda_2(Q) \geq (1- \lambda_2(P)) \cdot \frac{\Delta}{d} \cdot \left( \frac{d}{\delta} \right)^2.
\]
\end{proof}

It turns out that the hitting time bound of \autoref{thm:awesomeregularhit} is tight in the sense that for for any  $\Phi$ there exists a graph with conductance $\Phi$ and hitting time and coalescence time of order $\Omega(n/\Phi)$.
\begin{proposition}[{\cite{BGKM16} }]\label{lem:lowerbound}
For every $n$, $d\geq 3$, and constant $\Phi$, there exists
a $d$-regular graph $G$ with $n$ nodes and a constant conductance such that the expected consensus time on $G$ is $\Omega(n ) $.
Furthermore, for every even $n$, $\Phi>1/n$, and constant $d$, there exists a $d$-regular graph $G$ with $\Theta(n)$ nodes and a conductance of $\Theta(\Phi)$ such that the meeting time time on $G$ is $\Omega(n/\Phi ) $. Therefore, the coalescence time and hitting time are of order $\Omega(n/\Phi)$.
\end{proposition}

\begin{theorem}\label{thm:cooperimproved}
Let $G$ be  any non-regular graph with maximum degree $\Delta$, average degree $d$ and minimum degree $\delta$, we have
\[
  \tmeet = O \left(  \frac{\sqrt{\Delta d}}{\delta} \cdot \frac{n}{\sqrt{1-\lambda_2}} \right) .
\]
In particular, 
\[
 \tcoal = O(\thit \log (\Delta/\delta)) = O \left(  \frac{\sqrt{\Delta d}}{\delta} \cdot \frac{n}{\sqrt{1-\lambda_2}} \log (\Delta/\delta) \right) 
\]

Furthermore,
\[ 
 \tcov  = O \left(  \frac{\sqrt{\Delta d}}{\delta} \cdot \frac{n}{\sqrt{1-\lambda_2}} \log n\right)
\]
\end{theorem}
The upper bound on $\tmeet$ and $\tcoal$ gives $\tmeet=O(n^2)$ for cycles and paths, and $\tmeet=O(n)$ on regular expanders (since $1/(1-\lambda_2)=O(1)$).\fnote{remove the stuff before?}
 It thus improves the bound by Cooper \etal~\cite[Theorem~1]{CEOR13}, which states that for any regular graph, $\tmeet = O(n/(1-\lambda_2))$. 

\begin{proof}
The proof of the first part follows from \autoref{thm:awesomeregularhit} and $\tmeet \leq 4 \thit$ (\autoref{pro:relatingmeetandhit}).
The Second part is due to \autoref{thm:awesomeregularhit}  and \autoref{thm:hittingtime}.
The last statement follows from  \autoref{thm:cooperimproved} and the well-known trivial bound $\tcov=O(\thit \cdot \log n)$.
\end{proof}

For any $\gap=O(1)$-approximate regular graph, we also improve the best-known bound on the cover time $\tcov$ in terms of the eigenvalue gap, which is $\tcov = O\left( n \log n/ (1-\lambda_2) \right)$ established by Broder and Karlin in 1989~\cite{BK89}.

As mentioned earlier, $\sum_{t=0}^{\tmix-1} p_{u,u}^t \leq \tmix \cdot \pi(u) + \frac{1}{1-\lambda_2}$ 
are well-known bounds. The next corollary provides an improvement in many cases:

\begin{corollary}\label{cor:returns}
For any $\gap$-approximate regular graph $G=(V,E)$  
\[
  R_{\max} = \max_{u \in V} \sum_{t=0}^{\tmix-1} p_{u,u}^t =O\left( \min \left\{ \gap \sqrt{\frac{d}{\delta}\frac{\log (1/(1-\lambda_2))}{1-\lambda_2}},  \gap^{3/2} \sqrt{\tmix} \right\} \right)= O\left( \gap^{3/2} \cdot \sqrt{\beta} \right).
\]
\lfnote{We might want to derive the exact constants commented out in the following }
\end{corollary}

\begin{proof}

Thus, using that $p^t_{u,u}$ is non-increasing (\eg \autoref{lem:loop}) and \autoref{lem:AF}, we derive
for any $\tau$
 \begin{align}\label{eq:SDAX}
 	 \sum_{t=0}^{\tau-1} p_{u,u}^t & \leq  2\gap \sqrt{5\tau} +\sum_{t=5n^2}^{\tau-1} p^{5n^2}_{u,u}\leq  2\gap \sqrt{5\tau}+\tau \cdot \frac{20\Gamma}{\sqrt{5}n},
 	  \end{align}
where we used that $p^{5n^2}_{u,u} \leq \frac{20\Gamma}{\sqrt{5}n}$, by \autoref{lem:AF}.
In particular, using $\tmix\leq \thit = O(\gap\cdot  n^2)$  (\cite[Corollary 6.9]{AF14})
 \begin{align*}
 	R_{\max} &\leq 2 \gap \sqrt{5 \tmix}+\tmix \cdot \frac{20\Gamma}{\sqrt{5}n} \\
 	&= O\left(\gap \cdot \sqrt{\tmix} \left(1+ \sqrt{\tmix}/n\right) \right) \stackrel{\tmix = O(\gap\cdot  n^2)}{=} 	
O\left(\gap^{3/2} \sqrt{\tmix}\right)
 \end{align*}

In the remainder we derive a bound in terms on $R_{\max}$ in terms of $\frac{\log (1/(1-\lambda_2))}{1-\lambda_2}$.
We split the expected number of returns to $u$ at time $x \leq \tmix$ and obtain
\begin{align*}
   \sum_{t=0}^{\tmix-1} p_{u,u}^t &=  \sum_{t=0}^{x-1} p_{u,u}^t
  +  \sum_{t=x}^{\tmix-1} p_{u,u}^t
  \stackrel{ \text{\autoref{lem:loop} \& \eqref{eq:SDAX}}}{\leq}  \left(2 \gap \cdot \sqrt{5x} + x \cdot 20\gap/n \right)+ \sum_{t=x}^{\tmix-1} \left( \pi(u) + \lambda_2^t \right) 
  \\
   & \leq 2 \gap \cdot \sqrt{5x} + 21 \gap \cdot \tmix / n + \frac{\lambda_2^x }{1-\lambda_2}.
\end{align*}
Next choose $x=\frac{\ln\left( 1/(1-\lambda_2) \right)}{1-\lambda_2}$. Since $f(y) = y^{1/(1-y)}$ is bounded from above by $1/e$ for any $y \in (0,1)$, we have
\begin{align*}
   \sum_{t=0}^{\tmix-1} p_{u,u}^t \leq 10 \gap \cdot \sqrt{ \frac{d}{\delta} \frac{\ln\left( 1/(1-\lambda_2) \right)}{1-\lambda_2} }  + 21\Gamma \cdot \tmix/n + 1.
\end{align*}
We next prove that the second term in the bound above is always asymptotically upper bounded by the first one. This is established via a simple case distinction. First, if $1/(1-\lambda_2) \leq n^{1/2}$, then the claim holds because of $\tmix=O(\log n/(1-\lambda_2))=O(\sqrt{n}\log n)$ and hence $\gap \tmix/n =o( \gap )$ whereas the first term is $\Omega(\gap)$.
Secondly, if $1/(1-\lambda_2) \geq n^{1/2}$, then using the same bound on $\tmix$ along with the fact that $1-\lambda_2 \geq \frac{\delta}{d n^2 }$, where $d$ is the average degree and $\delta$ the minimum degree:

By
\cite[Lemma 1.9]{chung1997spectral}, we have $1-\lambda_2 \geq \frac{1}{\diam d n }$, and
since the diameter of a graph is at most $n/\delta$, we get
$1-\lambda_2 \geq \frac{\delta}{ d n^2 }$.

\[
   \tmix = O \left(\frac{\log n}{\sqrt{1-\lambda_2} \cdot \sqrt{1-\lambda_2}} \right)
   = O \left( \sqrt{\frac{d n^2}{\delta}} \frac{ 2 \log (n^{1/2})}{\sqrt{1-\lambda_2}}   \right)
   = O \left( n \cdot  \sqrt{ \frac{d}{\delta} \frac{ \log(1/(1-\lambda_2))  }{1-\lambda_2}  }  \right).
\]

\end{proof}

We now derive an extension of \autoref{thm:cooperimproved} that is more suited for graphs with a very high degree discrepancy.
\begin{theorem}\label{thm:nonreg}
Let $G=(V,E)$ be any $\gap$-approximate regular graph. Then,
\[
  \tmeet =  O \left( \frac{\gap^{3/2} \log_2(\gap) \cdot \sqrt{\beta}}{ \| \pi\|_2^2}   \right).
\]
Furthermore,
\[
  \tmeet = O \left(
  \frac{\frac{\log (\gap)}{1-\lambda_2} + \log_2(\gap) \cdot \tmix \cdot \max_{u \in V} \pi(u)   }{ \| \pi \|_2^2}
  \right).
\]
\end{theorem}
We point out that for constant $\gap$, the first statement of the theorem recovers the second statement of \autoref{thm:cooperimproved}.
\begin{proof}
Similar to the proof of \autoref{thm:meetingtime},
we define $S_i := \left\{ u \in V \colon \deg(u) \in (2^{i-1},2^{i}] \right\}$, and note that $S_0,\ldots,S_{\log_2 n}$ forms a partition of $V$. Since the graph is $\gap$-approximate regular, at most $4 + \log_2 (\gap)$ of the $S_i$'s are non-empty. Hence there exists a set $S_j$ with
	\[
	   \sum_{j \in S_j} \pi(j)^2 \geq \frac{1}{4 + \log_2(\gap)} \cdot \| \pi \|_2^2.
	\]
We will only count collisions on vertices in that bucket, \ie $Z:=\sum_{t=\tsep}^{\tsep+\tmix-1} \mathbf{1}_{X_t=Y_t} \cdot \mathbf{1}_{X_t \in S_j}$, where $(X_t)_{t\geq 0}$ and $(Y_t)_{t \geq 0}$ are two arbitrary walks. Then,
\begin{align*}
  \E{Z} \geq \frac{\tmix \cdot \|\pi\|_2^2}{e^2 \cdot (4 + \log (\gap))} .
\end{align*}
Furthermore,
\begin{align*}
  \E{Z \, \mid \, Z \geq 1} \leq \max_{u \in S_j} \sum_{t=0}^{\tmix-1} \sum_{v \in S_j} \left( p_{u,v}^t \right)^2 
  &\leq 4 \max_{u \in S_j} \sum_{t=0}^{\tmix-1} p_{u,u}^{2t} = O\left(\gap^{3/2} \cdot \sqrt{\beta}\right) 
  \end{align*}
  where the second inequality holds due to the fact that vertices in $S_j$ have the same degree up to a factor of $2$ and the final inequality holds due to \autoref{cor:returns}.
Plugging the two bounds into \eqref{eq:central} yields
\begin{align*}
  \Pr{ Z \geq 1} &= \Omega \left( \frac{\tmix \cdot \| \pi \|_2^2}{ \gap^{3/2} \log_2(\gap) \cdot \sqrt{\beta}} \right) =: p.
\end{align*}
Hence, by iterating over consecutive time-intervals of length $\tsep+\tmix \leq 5\tmix$ that are independent, we conclude
\begin{align*}
 \tmeet &\leq \frac{1}{p} \cdot 5 \, \tmix =
 O \left( \frac{\gap^{3/2} \log_2(\gap) \cdot \sqrt{\beta}}{ \| \pi\|_2^2}   \right).
\end{align*}
For the second statement, we also have, by \autoref{lem:loop}, 
  \begin{align*}
   4 \max_{u \in S_i} \sum_{t=0}^{\tmix-1} p_{u,u}^{t} \leq 4 \cdot \frac{1}{1-\lambda_2} + \tmix \cdot \max_{u \in V}\pi(u),
  \end{align*}
and the bound on $\tmeet$ is derived in exactly the same way as before.
\end{proof}

\subsection{Relating Meeting Time to $\thit$}\label{sec:meetvt}

In this section we prove the following proposition which can be seen as an
analogous version of \cite[Proposition 14.5]{AF14}  in discrete time. 

\begin{proposition}\label{pro:relatingmeetandhit}
For any graph $G=(V,E)$ and $u,v \in V$ we have  
\[
 (\min_{u'\in V} \thit( \pi, u') +  \thit(u,v)-\thit(\pi, v))/2 \leq  \tmeet(u,v) 
\leq 2(\max_{u'\in V} \thit(\pi, u')+\thit(u,v)-\thit(\pi, v) ).
\]
Consequently, for any graph we have $\tmeet \leq 4\thit$ and  for any vertex transitive graph $G$ we have $\thit/2 \leq  \tmeet\leq 2\thit$. 
\end{proposition}
\begin{proof}

We define a pair of chains $((X_t)_{t\geq 0},(Y_t)_{t\geq 0})$ with \emph{arbitrary} start vertices $X_0,Y_0 \in V$, called  \emph{sequential} random walks, by 
\begin{align*}
	X_{t+1} &= \begin{cases}
	  v \in N(X_t) \mbox{  w.p. $\frac{1}{2|N(X_t)|}$}  & \mbox{if $t$ is even} \\
	X_t					&  \mbox{otherwise}	
	\end{cases},\\
	\intertext{ and }
		Y_{t+1} &= \begin{cases}
			 v \in N(Y_t) \mbox{  w.p. $\frac{1}{2|N(Y_t)|}$}				&  \mbox{if $t$ is odd}	\\
	 Y_t  & \mbox{otherwise} 
	\end{cases}.
\end{align*}
In particular, for odd $t$ (even $t$, respectively)  the random-walk is lazy meaning $X_{t+1}=X_t$ (and $Y_{t+1}=Y_t$, respectively).

Consider two ``non-sequential'' random walks
$(X'_t)_{t\geq 0}$ and $(Y'_t)_{t\geq 0}$ with $X'_0=X_0$ and $Y'_0=Y_0$.
We will couple their decisions with the walks $(X_t)_{t\geq 0}$ and $(Y_t)_{t\geq 0}$, by setting
$X'_t=X_{2t} $ and $Y'_t=Y_{2t}$.
Due to this coupling and since each random walk is lazy w.p. $1/2$,
\[ 
\Pr{X'_{t+1} = Y'_{t+1} \,\mid\,    X_{2t}=Y_{2t}}=\Pr{X_{2t+2}=Y_{2t+2} \,\mid\,    X_{2t}=Y_{2t}} \geq 1/4,
\]
and
\[ 
\Pr{X'_{t+1} = Y'_{t+1} \,\mid\,    X_{2t+1}=Y_{2t+1}}=\Pr{X_{2t+2}=Y_{2t+2} \,\mid\,    X_{2t+1}=Y_{2t+1}} =1/2
\]

Let $\tmeet^{seq}(u,v)$ be the meeting time of the sequential chains $X_t$ and $Y_t$, \ie
\[
\tmeet^{seq}(u,v)=\min\{ t\geq 0~\mid~ X_t=Y_t , X_0=u , Y_0=v\}
\]
and $\tmeet^{seq}=\max_{u,v} \tmeet^{seq}(u,v)$.
We seek to relate $\tmeet$ with $\tmeet^{seq}$. Clearly,
$\tmeet^{seq}/2 \leq \tmeet $ since a meeting of $X'_t=Y'_t$ implies that
 $X_{2t}=Y_{2t}$.
For an upper bound on $\tmeet$
recall that $X'_t$ and $Y'_t$ meet, \ie  $X'_t=Y'_t$ w.p. at least $1/4$ whenever  $X_{2t-2}=Y_{2t-2}$ or $X_{2t-1}=Y_{2t-1}$.
Hence, by independence $\tmeet = \max_{u,v} \tmeet(u,v) \leq 4 (\tmeet^{seq}/2) = 2 \tmeet^{seq}.$
We conclude,
\begin{equation}\label{Honigkuchen}
\tmeet^{seq}(u,v)/2 \leq \tmeet(u,v) \leq 2\tmeet^{seq}(u,v).	
\end{equation}

We proceed by deriving upper and lower bounds on $\tmeet^{seq}$, which  gives us bounds on $\tmeet$. We will make use of the following statement that is a weaker version of the original statement \citet[Proposition 3.3]{AF14}.
For all $u,v\in V$ we have
\[
\min_{u'\in V} \thit( \pi, u') \leq \tmeet^{seq}(u,v) - (\thit(u,v)-\thit(\pi, v))\leq \max_{u'\in V} \thit(\pi, u')
\]
Using \eqref{Honigkuchen} we derive,
\[ 
\tmeet(u,v) \leq 2\tmeet^{seq}(u,v) \leq 2(\thit(u,v)-\thit(\pi, v) + \max_{u'\in V} \thit(\pi, u')),
\]
and
\[
\tmeet(u,v) \geq  \tmeet^{seq}(u,v)/2 \geq (\min_{u'\in V} \thit( \pi, u') +  (\thit(u,v)-\thit(\pi, v)))/2, 
\]
which yields the first part of the claim.
For vertex transitive chains we get using  $ \thit(\pi, u)=  \thit(\pi, u')$ for all $u,u'\in V$ and thus
\[
	\tmeet^{seq}(u,v)= \thit(u,v).
\]
Thus, putting everything together and fixing $u,v\in V$ to be the nodes maximizing $\thit(u,v)$, we derive 
\[ 
\thit =  \thit(u,v) = \tmeet^{seq}(u,v) \leq  \tmeet^{seq} \leq 2\tmeet.
\]
Similarly,
\[ 
	\tmeet \leq   2\tmeet^{seq}= 2 \max_{u,v} \tmeet^{seq}(u,v) \leq 2\thit.
\]
This yields \autoref{pro:relatingmeetandhit}.
\end{proof}

\lfnote{checkout if we can get Proposition 14.6 - relates it to the variation threshold}

\subsection{Proof of \autoref{thm:spectral}}
The proof follows from \autoref{thm:awesomeregularhit} and \autoref{thm:cooperimproved}.

\section{Applications to Concrete Topologies}\label{sec:special}

Here we derive $\tmeet$ and $\tcoal$ on specific topologies. Note that more general bounds for certain graph classes like regular graphs or vertex-transitive graphs have been stated earlier, see, \eg \autoref{thm:hittingtime} or \autoref{sec:meeting}.

\subsection{$2$-Dimensional~Grids/Tori and Paths/Cycles}

Next we apply our machinery to the $2$-dimensional grid and the $2$-dimensional torus.
For the continuous case a manual approach to bound meeting and coalescence times can be found in \cite{C89}.
Thanks to our general results, we can not only easily derive the correct bound on $\tmeet$, but also on $\tcoal$.
First, we recall the following well-known fact that for $2$-dim.~grid and torus: For any integer $t = O(n)$, \lfnote{add a reference} 
\begin{align}
  p_{u,u}^t = \pi(u) + \Omega( t^{-1} ), \label{eq:spectralgrid}
\end{align}
which can be derived, \eg by using the central limit theorem. Further, $\tmix = \Theta(n)$, and combining these two results, we immediately obtain
\[
  C_{\min} := \min_{u \in V} \sum_{t=0}^{\tmix} \sum_{v \in V} (p_{u,v}^t)^2 \geq \frac{1}{2} \min_{u \in V} \sum_{t=0}^{\tmix} p_{u,u}^{2t} = \Omega(\log n).
\]
Thus, by \autoref{thm:meetingtime}.$(iii)$, $\tmeet = \Omega(n \log n)$. For the upper bound, we apply  \autoref{thm:hittingtime} together with
the well-known bound $\thit=O(n\log n)$ to derive $\tcoal =O(\thit)=O(n\log n)$.

For cycles or paths, the corresponding formula to \eq{spectralgrid} is, 
for $t=O(n^2)$,
\begin{align*}
  p_{u,u}^t = \pi(u) + \Omega( t^{-1/2} ).
\end{align*}
 Hence $C_{\min} = \Omega(n)$, and therefore the third statement of \autoref{thm:meetingtime} implies $\tmeet = \Omega(n^2)$. For the upper bound, we apply  \autoref{thm:hittingtime} together with
$\thit=O(n^2)$ to derive $\tcoal =O(\thit)=O(n^2)$. Alternatively, the upper bound on $\tcoal$ could be also shown by using $\tmix = O(n^2)$ and applying the third statement \autoref{thm:cooperimproved}.

\subsection{$d$-Dimensional~Grids and Tori, $d \geq 3$}

Here the bounds on $\tmeet$ and $\tcoal$ follow immediately from our general results. First, for any regular graphs we have $\tmeet = \Omega(n)$ (\autoref{thm:meetingtime}.iii). Further, it is well-known that $\thit = O(n)$~(\eg~\cite{LPW06}), and the result follows by $\tcoal=O(\thit)$ shown in~\autoref{thm:hittingtime}. Alternatively, we could also use $\tmeet=O(\thit)$ (\autoref{pro:relatingmeetandhit}) to deduce $\tmeet=O(n)$. Combining this with the fact that $\tmix=O(n^{2/d})$~\cite{AF14}, \autoref{thm:mixtradeoff} yields the correct bound $\tcoal = O(n)$.

\subsection{Hypercubes}
Tight bounds for the hypercube can be obtained through different tools we provide.

Firstly, it follows trivially from \autoref{thm:hittingtime}:
Since the  hypercube is regular (in fact, it is even vertex-transitive) it suffices to consider the hitting time.
 We have  $\thit=O(n)$ (see \eg~\cite{L93}) and
recall that $\thit=\Omega(n)$ by \autoref{thm:meetingtime}. Hence applying~ \autoref{thm:hittingtime} yields $\tcoal=\Theta(\thit)=\Theta(n)$.

Alternatively, we could also use the more elementary bound $\tmeet \leq 2 \thit$ by \autoref{pro:relatingmeetandhit} to conclude $\tmeet=O(n)$. Since it is a well-known fact that $\tmix = O(\log n \cdot \log \log n)$ \cite{LPW06}, we obtain by~\autoref{thm:mixtradeoff} that $\tcoal = \Theta(\tmeet) = \Theta(n)$. 

\subsection{(Regular) Expanders}

It is not surprising that on regular expander graphs, we have $\tcoal = \Theta(n)$ and there is a multitude of approaches to establish this (for instance, the result is a consequence of the main result by \cite{CEOR13}). With regard to our bounds, the easiest route is to follow the arguments for the hypercube: Combine the result $\thit = O(n)$ (\eg~\cite{BK89}) together with our bound $\tcoal = O(\thit)$ (\autoref{thm:hittingtime}). The lower bound $\tcoal \geq \tmeet = \Omega(n)$ holds for any regular graph.

\subsection{Real World Graph Models}\label{sec:realworld}

There is a variety of different graph models for ``real world'' networks. In this subsection we  demonstrate that  random walks coalesce quickly on these graphs by establishing several  bounds on $\tcoal$ which are sublinear in $n$.

First note that common features of real world graph models are (i) a power law degree distribution with exponent $\beta \in (2,3)$ and (ii) high expansion, \ie $1-\lambda_2$ is not too large, and hence $\tmix=O(\log n)$. Notice that (i) $\beta \in (2,3)$ implies that w.h.p. we have $\Delta=O(n^{1-\epsilon})$, and hence $\| \pi \|_2^2 \leq \max_{u \in V} \pi(u) \leq n^{-\epsilon}$, for $\varepsilon >0$. 

For the sake of concreteness, let us take a specific model by Gkantsidis, Mihail and Saberi~\cite{GMS03}, which was also analyzed by Cooper \etal~\cite{CEOR13}. In this model, for some $\alpha \in (2,3)$ we generate a random graph which has $\Theta(n/d^{\alpha})$ vertices of degree $d$ and an eigenvalue gap $1/(1-\lambda_2) = O(\log^2 n)$. Cooper \etal~\cite{CEOR13} derived the general bound $\tcoal=O(\frac{1}{1-\lambda_2} \cdot (\|\pi\|_2^2 + \log^ 4n))$, which implies $\tcoal = O(n^{(\alpha-1)/2} \cdot \log^2 n)$ - a sublinear bound on the coalescing time. However, this leaves open how close $\tcoal$ and $\tmeet$ are.

Combining  \autoref{thm:mixtradeoff} with the fact that $\tmix = O( \log n/(1-\lambda_2))=O(\log ^3 n)$, we immediately obtain $\tcoal = \Theta(\tmeet)$, without having to know the actual value of $\tmeet$.\footnote{That being said, deriving the correct bound on $\tmeet$ is an interesting open problem. So far, it seems rather difficult to use one of our ``off-the-shelf'' bounds or the results from~\cite{CEOR13}. One potential route towards a tight bound may involve stronger bounds on $R_{\max}$, as suggested by the second upper bound on $\tmeet(u,v)$ in~\autoref{thm:meetingtime}.}

More generally, we have the following result, saying that we have $\tcoal = \Theta(\tmeet)$ whenever $\tmix$ is slightly smaller than $1/\| \pi \|_2^2$:

\begin{theorem}\label{thm:sloppy} \NOTE{T: This corresponds to the Oliveira Bound. Maybe it could be moved to Sec 3, or dropped altogether?}
Let $G=(V,E)$ be any graph. Then,
\[
  \tcoal = O \left( \tmeet \cdot \left(1 + \sqrt{ \tmix \cdot \| \pi \|_2^2} \cdot \log n \right) \right).
\]
In particular, whenever $ \tmix \cdot \log^2 n \leq 1/\| \pi \|_2^2 $, we have $\tcoal = \Theta(\tmeet)$.
\end{theorem}
\begin{proof}
First, by the third statement of \autoref{thm:meetingtime}, we have 
$
 \tmeet \geq \frac{1}{9 \| \pi \|_2^2}.
$
Inserting this into~\autoref{thm:mixtradeoff} yields the upper bound.
The lower bound for the setting $ \tmix \cdot \log^2 n \leq 1/\| \pi \|_2^2 $ trivially holds since $\tcoal\geq \tmeet$.
\end{proof}
It is worth comparing this result with the bound derived by Cooper \etal~\cite{CEOR13}:
\begin{align}
    \tcoal = O \left( \frac{1}{1-\lambda_2} \cdot \left( \frac{1}{\| \pi \|_2^2} + \log^4 n \right)   \right). \label{eq:cooper}
\end{align}
The advantage of \eqref{eq:cooper} is that requires relatively little knowledge about $G$; only $\frac{1}{1-\lambda_2}$ and $\frac{1}{\| \pi \|_2^2}$ (which is equivalent to knowing the degree distribution) are needed. 
One potential drawback of the bound in \eq{cooper} however, is
that it involves the product of two factors $\frac{1}{1-\lambda_2}$ and $\frac{1}{\| \pi \|_2^2}$, each of which is a lower bound on the meeting time on its own.
For instance for regular graphs, by \autoref{thm:cooperimproved}, we immediately obtain that $\tcoal = O(\tmeet \log n) = O \left(n \log n \cdot \sqrt{\frac{1}{1-\lambda_2} \cdot \log(\frac{1}{1-\lambda_2})} \right)$. As a consequence, for regular graphs, our bound improves over the bound in \eqref{eq:cooper} whenever $\frac{1}{1-\lambda_2} \geq \log^{2+\epsilon} n$ for an arbitrarily small constant $\epsilon>0$. 

It is also interesting to consider an alternative graph model for real world networks, proposed by \citet{MPS06}. Also in this model, the degree distribution has the same Power law with exponent $\alpha \in (2,3)$, but there is a stronger bound on the spectral gap, $\frac{1}{1-\lambda_2}=O(1)$ \cite{MPS06}. Hence \autoref{thm:sloppy} implies $\tcoal=\Theta(\tmeet)$.\fnote{to argue this formally we might want to derive a bound on pisuared. } Further, thanks to \autoref{pro:frederik} (or alternatively, the bound by Cooper \etal~\eqref{eq:cooper}) we get the explicit bound $\tcoal = O( 1/\| \pi \|_2^2)$, which is asymptotically tight due to the trivial lower bound $\tmeet = \Omega(1 / \| \pi \|_2^2)$~(\autoref{thm:meetingtime}).

\subsection{Binary Trees}

In this subsection, we derive a lower bound $\tmeet = \Omega(n \log n)$ for complete binary trees. Unfortunately, this bound does not follow directly from our general results and a manual analysis is required. To some extent, this is due to the structural difference between nodes close to the leaves and nodes close to the root. While a collision close to the leaves triggers $\Theta(n \log n)$ expected additional collisions, a collision near the root triggers only $\Theta(n)$ additional collisions. 

Our proof consists of the following two steps. In \autoref{subsec:one}, we first provide a lower bound on the probability that a random walk starting from any nodes $u \in V$ is on a leaf after $O(\log n)$ steps. We also show that any $t$-step probability $p_{u,v}^t$ is at the most return probability for a leaf. Both results shown in \autoref{lem:treeloop} are derived by projecting the random walk on the tree to a random walk on a weighted path of length $\log_2 n - 1$. 

In \autoref{subsec:two}, we proceed to analyzing the expected number of collisions between two random walks in $n$ steps. The main component is \autoref{lem:hutmitdreiecken}, establishing that this number is at least $\Omega(\log n)$ provided the walks start from the same vertex not too far from the root. This result is complemented by a union-bound type argument in \autoref{lem:evente}, showing that it is unlikely that two random walks collide on a vertex close to the root. Combining the two results and applying them to \eqref{eq:central} establishes the desired lower bound $\tmeet = \Omega(n \log n)$. 

For the other bounds on $\tmeet$ and $\tcoal$, we combine $\tcoal=O(\thit)$ (\autoref{thm:hittingtime}) with the well-known fact $\thit = O(n \log n)$~(cf.~\cite{AF14}) to obtain $\tcoal=O(n \log n)$. Together with the established lower bound, this shows that $\tmeet$ and $\tcoal$ are both of order $\Theta(n \log n)$.

 \subsubsection{Bounds on the $t$-step probabilities}\label{subsec:one}
 
 We assume that the complete binary tree has $\log_2 n - 1$ levels, \ie there are $n/2$ leaves and the total number of nodes is $n-1$.
We define $\mathcal{L} \subseteq V$ to be the set of leaves. For the analysis, it will be helpful to relate a random walk on the binary tree to a corresponding random walk on a weighted path $\tilde G$ of length $\log_2 n -1$ with nodes $\tilde V= \{1,2,\ldots,\log_2 n-1\}$, 
where each vertex on the path corresponds to all vertices in the binary tree on the same level. 
Let $Q$ denote the $(\log_2 n-1) \times (\log_2 n-1)$
transition matrix of the corresponding weighted random walk. For $i,j \in \tilde V$ we have
\[
q_{ij} =\begin{cases} \ifrac{1}{2} & \text{ if $i =j$ } \\
		\ifrac{1}{3} & \text{ if $j > i, i\neq 1$ } \\
				\ifrac{1}{2} & \text{ if $j > i, i = 1$ } \\
	\ifrac{1}{6} & \text{ if $j < i, i\neq \log_2 n -1$ } \\
	\ifrac{1}{2} & \text{ if $j < i, i= \log_2 n -1$ } .
\end{cases}
\]
Let $\tilde \pi$ denote the stationary distribution of this process. Since the random walk on $\tilde{G}$ is also lazy, \autoref{lem:loop} implies that $q_{\log_2 n-1, \log_2 n-1}^t \geq \frac{1}{4}$ for all $t \geq 0$.   
Define $\tau_{u,\mathcal{L}} := \min_{v \in \mathcal{L}} \tau_{u,v}$  to be the first time-step  a leaf is visited, where the walk starts at $u$; $\tau_{u,\mathcal{L}}=0$ if $u$ is a leaf. We will frequently use the following two simple facts about random walks on binary trees:

\begin{lemma}[\cite{Mo73}]\label{lem:treefact}
Let $G$ be any tree, and $u$ and $v$ be two adjacent nodes. Then $\thit(u,v) = 2 \cdot n_{u,v} - 1$, where $n_{u,v}$ is the number of vertices in the subtree containing $u$ obtained by deleting the edge $\{u,v\}$.  
\end{lemma}

\begin{lemma}\label{lem:treeloop}
Let $G$ be a complete binary tree, and let $u$ be an arbitrary node. Then the following statements hold: there is a constant $c_1 > 0$, so that
\begin{enumerate}
\item  for any $t \geq c_1 \log_2 n$, $p_{u,\mathcal{L}}^{t} \geq \frac{1}{5}$. Moreover, if $u$ is a leaf, then the same inequality holds for all $t \geq 0$.
\item For any vertex $v \in V$, $p_{u,v}^t \leq 7 \cdot p_{w,w}^{t- 2c_1 \log _2 n}$, where $w$ is any leaf.
\end{enumerate}
\end{lemma}
\begin{proof}
To derive a lower bound on $p_{u,\mathcal{L}}^t$, we consider the contracted binary tree
$\tilde G$. 
Recall that $\tau_{u,\mathcal{L}}$ is the random variable of the first time-step at which the  random walk visits a leaf where the random walk starts at $u$; $\tau_{u,\mathcal{L}}=0$ if $u$ is a leaf. By conditioning on the first visit to a leaf, 
\begin{align*}
 p_{u,\mathcal{L}}^t &\geq \sum_{s=0}^{2c_1 \log_2 n} \Pr{ \tau_{u,\mathcal{L}} = s} \cdot q_{\mathcal{L},\mathcal{L}}^{t-s} 
 			\geq \frac{1}{4} \cdot \left(1 - \Pr{ \tau_{u,\mathcal{L}} > 2c_1 \log_2 n } \right),
 \end{align*}
 since the chain $Q$ can be seen as a projection of $P$ to the line.
Our next claim is that $\Pr{ \tau_{u,\mathcal{L}} > c_1 \log_2 n} \leq n^{-2}$, provided that the constant $c_1 > 0$ is sufficiently large.
This can be derived by coupling the random walk on a binary tree, starting from the root, with a random walk on the integers, starting from zero and waiting until the random walk reaches the vertex $\log_2 n - 1$. Using a Chernoff bound for $X:=\sum_{i=1}^{c_1 \log_2 n} X_i$, with $\Pr{X_i=+1} = 1/3$, $\Pr{X_i=-1} = 1/6$ and $\Pr{X_i =0} = 1/2$, we conclude that $\Pr{ \tau_{u,\mathcal{L}} > c_1 \log_2 n} \leq n^{-2}$, which implies the first statement. 
 
To prove the second statement, consider the first $c_1 \log_2 n$ steps of a random walk starting at $u$. Similarly as before,
\begin{align*}
  p_{u,v}^t &\leq \sum_{s=0}^{c_1 \log_2 n} \Pr{ \tau_{u,\mathcal{L}} = s} \cdot \max_{w \in \mathcal{L}} p_{w,v}^{t-s} + \Pr{ \tau_{u,\mathcal{L}} > 4 \log_2 n}.
\end{align*}
As seen above, $\Pr{ \tau_{u,\mathcal{L}} > c_1 \log_2 n} \leq n^{-2}$ and therefore 
\begin{align}
 p_{u,v}^t &\leq \max_{t-c_1 \log_2 n \leq s \leq t} ~\max_{w \in \mathcal{L}} p_{w,v}^s + n^{-2}. \label{eq:peach}
\end{align}
Let us now compare $p_{w,v}^t$ to $p_{v,w}^t$. Since the random walk is time-reversible, we have
\begin{align*}
  p_{w,v}^t \cdot \pi(w) = p_{v,w}^t \cdot \pi(v),
\end{align*}
and hence \begin{align*}
 p_{u,v}^t &\leq \max_{t-c_1 \log_2 n \leq s \leq t}~ \max_{w \in \mathcal{L}} 3 \cdot p_{v,w}^s + n^{-2}.
\end{align*}
Applying \eq{peach} to each $p_{v,w}^{s}$, we conclude that
\begin{align*}
 p_{u,v}^t &\leq 6 \cdot \max_{t-2c_1 \log_2 n\leq s \leq t} \max_{w,w' \in \mathcal{L}} p_{w,w'}^{s} + 4 n^{-2}.
\end{align*}
Further, by symmetry $p_{w,w'}^{s}$ is maximized if $w=w'$, so that
\begin{align*}
 p_{u,v}^t &\leq 6 \cdot \max_{t-2c_1 \log_2 n\leq s \leq t} p_{w,w}^{s} + 4 n^{-2},
\end{align*}
where $w$ is any leaf. Applying \autoref{lem:loop}, it follows that the maximum is attained for $s=t-2 c_1 \log_2 n$ and $p_{w,w}^s \geq \frac{1}{2n-2}$, which implies the second statement.
\end{proof}

\subsubsection{Establishing the Lower Bound on the Meeting Time}\label{subsec:two}

We  now prove that the meeting time on binary trees is $\Omega(n \log n)$. The intuition for this is as follows. While two random walks of length $\Theta(n \log n)$ will lead to $\Theta(\log n)$ expected collisions, it turns out that the distribution of collisions is poorly concentrated. In fact we will prove that, conditional on the existence of at least one collision, the expected number of total collisions is $\Omega(\log n)$. This will imply the desired lower bound on the meeting time. A slight complication is that the collision could occur on different nodes, which is why we will first bound the probability for a collision to occur close to the root.

Let us define $U$ to be the set of all nodes that have distance at least $\frac{1}{2} \log_2 n$ from the root. Note that $|V \setminus U| \leq 2 \cdot \sqrt{n}$.
Further, let $\mathcal{E}$ denote the event that two random walks starting from the stationary distribution of length $n \log_2 n$ meet on a vertex in $V \setminus U$.

\begin{lemma}\label{lem:evente}
We have $\Pr{ \mathcal{E} } \leq  n^{-1/3}$.
\end{lemma}
\begin{proof}
By the Union Bound,
$
\Pr{ \mathcal{E} } \leq 
\sum_{t=1}^{n \log_2 n} \sum_{u \in V \setminus U} \pi(u)^2 = n \log_2 n \cdot 2 \sqrt{n} \cdot  \left( \frac{2}{n} \right)^2 \leq \frac{8 \log_2 n}{\sqrt{n}}.
$
\end{proof}

\begin{lemma}\label{lem:hutmitdreiecken}
For any node $u \in U$ and any $2 \sqrt{n} \leq t \leq n/200$, we have $\sum_{v \in \mathcal{L}} (p_{u,v}^t)^2 = \Omega({1}/{t} )$. 
\end{lemma}
\begin{proof}
Recall that $\sum_{v \in \mathcal{L}} (p_{u,v}^t)^2$ is the probability of two non-interacting, independent random walks starting from $u$ to meet at the same leaf at time $t$. 

Our first claim is that with probability at least $c_1 > 0$, both random walks reach a leaf before returning to $u$ within $4 \log_2 n$ steps. To prove this claim, recall that with probability at least $1-2 n^{-2}$, both random walks reach a leaf before step $4 \log_2 n$. Secondly, by \cite[Proposition 2.3]{L93}, applied to the collapsed binary tree  $\tilde{G}$, where node $u$ is at level $\ell$ and the leafs are in level $\log_2 n - 1$, it follows that the probability that a random walk starting at $u$ visits a leaf before returning to $u$ is
\[
  p :=\frac{1}{(\thit^{\tilde{G}}(\ell,\log_2 n-1) + \thit^{\tilde{G}}(\log_2 n-1,\ell)) \cdot \pi^{\tilde{G}}_{\ell}}.
\]
Further, $\thit^{\tilde{G}}(\ell, \log_2 n - 1) = O(\log n)$ and $\thit^{\tilde{G}}(\log_2 n -1, \ell) \leq \min_{w \in \mathcal{L}} \thit(w,u) \leq 2^{\log_2 n - \ell+2}$ by \autoref{lem:treefact} and $\pi^{\tilde{G}}_{\ell} \leq 2^{-\log_2 n +\ell +2}$, where $\dist(\mathcal{L},u) = \min_{w \in \mathcal{L}} \dist(w,u)$. Hence $p$ is at least some constant $>0$. 
Hence with probability at least $p^2 - 2n^{-1}$, both random walks reach a leaf before time $4 \log_2 n$ without returning to $u$.

Consider now the original binary tree, and one of the two random walks starting from a leaf $w$ at some time $\in [1,4 \log_2 n]$ up until time step $t - 4 \log_2 n$. Consider the shortest path from $w$ to the root, and let $z$ be a node that is on this shortest path and has distance $\log_2(100t)$ from $w$. Applying \lemref{treefact},  it follows that $\thit(w,z) =  \sum_{i=1}^{\log_2(100t)} 2^{i} - 1 \geq 49t$. By Markov's inequality, 
\begin{align}
\Pr{ \hit(w,z) \geq 2 \thit(w,z)} \leq 1/2. \label{eq:markoff}
\end{align}
 Now divide the random walk into consecutive epochs of length $2 \thit(w,z) + 4 \log_2 n$. Combining \eqref{eq:markoff} and \autoref{lem:treeloop} it follows that the random walk will visit the vertex $z$ in each epoch with probability at least $\frac{1}{2}$, conditional on having not visited the vertex $z$ in any of the previous epochs. Therefore for any integer $\lambda \geq 1$,
 \begin{align*}
 \Pr{ \hit(w,z) \geq \lambda \cdot 3 \cdot\thit(w,z) ) } \leq \Pr{ \hit(w,z) \geq \lambda \cdot (2 \cdot \thit(w,z) + 4 \log_ 2n ) } \leq 2^{-\lambda}, 
 \end{align*}
 where the first inequality holds since $\thit(w,z) = \Omega(t) = \omega(\log n)$.
 Hence,
 \begin{align*}
  \thit(w,z)  \leq \frac{1}{20} \thit(w,z) + \Pr{ \hit(w,z) \geq \frac{1}{20} \thit(w,z)} \cdot c_1 \cdot 3 \thit(w,z)  + \sum_{\lambda=c_1}^{\infty} 2^{-\lambda} \cdot \lambda \cdot 3 \thit(w,z),
 \end{align*}
 and 
it follows that by choosing the constant $c_1 > 0$ large enough, there is a constant $c_2=c_2(c_1) > 0$ so that
 \begin{align*}
   \Pr{ \hit(w,z) \geq \frac{1}{10} \thit(w,z)} \geq c_2.
 \end{align*}
 Hence with probability at least $c_2 > 0$, the random walk does not reach the node $z$ before time $t$. Further, with probability at least $1-n^{-2}$, the random walk visits a leaf, say, $w'$, before step $t$, say at step $s$, and therefore by \autoref{lem:treeloop}, the random walk is at a leaf at step $t$ with probability at least $(1-n^{-2}) \cdot \frac{1}{5}$. 
 
 \NOTE{T: Usually $F_t$ is the filtration and not an event...}
 Now define $\mathcal{F}_1 := \left\{  \hit(w,z) \geq t \right\}$ and $\mathcal{F}_2 := \left\{ X_t \in \mathcal{L}  \right\}$. Clearly, the events $\mathcal{F}_1$ and $\mathcal{F}_2$ are positively correlated so that
 \begin{align*}
  \Pr{ \mathcal{F}_1 \cap \mathcal{F}_2 } &\geq \Pr{ \mathcal{F}_1 } \cdot \Pr{ \mathcal{F}_2} \geq c_2 \cdot \left(1 - n^{-2} \right) \cdot \frac{1}{5}.
 \end{align*}

Combining all the events, we conclude that 
with constant probability $c_3 > 0$
both random walks are on a leaf at step $t$ and have never left the subtree with root $z$. For one walk, the distribution will be uniform over all the leafs within a subtree whose root is the vertex closest to the root ever visited. Hence let $\mathcal{L}_1$ be all the leafs that have a non-zero probability to be visited at step $t$ by the first random walk, and $\mathcal{L}_2$ similarly. W.l.o.g. let $|\mathcal{L}_1| \leq |\mathcal{L}_2|$ and observe that $\mathcal{L}_1 \subseteq \mathcal{L}_2$ since both walks start at the same node.  Therefore,
\begin{align*}
 \sum_{v \in \mathcal{L}} (p_{u,v}^t)^2 
 &\geq \left( p^2 - 2 n^{-1} \right) \cdot c_3 \cdot \sum_{v \in \mathcal{L}_1 \cap \mathcal{L}_2 }  \frac{1}{|\mathcal{L}_1|} \cdot \frac{1}{|\mathcal{L}_2|} \geq \left( p^2 - 2 n^{-1} \right) \cdot c_3 \cdot \frac{1}{|\mathcal{L}_2|} \geq \left( p^2 - 2 n^{-1} \right) \cdot c_3 \cdot \frac{1}{10t},
\end{align*}
where the last inequality holds since we are conditioning on the event that none of the two random walks reaches the vertex $z$.
\end{proof}

\begin{theorem}\label{thm:treelower}
For the binary tree it holds that $\tmeet=\Omega(n \log n)$.	
\end{theorem}
\begin{proof}
We first only consider collisions on nodes in $U$ by two random walks $(X_t)_{t \geq 0}$, $(Y_t)_{t \geq 0}$ starting from the stationarity distribution.
More formally, we are interested in the random variable
\[
  Z := \sum_{t=1}^{c n \log_2 n} \sum_{v \in V \setminus U} \mathbf{1}_{X_t=Y_t=v}.
\]
By linearity of expectations,
\[
 \E{Z} \leq \sum_{t=1}^{c n \log_2 n} \sum_{v \in V} \Pr{X_t=Y_t=v} \leq c n \log_2 n \cdot \sum_{v \in V} \pi(v)^2 \leq n \log_2 n \cdot n \cdot\left( \frac{2}{n} \right)^2 = 4c \log_2 n,
\]
and clearly,
\begin{align*}
 \Pr{ Z \geq 1} &= \frac{ \E{Z } }{ \E{Z \, \mid \, Z \geq 1}}.
\end{align*}
Hence to derive an upper bound on $\Pr{ Z \geq 1}$, we will derive a lower bound on $\E{Z \,\mid\, Z \geq 1}$. In order to this, it will be helpful if we can work under the assumption that the first collision occurs in the first half of the walk. To this end, let $Z_1$ be the indicator random variable that is $1$ if a collision appears before time $\frac{1}{2} n \log_2 n$ and $Z_2$ be the indicator random variable if a collision appears after time $\frac{1}{2} n \log_2 n$. 
Applying \autoref{lem:fresenius}, we get  $\Pr{Z_1 \geq 1 \, \mid \, Z \geq 1} \geq \frac{1}{2}$. Thus,
\begin{align*}
\E{ Z \, \mid \, Z \geq 1 }
&\geq \Pr{ Z_1 \geq 1 \, \mid \, Z \geq 1} 
\cdot \E{ Z \, \mid \, Z_1 \geq 1 } \\
&\geq \frac{1}{2} \cdot
\sum_{t=0}^{\tmix} \sum_{v \in V} \left(p_{u,v}^t \right)^2 \\ &\geq \frac{1}{2} \cdot \sum_{t=2 \sqrt{n}}^{\tmix} \sum_{v \in \mathcal{L}} (p_{u,v}^t)^2 \geq \sum_{t=2 \sqrt{n}}^{\tmix} \Omega(1/t  ) = c_4 \log_2 n,
\end{align*} 
where $c_4 > 0$ is a constant (the penultimate inequality is due to \autoref{lem:hutmitdreiecken}).

Consequently, for the modified process where collisions are only allowed on nodes in $U$
\begin{align*}
 \Pr{ Z \geq 1} \leq \frac{4c \log_2 n}{c_4 \log_2 n}. 
\end{align*}
Hence by choosing $c = \min\{c_4/8,1 \}$, it follows that $\Pr{ Z \geq 1} \leq \ifrac{1}{2}$.

This implies that with probability at least $\ifrac{1}{2}$, no collision occurs on nodes in $U$ before time $c n \log_2 n$. Furthermore, by \autoref{lem:evente}, with probability at least $1-n^{-1/3}$ there is no collision on nodes in $V \setminus U$. Therefore, with probability at least $\ifrac{1}{2} - n^{-1/3}$, there is no collision among two random walks before time $c n \log_2 n$, and we have shown that $\tmeet = \Omega(n \log n)$.
\end{proof}

\subsection{Star}

Clearly, the coalescing time is $\Theta(\log n)$ which could be easily shown by a direct analysis. For the sake of completeness, we point out that the upper bound also follows from \autoref{lem:beer} and the fact that $\tmeet=O(1)$. The matching lower bound follows from the general bound $\tcoal=\Omega(\log n)$, holding for {\em any graph} (see~\autoref{lem:lognlower}).

\fi

\end{document}